\documentclass[journal]{IEEEtran}

\makeatletter
\def\endthebibliography{%
	\def\@noitemerr{\@latex@warning{Empty `thebibliography' environment}}%
	\endlist
}
\makeatother
\usepackage{xfp}
\usepackage[ruled,vlined,linesnumbered]{algorithm2e}
\SetKwComment{Comment}{/* }{ */}
\SetKwInput{KwInput}{Input}
\SetKwInput{KwOutput}{Output}
\SetKwProg{Fn}{Function}{}{}

\usepackage{algpseudocode}
\usepackage{cuted}
\usepackage{amsmath,amsfonts,amsthm}
\usepackage{mathtools}
\usepackage{multirow}
\usepackage{enumitem}

\usepackage[utf8]{inputenc}
\usepackage[english]{babel}
\addto\captionsenglish{\renewcommand{\figurename}{Fig.}\renewcommand{\tablename}{TABLE}}

\usepackage[export]{adjustbox}

\usepackage{booktabs,tabularx}
\usepackage{nicefrac}
\usepackage{xcolor}
\usepackage{colortbl}
\usepackage{soul}
\usepackage{authblk}

\usepackage{mathabx}

\usepackage{placeins}

\usepackage{comment}

\definecolor{metrictrack}{HTML}{EDF1F2}
\definecolor{metricfill}{HTML}{B7D8CF}
\definecolor{metricfillblue}{HTML}{C6DFEB}
\definecolor{metricfillgreen}{HTML}{C6E2D2}
\newcommand{\metricbarfill}{metricfill}
\definecolor{panelshade}{HTML}{F1F4F5}
\newlength{\metricbarwidth}
\newcommand{\signedchange}[2]{\ifdim\fpeval{#1}pt>0pt+\fi\fpeval{round(#1,#2)}}
\newcommand{\metricbarannotated}[3]{%
    \makebox[\metricbarwidth][l]{%
        \makebox[0pt][l]{\raisebox{-0.16ex}{\color{metrictrack}\rule{\metricbarwidth}{1.5ex}}}%
        \makebox[0pt][l]{\raisebox{-0.16ex}{\color{\metricbarfill}\rule{\fpeval{#1}\metricbarwidth}{1.5ex}}}%
        \makebox[\metricbarwidth][c]{\strut\makebox[10mm][r]{#2}\hspace{0.5mm}{\scriptsize$(#3)$}}%
    }%
}
\newcommand{\comparehigh}[7]{%
    \metricbarannotated{(#1-#3)/(#4-#3)}{#5}
    {\signedchange{100*(#1-#2)/#2}{1}\%}%
}
\newcommand{\comparelow}[7]{%
    \metricbarannotated{(#4-#1)/(#4-#3)}{#5}
    {\signedchange{100*(#1-#2)/#2}{1}\%}%
}

\usepackage{newfloat}
\usepackage{listings}

\newcolumntype{C}{>{\centering\arraybackslash}X}

\newcommand{\para}[1]{{\vspace{4pt} \noindent \textit{#1} \hspace{1pt}}}
\AtBeginDocument{%
	\providecommand\BibTeX{{%
			\normalfont B\kern-0.5em{\scshape i\kern-0.25em b}\kern-0.8em\TeX}}}

\newcommand{\Description}[1]{}

\newtheorem{assumption}{Assumption}
\newtheorem{theorem}{Theorem}
\newtheorem{lemma}[theorem]{Lemma}
\newtheorem{prop}[theorem]{Proposition}

\usepackage{pdftexcmds}
\usepackage{catchfile}
\usepackage{ifluatex}
\usepackage{ifplatform}
\usepackage{color}

\usepackage{makecell}
\DeclareGraphicsExtensions{.eps}

\ifCLASSINFOpdf
\else
\fi
\begin{document}

\title{Cloud-Edge LLM Collaborative Learning via Accelerating Personalized Bilevel Optimization}

\markboth{IEEE Transactions on Pattern Analysis and Machine Intelligence, Sept.~2026}%
{Shell \MakeLowercase{\textit{et al.}}: Bare Demo of IEEEtran.cls for IEEE Journals}
\author{Victor H.~Chen\IEEEauthorrefmark{1}, H. Yu, Stella K. Chung, Z. Li  and H. Yan, \IEEEmembership{Life Fellow,  IEEE}
}

\allowdisplaybreaks[4]

\maketitle


\begin{abstract}
	Amid the rapid advancement of physical-world intelligence, cloud-edge collaborative large language models (LLMs) have emerged as a promising roadmap for practical LLM deployment. However, existing cloud-edge paradigms struggle to balance global consensus with local personalization, which fails to satisfy the need for a unified knowledge foundation on the cloud and domain-specific adaptation at the edge. To address this, we propose a personalized bilevel optimization framework that formalizes cloud-edge LLM collaboration as a dual structure: the upper level optimizes edge-side personalization, while the lower level governs cloud-side knowledge transfer, reaching cloud-edge evolving in coordination. We then propose CIDERS: while the cloud performs knowledge transfer to the learnable backbone, the key lies in embedding global trajectories into each local personalization step via consensus-variate correction to reconcile personalization with consensus. We provide a comprehensive theoretical analysis, including a geometric characterization of the local trajectory and a full convergence guarantee, revealing an explicit trade-off structure between personalization and global convergence. Extensive experiments on Qwen2.5-3B and Qwen2-1.5B show  CIDERS high performance and can significantly improve the accuracy (as much as 200\%+). Control-state diagnostics further  show strong preservation of the local orthogonal component, while a separate cross-domain evaluation finds only a small matched-domain response advantage.
\end{abstract}


\IEEEpeerreviewmaketitle

\section{Introduction}
\IEEEPARstart{L}{arge} language models (LLMs) have evolved into general production tools for their powerful cognitive reasoning capabilities~\cite{openai2023gpt4,moenks2025industryllm}. They are expanding beyond pure information processing tasks toward cyber-physical systems (CPS) that interact with the physical real-world~\cite{driess2023palme,xu2024llmenabledcps}. This demands that the model’s perception, planning, and control capabilities serve the physical entity operation and tightly align with their real-time processes, imposing new challenges on the deployment and execution~\cite{chen2025edgegeneralintelligence}.

Despite their broadening applications, mainstream LLMs are cloud-centric: foundation models are pre-trained on supercomputing clusters~\cite{narayanan2021megatron}, while their inference services are commonly hosted in centralized cloud infrastructure~\cite{friha2024llmbasededge}. Edge devices mostly function as sensing and interaction terminals, uploading prompts and data while receiving inference outputs~\cite{chen2025edgegeneralintelligence}. For clarity, we organize production-grade LLM systems into three stages: pre-training, inference serving, and vertical-domain adaptation. Pre-training relies on multi-granular system-level parallelism: inter-layer partitioning enables pipeline parallelism, intra-layer partitioning enables tensor parallelism, and these are combined with data parallelism to form three-dimensional (3D) parallelism~\cite{narayanan2021megatron}. Subsequent instruction alignment further transforms the models into general-purpose service models following human instructions~\cite{ouyang2022instructgpt}. Inference organizes step-wise autoregressive generation into a concurrent serving system: cached historical key-value (KV) states avoid recomputing KV states for previous tokens~\cite{kwon2023vllm}, and dynamic batching accommodates irregularly arriving requests~\cite{yu2022orca}. Serving metrics such as time-to-first-token, per-token generation latency, and goodput can then be measured and optimized~\cite{zhong2024distserve}. Task adaptation absorbs domain discrepancies through parameter-efficient fine-tuning, with LoRA providing a lightweight approach to adapting pretrained weights~\cite{hu2021lora}. It also incorporates domain-specific information through external knowledge bases and retrieval-augmented generation~\cite{lewis2020rag}, and compresses LLMs into deployable compact forms via knowledge distillation~\cite{gu2024minillm}.

While the cloud-centric paradigm provides powerful computing resources, it can be mismatched with scenarios that impose strict latency and privacy requirements~\cite{chen2025edgegeneralintelligence}. The resulting latency bottleneck can limit support for real-time applications, including autonomous driving decision-making, industrial robotic control, and field edge deployments with limited network backhaul~\cite{xu2024llmenabledcps,ye2024galaxy}. Edge data may also be subject to privacy and data-governance constraints that restrict offloading, making cloud-based LLM APIs unsuitable for fine-tuning on such local data~\cite{wu2026fedllmsurvey}. Moreover, communication and memory overhead are affected by token exchange and KV-cache management~\cite{sun2025disco,kwon2023vllm}. A more salient mismatch concerns personalization: user trajectories, industrial logs, and sensor-text samples for model customization originate at the edge, whereas the modifiable model parameters reside in the cloud. This mismatch motivates cloud-edge collaboration, and we summarize its inference and learning paradigms as follows:

\para{Cloud-edge collaborative inference:} the system first assesses available computing resources and then makes placement decisions~\cite{ye2024galaxy,mudvari2024splitllm} for model partitioning and deployment~\cite{borzunov2023petals,macario2025mdi}. In pipeline-parallel inference, each device computes its local shard and forwards intermediate activations to downstream components, overlapping computation with communication~\cite{hu2022pipeedge}. Nevertheless, coordination across partitioned devices may introduce idle time~\cite{macario2025mdi}. A lightweight model can run in parallel at the edge to generate candidate tokens while the cloud verifies earlier candidates~\cite{park2025specedge,han2026pipesd}. The cloud verifies the proposed tokens, and the edge commits the accepted prefix or rolls back the KV cache when a token is rejected before resuming generation~\cite{zhang2026picospec}.

\para{Cloud-edge collaborative learning:} this paradigm integrates federated learning (FL) with parameter-efficient fine-tuning (PEFT)~\cite{wu2026fedllmsurvey,wen2025fedpeftsurvey}. Specifically, each client introduces lightweight low-rank adaptation modules for the frozen LLM, and performs local learning based on its private data~\cite{wang2024flora,cho2023heterogeneous}. With the client adapters uploaded, the cloud employs stacking-based aggregation or selective sharing to mitigate the noise~\cite{wang2024flora,guo2025selective}. To further alleviate the computation and communication burden, a split FL architecture is introduced: the model is logically partitioned into client-side and server-side sub-models, where clients only need to perform lightweight forward and backward propagation and upload activations to the cloud, while offloading the majority of gradient computation and parameter updates to high-performance cloud servers~\cite{lin2024splitlora,li2026splitcom}.

However, we observe one paradox. Existing cloud-edge collaborations for LLMs are fundamentally designed on model consensus, where it learns the global model parameters~\cite{mcmahan2017fedavg,wu2026fedllmsurvey}. This inherently overlooks one fact that users are naturally personalized in edge scenarios~\cite{tan2022pflsurvey}. Therefore it is necessary to consider the edge personalization: each edge node, while sharing a common global knowledge, can evolve its model based on the local data. It can achieve fast general capabilities while personalized fine-tuning adapts to specific scenarios. To this end, we design CIDERS, which is precisely architected to address this gap: it decomposes each edge model into a globally shared learnable backbone that captures general linguistic and reasoning capabilities via cloud-based task-aware distillation, and a locally adaptive messenger that captures domain-specific personalization via edge-local updates~\cite{xiao2023offsitetuning,wu2024fedbiot}. The two components are orchestrated through our proposed personalized bi-level optimization framework. This design skillfully decouples the learning procedures to achieve what prior cloud-edge paradigms cannot: a systematic continuum from a shared common knowledge base to individually evolved edge models.In summary, our contributions are:

\para{Personalized bilevel optimization formulation:} Based on the aforementioned decomposition of the LLM into the messenger and the backbones, we formulate a personalized bilevel optimization tailored to this messenger--backbone structure. Specifically, the upper-level problem optimizes the messenger for fast client-specific personalization that adapts to heterogeneous local distributions. The lower-level problem learns the   learnable  backbone via task-aware knowledge distillation from the frozen full backbone.  This ensures that the   learnable  backbone faithfully emulates the teacher's foundational capabilities while remaining aligned with the downstream task manifold.

\para{The novel algorithm CIDERS:} We propose  CIDERS that solves the personalized bilevel optimization. Specifically, the edge performs  a  consensus-variate-corrected   geometric  update on two distinct scales: the global trajectory is incorporated into the local counterpart. This steers each client's trajectory toward a tunable balance between the local specialization (slower trajectory) and global consensus (fast trjectory). Then, CIDERS transmits only consensus-variate increments, enabling exact reconstruction of local displacements on the cloud. Finally, with  the global  messenger fixed, it performs task-aware knowledge distillation on the learnable backbone.

\para{Comprehensive theoretical analysis:}  
We first comprehensively demonstrate that CIDERS attains the convergence to a stationary point. Then, our geometric analysis provides an interpretable characterization of the local update trajectory: the effective direction of the messenger lies on the line segment connecting the local personalized gradient and the global consensus direction, a result that reveals how the consensus scalar sets the  personalization–globalization balance. Based on this insight, we propose two diagnostic metrics for  a direct quantitative readout of where any given trajectory sits on the personalization–consensus spectrum. 

\para{Experimental findings:} We conduct comprehensive experiments on Qwen2.5-3B/1.5B that probe the underlying mechanisms and benchmarking performance. It first establishes CIDERS's empirical superiority.  On Qwen2.5-3B, CIDERS improves GSM8K accuracy by 11.15 percentage points over state-of-the-arts and obtains the best Dolly PPL, ROUGE-L, and HumanEval-X scores among the compared compressed methods.  The  consensus-state  diagnostics show a strong client-specific directional component as training loss falls. Collectively, these results support consensus-guided personalization in cloud-edge LLM deployment.

\section{Related Works}
\subsection{Federated Learning and Personalization}

Federated learning (FL) enables distributed clients to train a shared model without centralizing their raw data~\cite{kairouz2021advances}. FedAvg performs local optimization and aggregates client updates by weighted averaging~\cite{mcmahan2017fedavg}. Under heterogeneous data and workloads, local trajectories can drift and produce inconsistent global updates. FedProx constrains local deviation through a proximal term~\cite{li2020fedprox}, while FedDyn introduces dynamic regularization to improve consistency between local and global optimization~\cite{acar2021feddyn}. FedNova normalizes local updates to correct objective inconsistency caused by heterogeneous local computation~\cite{wang2020fednova}. MOON promotes local--global representation consistency through contrastive learning~\cite{li2021moon}. SCAFFOLD maintains server and client control variates whose difference adjusts local update directions to mitigate client drift~\cite{karimireddy2020scaffold}. For personalization, FedPer and FedRep separate shared representations from client-specific layers or prediction heads~\cite{arivazhagan2019fedper,collins2021fedrep}. pFedMe formulates personalization through a Moreau-envelope-based bilevel objective~\cite{dinh2020pfedme}, whereas Ditto jointly learns global and personalized models under a regularization term~\cite{li2021ditto}. Per-FedAvg adapts model-agnostic meta-learning to federated personalization by optimizing an initialization that can be adapted through local gradient steps~\cite{fallah2020perfedavg}, following the MAML principle~\cite{finn2017maml}. These methods provide important foundations for distributed personalization, but their direct application to cloud-edge LLM systems still faces resource constraints~\cite{wu2026fedllmsurvey,wen2025fedpeftsurvey}.

\subsection{Cloud-Edge LLM Collaborative Learning via FL}

Full-parameter LLM fine-tuning is costly, making PEFT a practical basis for federated adaptation~\cite{wu2026fedllmsurvey,wen2025fedpeftsurvey}. Common PEFT methods include adapters, prefix-tuning, and LoRA, which introduce trainable modules, learn continuous prefix vectors, and represent weight updates through low-rank factors, respectively~\cite{houlsby2019adapter,li2021prefix,hu2021lora}. Federated instruction tuning extends this paradigm to distributed private data~\cite{zhang2024federated}. For federated LoRA, FLoRA introduces stacking-based aggregation to address the inconsistency between averaging low-rank factors and aggregating their corresponding full weight updates~\cite{wang2024flora}. FedSA-LoRA selectively shares the A factor while retaining client-specific B factors locally~\cite{guo2025selective}, whereas FFA-LoRA fixes the randomly initialized A factor and updates only B~\cite{sun2024improving}. FedALT further combines client-specific LoRA modules with a shared Rest-of-World (RoW) LoRA through an adaptive mixer~\cite{bian2026fedalt}.Knowledge distillation transfers teacher knowledge to a student and has also been applied to LLM compression~\cite{hinton2015distilling,gu2024minillm}. Offsite-Tuning enables adaptation without full-model access by sending a lightweight adapter and a compressed emulator to the data owner~\cite{xiao2023offsitetuning}. FedBiOT extends this decomposition to FL: the server aligns an emulator with the full model using public data, while clients optimize adapters with the emulator fixed. The server subsequently aggregates client adapters and further updates the emulator through knowledge distillation. FedBiOT formulates this coupled process as bilevel optimization to address the distribution discrepancy between server-side alignment data and client-side private training data~\cite{wu2024fedbiot}.

\subsection{Discussion}

FL, as a distributed learning framework, has been widely applied across deep learning. Cloud-edge LLM collaborative learning is inherently distributed, so FL provides a natural fit. Both PEFT and its combination with model decomposition explore heterogeneous adapter configurations, yet they lack systematic handling of drift-aware updates and are insufficient for edge personalization to meet complex applications, since they meet the paradox that personalization and consensus are difficult to coexist in the local adapter update path.  To address this, CIDERS integrates both into a unified bilevel optimization framework that injects global trajectory into each local personalization step for coupling local personalization and global consensus within the same optimization process. We target LLMs  continuously evolving at the edge while staying coordinated with global consensus in cloud-edge collaboration.

\begin{figure*}
	\centering
	\includegraphics[width=1\textwidth]{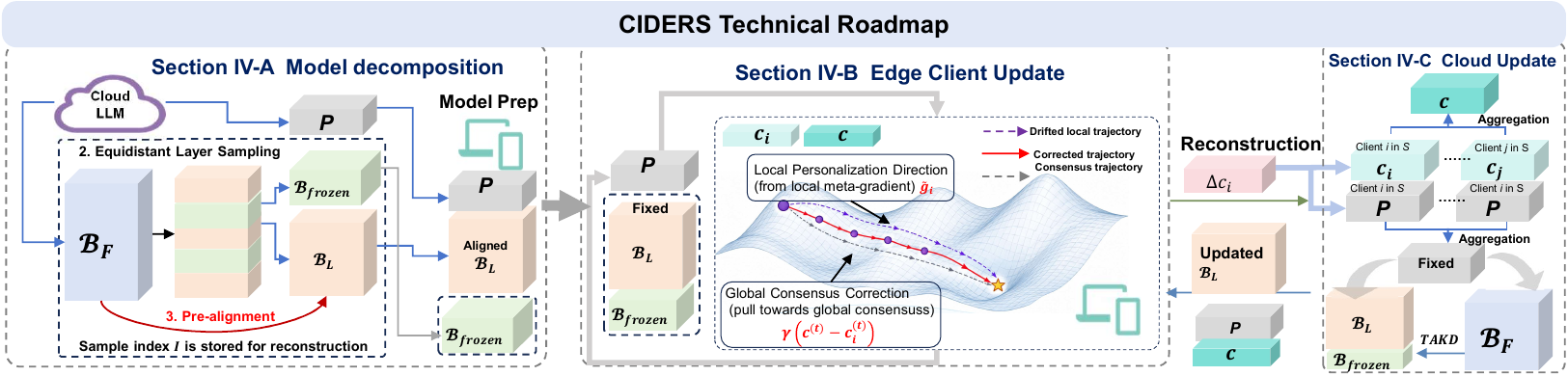}
	\caption{\small\textbf{CIDERS Workflow.}
		(I) The framework decomposes the LLM into a cloud server-side full backbone  $\mathcal{B}_F$, edge client-side messengers  $P_M$ and learnable backbone $\mathcal{B}_L$. (II) Each round the server broadcasts  $\{w_{P_M}, w_{B_L}, c\}$; clients run a ${K}$-step   consensus-variate-corrected  geometric update on private data and upload $\Delta c_i$ only. (III) The server reconstructs $\Delta_i$, aggregates the  global  messenger and  global  consensus  variate , then performs  task-aware  knowledge   distillation  from $\mathcal{B}_F$ to $\mathcal{B}_L$ with $\mathcal{P}_M$ held fixed before the next broadcast.}
	\label{fig:ciders_workflow}
\end{figure*}


\section{Problem Formulation}
We consider a cloud-edge distributive system comprising a cloud server with a public dataset $\mathcal{D}_{\text{public}}$ and $N$ client edges, each possessing a private heterogeneous dataset $\mathcal{D}_i$. Let $\mathcal{M}_i$ denote the full composite model on client $i$, and $\mathcal{M}_*$ denote its server-side counterpart. To balance global coordination and local adaptation, we partition these models into distinct yet structurally interdependent components: $\mathcal{M}_i = \mathcal{P}_{M,i} \circ \mathcal{B}_L$ on the client side, and $\mathcal{M}_*=\mathcal{P}_{M} \circ \mathcal{B}_F$ on the server side. Specifically, we define:
\begin{itemize}[leftmargin=*]
	\item A {full backbone} $\mathcal{B}_F$, parameterized by $w_{\mathcal{B}_F}$, which serves as a static, frozen teacher preserving foundational linguistic capabilities and world knowledge.
	\item A {learnable backbone} $\mathcal{B}_L$, parameterized by $w_{\mathcal{B}_L}$, which acts as the global representation student distilled from $\mathcal{B}_F$ using public data.
	\item The {messengers} $\mathcal{P}_{M,i}$, parameterized by $w_{\mathcal{P}_{M,i}}$ for $i=1,\cdots,N$. These capture domain-specific knowledge from private datasets, while $\mathcal{P}_{M}$ (parameterized by $w_{\mathcal{P}_M}$) denotes their globally aggregated meta-initialization counterpart on the server.
\end{itemize}

\para{Upper-Level: Personalized Meta-Objective.}
We formulate the upper-level problem by adopting the personalization objective as follows
\begin{equation}\label{upperlevel}
	\mathcal{L}^{\text{pers}}(w_{\mathcal{P}_{M}}, w_{\mathcal{B}_{L}}) = \frac{1}{N}\sum\nolimits_{i=1}^N \mathcal{L}_i\!\left(\widetilde{w}_{\mathcal{P}_{M,i}}, \, w_{\mathcal{B}_{L}} \right),
\end{equation}
where $\widetilde{w}_{\mathcal{P}_{M,i}} := w_{\mathcal{P}_{M}} - \eta_{in} \nabla_{w_{\mathcal{P}_{M}}} \mathcal{L}_i(w_{\mathcal{P}_{M}}, w_{\mathcal{B}_{L}})$ denotes the one-step personalized messenger for client $i$, and $\eta_{in} > 0$ is the inner learning rate. Note  (\ref{upperlevel}) optimizes the global messenger $w_{\mathcal{P}_{M}}$  as a meta-model that facilitates rapid client-specific adaptation.

\para{Lower-Level: Task-Aware Knowledge Distillation.}
The lower-level objective aims to transfer server-side knowledge to the learnable backbone. While general KD  anchors the student to the teacher's latent space via intermediate representation matching and output logit alignment, it treats the frozen teacher as an infallible oracle. This risks propagating pre-training flaws or calibration biases without optimizing for downstream utility. To mitigate this, we propose Task-Aware Knowledge Distillation, denoted as task-aware KD below, with loss $\mathcal{L}_{\text{KD}}$ on $\mathcal{D}_{\text{public}}$ and fixed global messenger $w_{\mathcal{P}_{M}}$. This task-aware KD objective is distinct from the teacher-assistant KD method cited in the experimental setup:
\begin{equation}
	\begin{aligned}
		&\mathcal{L}_{\text{KD}}(w_{\mathcal{B}_{L}}; w_{\mathcal{P}_{M}}) = \mathbb{E} \bigg[ \| \mathcal{B}_{L}(x;w_{\mathcal{B}_{L}}) - \mathcal{B}_{F}(x;w_{\mathcal{B}_{F}}) \|_2^2 \\
		&+ \lambda D_{\mathrm{KL}}\!\left( \mathcal{M}_*(x;\{w_{\mathcal{P}_{M}},w_{\mathcal{B}_{F}}\}) \,\|\, \mathcal{M}(x;\{w_{\mathcal{P}_{M}},w_{\mathcal{B}_{L}}\}) \right) \\
		&\quad\quad\quad\quad\quad\quad+ \lambda_{\text{task}} \cdot \ell_{\text{task}}\!\left( \mathcal{M}(x;\{w_{\mathcal{P}_{M}},w_{\mathcal{B}_{L}}\}), y \right) \bigg],
		\label{eq:emulator_obj}
	\end{aligned}
\end{equation}
where $\lambda,\lambda_{\text{task}}>0$. It aims to explicitly injects task-aware supervision grounded in true labels $y$. This ensures  $\mathcal{B}_L$ fully emulates $\mathcal{B}_F$ and aligns with the target task manifold, providing a high-quality foundation for the upper-level meta-adaptation.

  \para{The Personalized Bilevel Optimization (PBO)}
Alternating the optimization of $w_{\mathcal{P}_M}$ and $w_{\mathcal{B}_L}$ via the upper and lower objectives respectively, we formulate the personalized bilevel optimization (PBO) for  the cloud-edge LLM collaborative learning:
\begin{equation}\label{eq:PBO}
	\begin{aligned}
		\min_{w_{\mathcal{P}_{M}}} \quad & \mathcal{L}^{\text{pers}}(w_{\mathcal{P}_{M}}, w_{\mathcal{B}_{L}}) + \frac{\epsilon}{2} \|w_{\mathcal{P}_{M}} - w_{\mathcal{P}_{M}}^{-}\|_2^2\\
		\text{s.t.} \quad & w_{\mathcal{B}_{L}} = \arg\min_{w_{\mathcal{B}_{L}}} \mathcal{L}_{\text{KD}}\!\left( w_{\mathcal{B}_{L}}; w_{\mathcal{P}_{M}} \right).
	\end{aligned}
\end{equation}
In this architecture, we regularize \eqref{upperlevel} by a proximal term  to constrain drift from the previous state $w_{\mathcal{P}_{M}}^{-}$, then   the upper-level optimization learns the global messenger $w_{\mathcal{P}_M}$ to strike a balance between   global consensus  and personalization. Alternately, the lower-level optimization trains $w_{\mathcal{B}_L}$ via task-aware KD with $w_{\mathcal{B}_F}$. This hierarchy explicitly decouples global knowledge alignment from client-specific personalization.

\para{Optimization Challenges.}
While \eqref{eq:PBO} forms the first personalized bilevel optimization, we pose several key insights in challenges: first is obvious, the computation resource is limited at edges, which requires fast global training. Second, since the cloud and the edges are mutually dependent, if the cloud fails the distillation or the edges result in heterogeneity issues, they will alternate their updates  to a continuous deterioration of mutual learning. Third, edges generally demand personalization, which exacerbates the heterogeneity issue and leads to a decrease in the overall training efficiency.

\section{Methodology}
To enable thorough learning under computational efficiency, CIDERS coordinates two distinct roles, i.e., client-specific messenger adaptation and global teacher-guided backbone alignment. It considers only messenger updates at edges. Then it meticulously navigates the personalization with the global trajectory for fast adapting meta-knowledge encoded within the messenger. On the cloud, the client messengers are aggregated and fixed, then $\mathcal{B}_L$ learns foundational linguistic representations encoded within the full backbone. This architectural decomposition and orchestration provide a complete coordination pathway for heterogeneous edge LLMs reaching personalization with the fast global convergence. We summarize the cloud-edge LLM collaborative learning procedure in Algorithm \ref{alg:ciders_workflow}.
\subsection{Compressed Model Preparation}
For the pre-trained LLM $\mathcal{M}_*$ with $n$ transformer layers, we decompose $\mathcal{M}_*$ into distinct functional modules ($\mathcal{B}_F$, $\mathcal{B}_L$, $\mathcal{P}_{M}$). This structural decomposition is grounded in the hierarchical representation learning of LLMs: lower layers encode shared  domain-agnostic linguistic priors, while upper layers capture heterogeneous higher-order semantic abstractions. Consequently, the topmost layers are chosen as the messenger $\mathcal{P}_M$ for rapid local adaptation, while the frozen full backbone $\mathcal{B}_F$ preserves foundational capabilities. The learnable backbone $\mathcal{B}_L$ is a uniform strided subsample of $\mathcal{B}_F$, yielding a compact global surrogate. Since directly launching cloud-edge collaborative learning  on these disjoint modules may flaw due to the heterogeneity of $\mathcal{P}_M$'s massive representation mismatches,  pre-alignment via task-aware KD in \eqref{eq:emulator_obj} for initializing $\mathcal{B}_L$ should be implemented before client updates so that the local updates start from a teacher-informed compressed path.

%

\begin{itemize}[leftmargin=*]
	\item  \textit{Step 1: $\mathcal{P}_{M}$ Identification.} The messenger $\mathcal{P}_{M}$, parameterized by $w_{\mathcal{P}_M}$, comprises the topmost $a$ layers of $\mathcal{M}_*$. It acts as the client-side adapter to capture domain-specific knowledge. The remaining model constitutes the frozen full backbone $\mathcal{B}_F$ ($w_{\mathcal{B}_F}$).
\item \textit{Step 2: Learnable Backbone Construction.}  To preserve the teacher’s depth-wise coverage on clients, we compress $\mathcal{B}_F$ by uniformly extracting its $n_E = \lfloor \beta \cdot (n - a) \rfloor$ to construct $\mathcal{B}_L$ ($w_{\mathcal{B}_L}$), where $\beta \in (0,1]$ is the retained fraction of the non-messenger backbone.
\item \textit{Step 3: Pre-alignment via task-aware KD.}  Prior to collaborative learning, $\mathcal{B}_L$ undergoes task-aware KD on  $\mathcal{D}_{\text{public}}$ for learning $\mathcal{B}_F$'s linguistic priors.
\end{itemize}

\subsection{Client Update Procedure}
The client update is meticulously designed to solve the upper-level personalized meta-objective in (\ref{upperlevel}) under the complex environment of the real-world. Our strategy targets incorporating the global learning tragectory into the meta-gradient for personalization adaptation under  data heterogeneity.

  \para{Consensus-variate-corrected geometric update.}
At  round $  t  $, each participated client $  i \in \mathcal{S}_t  $ receives the global state triplet {\tiny$  \{w_{\mathcal{P}_M}^{(t)}, w_{\mathcal{B}_L}^{(t)}, {c}^{(t)}\}  $}, then it  initializes its local messenger  at the current global messenger, i.e., {\tiny$  w_{\mathcal{P}_{M,i,0}}^{(t)} \leftarrow w_{\mathcal{P}_M}^{(t)}  $}, and retrieves its historical  local  consensus variate {\small$  {c}_i^{(t)}  $}.  To prevent overfitting to local noise while navigating the heterogeneous landscape, the messenger is updated via a consensus-variate-corrected geometric interpolation over $K$ steps:
\begin{equation}\label{eq:local_update}
			\begin{aligned}
				&\quad\quad w_{\mathcal{P}_{M,i,k+1}}^{(t)} = w_{\mathcal{P}_{M,i,k}}^{(t)} - \eta g^{(t)}_{traj},\,\text{ where }\\
				&g^{(t)}_{traj}= \widetilde{{g}}_{i,k} + \gamma \big({c}^{(t)} - {c}_i^{(t)}\big)  +\epsilon \big( w_{\mathcal{P}_{M,i,k}}^{(t)} - w_{\mathcal{P}_M}^{(t)} \big),
			\end{aligned}
\end{equation}
$\widetilde{{g}}_{i,k}$ is the personalized meta-gradient computed based on a mini-batch $\xi_{i,k} \sim \mathcal{D}_i$ such that
\begin{equation}
	\widetilde{{g}}_{i,k} = \nabla_{w_{\mathcal{P}_M}} \mathcal{L}_i\!\left( \widetilde{w}_{\mathcal{P}_{M,i,k}}^{(t)}, \, w_{\mathcal{B}_L}^{(t)}; \, \xi_{i,k} \right),
\end{equation}
with $\widetilde{w}_{\mathcal{P}_{M,i,k}}^{(t)}$ being the one-step look-ahead as follows:
\begin{equation}
	\widetilde{w}_{\mathcal{P}_{M,i,k}}^{(t)} = w_{\mathcal{P}_{M,i,k}}^{(t)} - \eta_{\text{inner}} \nabla_{w_{\mathcal{P}_M}} \mathcal{L}_i\!\left(w_{\mathcal{P}_{M,i,k}}^{(t)}, w_{\mathcal{B}_L}^{(t)}; \xi_{i,k}\right).
\end{equation}
In practical   implementations , (\ref{eq:local_update}) in fact uses  \texttt{SGD} in the optimizer tool, i.e., $w_{\mathcal{P}_{M,i,k+1}}^{(t)}\leftarrow \texttt{OPTIM}(\texttt{SGD},w_{\mathcal{P}_{M,i,k}}^{(t)}, g^{(t)}_{traj},\eta)$, and other optimizer such as \texttt{AdamW} works also well. To circumvent the heavy computation of $\widetilde{{g}}_{i,k}$, CIDERS employs two   Hessian-free  meta-optimization strategies . The first is   the first-order (FO)  strategy, which intentionally drops the second-order derivative, assuming the Hessian impact is locally negligible: $v_i:=\widetilde{{g}}_{i,k}^{\rm FO} =  \nabla_{w_{\mathcal{P}_M}} \mathcal{L}_i( w_{\mathcal{P}_M} - \eta_{\text{inner}} \nabla_{w_{\mathcal{P}_M}} \mathcal{L}_i )$. This requires only standard forward-backward passes, serving as the lightweight default.  The second is   the finite-difference (FD) strategy , which adopts  a  central-difference Hessian--vector  product (HVP)  estimate
\begin{equation}
	\widehat{H_i v_i}=
	\nicefrac{[\nabla \mathcal{L}_i(w+\varepsilon_{\mathrm{FD}}v_i)
		-\nabla \mathcal{L}_i(w-\varepsilon_{\mathrm{FD}}v_i)]}
	{2\varepsilon_{\mathrm{FD}}},
	\label{eq:fd_hvp}
\end{equation}
and the curvature-aware surrogate can be obtained $g_i^{\mathrm{FD}}=v_i-\lambda_{\mathrm{HVP}}\widehat{H_i v_i}$, and here $\lambda_{\mathrm{HVP}}=\eta_{\mathrm{in}}$. We further provide an insightful analysis of the   consensus-variate-corrected  geometric update in \eqref{eq:local_update} via the following three aspects. First, the term $\gamma ({c}^{(t)} - {c}_i^{(t)})  $ estimates   the discrepancy between the local meta-gradient and the global consensus direction, providing  consensus-informed   global--personal  coordination. Second, the   consensus-correction coefficient  $  \gamma \in [0,1]  $  parameterizes a continuous   balance  between global consensus and local specialization: when $  \gamma = 1  $, the correction  aligns the client   update with the global consensus direction, whereas  when $  \gamma \to 0  $, the update reduces to pure local meta-gradient descent. Third, the proximal term {\small$  -\epsilon \eta (w_{\mathcal{P}_{M,i,k}}^{(t)} - w_{\mathcal{P}_M}^{(t)})  $} bounds divergence from the current global state, improving stability  in  heterogeneous regimes. After completing $K$ local steps, the client's average update trajectory can be recorded as:
\begin{equation}
	\Delta_{i} := \nicefrac{1}{K\eta} \Bigl( w_{\mathcal{P}_M}^{(t)} - w_{\mathcal{P}_{M,i,K}}^{(t)} \Bigr),
\end{equation}
which  we name  as  the normalized local displacement. It can be seen $\Delta_i$ can be used to update both the  global messenger and the global consensus variate on the cloud in \eqref{eq:server_aggregation}. The local consensus variate is designed to estimate the local meta-gradient with the $K$ averaging of $\widetilde{{g}}_{i,k}$ and is subsequently refreshed using an exponential moving average (EMA) as follows:
\begin{equation}
	{c}_i^{(t+1)} = \bigl[1 + \alpha(\gamma - 1)\bigr]{c}_i^{(t)} + \alpha \Delta_{i} - \alpha \gamma {c}^{(t)}.
\end{equation}
Then client $i$ transmits the   consensus-variate increment
 \begin{equation}
	\Delta{c}_i = {c}_i^{(t+1)} - {c}_i^{(t)}
\end{equation}
to the cloud. It can be seen that with $\Delta{c}_i$, the cloud can recover $\Delta_{i}$ by storing $c_i$ for updating the global messenger.  Consequently, the upstream communication payload is reduced to a single low-dimensional   consensus-variate  increment.

\begin{algorithm}[!ht]
	\caption{Algorithmic Framework of CIDERS}
	\label{alg:ciders_workflow}
	\KwIn{Original LLM $\mathcal{M}_*$, messenger length $a$, compression rate $\beta$,
		global and local consensus variates $c^{(0)},\{c_i^{(0)}\}$,
		parameters $(\alpha,\gamma,\tau,\eta,K,T)$, public data $\mathcal{D}_{\rm pub}$.}
	\KwOut{Personalized edge models $\{\mathcal{M}_i\}$.}
	
	\tcp{MODELPREP}
	Obtain layer number $n \gets |\mathcal{M}_*|$\;
	Compute $\mathcal{B}_L$ layers $n' \gets \lfloor \beta \cdot (n - a) \rfloor$\;
	$\mathcal{B}_F \gets \{\mathcal{M}_{*,i}\}_{i=0}^{n-a-1},\quad
	\mathcal{P}_M \gets \{\mathcal{M}_{*,i}\}_{i=n-a}^{n-1}$\;
	$\mathit{Indices} \gets \bigl\{\lfloor j \cdot \nicefrac{(n-a-1)}{(n'-1)} \rfloor \mid j=0,\dots,n'-1\bigr\}$\;
	$\mathcal{B}_L \gets \{(\mathcal{B}_F)_k \mid k \in \mathit{Indices}\}$\;
	Initialize $w_{\mathcal{P}_M}$, $w_{\mathcal{B}_F}$, $w_{\mathcal{B}_L}$\;
	Train $w_{\mathcal{B}_L} \gets \arg\min_{w_{\mathcal{B}_L}} \mathcal{L}_{\rm KD}(w_{\mathcal{B}_L}; w_{\mathcal{P}_M}, \mathcal{D}_{\rm pub})$\;
	
	Cloud transmits $\{w_{\mathcal{P}_M}^{(0)}, w_{\mathcal{B}_L}^{(0)}, c^{(0)}\}$ to $\mathcal{S}_0$.
	
	\For{$t=0,\ldots,T-1$}{
		\tcp{Server: lower-level training}
		Recover exact updates $\{\Delta_{i}\}_{i\in\mathcal{S}_t}$ from $\{\Delta c_i\}_{i\in\mathcal{S}_t}$.
		
		Aggregate consensus and messenger:
		$c^{(t+1)} \gets (1-\alpha)c^{(t)} + \alpha \sum_{i\in\mathcal{S}_t} p_i \Delta_{i}$,
		$w_{\mathcal{P}_M}^{(t+1)} \gets w_{\mathcal{P}_M}^{(t)} - \tau \sum_{i\in\mathcal{S}_t} p_i \Delta_{i}$.
		
		Perform 
		$w_{\mathcal{B}_L}^{(t+1)} \gets \arg\min_{w_{\mathcal{B}_L}} \mathcal{L}_{\rm KD}$ on $\mathcal{D}_{\rm pub}$.
		
		Transmit $\{w_{\mathcal{P}_M}^{(t+1)}, w_{\mathcal{B}_L}^{(t+1)}, c^{(t+1)}\}$ to $\mathcal{S}_{t+1}$.
		
		\tcp{Clients: upper-level training}
		\ForEach{$i\in\mathcal{S}_t$ in parallel}{
			Obtain $(w_{\mathcal{P}_{M,i,K}}^{(t)}, \Delta_{i}, c_i^{(t+1)}) \leftarrow \texttt{LocalUpdate}(w_{\mathcal{P}_M}^{(t)}, c^{(t)}, c_i^{(t)}, K, \eta, \mathcal{D}_i)$\;
			
			Send $\Delta c_i \gets c_i^{(t+1)} - c_i^{(t)}$ to the cloud.
		}
	}
	
	\tcp{Edge inference}
	$\mathcal{M}_i \gets \text{Compose}\bigl(w_{\mathcal{P}_{M,i}}^{(T)}, w_{\mathcal{B}_L}^{(T)}\bigr)$.
	
	\Fn{\texttt{LocalUpdate}$(w_{\mathcal{P}_M}, c, c_i, K, \eta, \mathcal{D})$}{
		\For{$k=1,\ldots,K$}{
			Perform update on a minibatch of $\mathcal{D}$ via \eqref{eq:local_update}\;
		}
		$\Delta \gets \frac{1}{K\eta}\bigl(w_{\mathcal{P}_M} - w_{\mathcal{P}_{M,K}}\bigr)$\;
		$c_i^+ \gets \bigl[1 + \alpha(\gamma - 1)\bigr] c_i + \alpha \Delta - \alpha \gamma c$\;
		\KwRet $w_{\mathcal{P}_{M,K}},\, \Delta,\, c_i^+$\;
	}
\end{algorithm}

\subsection{Server Update Procedure}
As the centralized orchestrator, the server advances the global meta-initialization via consensus-informed aggregation and   refines the learnable  backbone via task-aware  knowledge  distillation. With our  design, CIDERS can update the client messengers and the   global consensus variate by  recovering $\Delta_{i}$ from $\Delta c_i$.

\para{Key information reconstruction.}
For each participated  client $i \in \mathcal{S}_t$, it receives the   consensus-variate  increment $\Delta {c}_i = {c}_i^{(t+1)} - {c}_i^{(t)}$. Then the server exactly recovers the averaged local messenger update trajectory $\Delta_{i}$ using its synchronized historical state ${c}_i^{(t)}$:
\begin{equation}
	\Delta_{i} = \nicefrac{1}{\alpha} \Delta {c}_i - (\gamma - 1){c}_i^{(t)} + \gamma {c}^{(t)}.
\end{equation}
The   server-side copies of the local consensus variates  are subsequently updated via ${c}_i^{(t+1)} \leftarrow {c}_i^{(t)} + \Delta {c}_i$. This strategic recovery avoids the upstream transmission of raw messenger weights, thereby reducing the bandwidth costs and shielding the local optimization trajectories from direct privacy exposure.

\para{Global aggregation.}
Given the reconstructed  $\{\Delta_{i}\}_{i \in \mathcal{S}_t}$, the server simultaneously advances the aggregation of the global messenger and the global consensus variate as follows:
\begin{equation}
	\begin{aligned}
		w_{\mathcal{P}_M}^{(t+1)} &= w_{\mathcal{P}_M}^{(t)} - \tau K \eta \sum\nolimits_{i \in \mathcal{S}_t} p_i \Delta_{i}, \\
		{c}^{(t+1)} &= (1 - \alpha){c}^{(t)} + \alpha \sum\nolimits_{i \in \mathcal{S}_t} p_i \Delta_{i},
	\end{aligned}
	\label{eq:server_aggregation}
\end{equation}
where $p_i = |\mathcal{D}_i| / \sum_{j \in \mathcal{S}_t} |\mathcal{D}_j|$ (here $p_i=1/N$ for simplicity), and $\tau>0$. Geometrically, the global messenger update acts as a convex combination of the current meta-initialization and the local adaptations, with the learning rate $\tau$ dampening aggregation noise inherent to partial client participation. Concurrently, ${c}^{(t+1)}$   is the global consensus variate . It tracks the moving average of the personalized gradients to smooth out round-to-round variance and provide an unbiased reference direction for subsequent client adaptations.

\para{Learning ${\mathcal{B}_L}$ via task-aware KD.}
Finally,  the lower-level objective in \eqref{eq:PBO} learns ${\mathcal{B}_L}$ from  $\mathcal{B}_F$. By fixing $\mathcal{P}_M$, this procedure also plays the role in coordination and adaptation, ensuring the alignment with the current meta-adaptation context. Specifically, we have:
\begin{equation}
	w_{\mathcal{B}_L}^{(t+1)} \leftarrow \arg\min_{w_{\mathcal{B}_L}} \mathcal{L}_{\text{KD}}\Big(w_{\mathcal{B}_L} ; w_{\mathcal{P}_M}^{(t+1)}\Big).
	\label{eq:server_takd_update}
\end{equation}
While this guarantees that ${\mathcal{B}_L}$ absorbs foundational linguistic features, it establishes a robust, task-aligned representation foundation for the next round of client update.

\section{Theoretical Analysis}
We have proposed CIDERS, and the key mechanisms are twofold: we first divide the LLM into functional modules and adopts different optimization strategies for each of them, then we incoporates the global tracjectory in  the personalized bilevel optimization framework \eqref{eq:PBO}. This will bring complex principles. Luckily, our proposed framework is theoretically robust.

\subsection{Convergence Analysis}
In this section, we analyze the proposed personalized bilevel structure solved by CIDERS, which we  show achieves a sublinear convergence rate to a neighborhood of the stationary point. For simplicity, we denote the PBO as:
\begin{equation}
	\min _v \Phi(v)+\frac{\varepsilon}{2}\left\|v-v^{-}\right\|^2,  \text { s.t. }  w_S=\arg \min _{w_S} \mathcal{L}_{\mathrm{KD}}\left(w_S ; v\right),
\end{equation}
where $v$ denotes the global messenger, $w_S$ the learnable  backbone, $\Phi(v)=$ $\nicefrac{1}{N} \sum_i \mathcal{L}_i\left(\widetilde{v}_i, w_S^*(v)\right)$ is the upper-level meta-objective with $\widetilde{v}_i=v-\eta_{\text {inner }} \nabla_v \mathcal{L}_i\left(v, w_S^*(v)\right)$, and $w_S^*(v)$ is the exact minimizer of the lower-level task-aware knowledge distillation loss $\mathcal{L}_{\mathrm{KD}}(\cdot ; v)$ on the public dataset.  Since CIDERS maintains an inexact backbone $w_S^{(t)}$ together with   consensus-variate-corrected  local messenger steps, the quantity of interest is the joint gap
\begin{equation}
	J_t
	:=\mathbb{E}\bigl[\Phi(v^{(t)})-\Phi^*\bigr]
	+\mathbb{E}\bigl\|w_S^{(t)}-w_S^*(v^{(t)})\bigr\|^2.
\end{equation}
Moreover, to track the coupled evolution of all error sources, we construct a Lyapunov function to couple all error sources, i.e., the virtual messenger error, client drift $E_t$, consensus lags $C_t,\widetilde{C}_t$, and the joint gap $J_t$:
\begin{equation}
	\begin{aligned}
		V^{(t)}=&\left(1-4 a^2\right) \mathbb{E}\left\|z^{(t)}-v^*\right\|^2+\omega_1 E_t\\
		&\qquad\qquad+Q_{w 2} \eta C_t+Q_{w 3} \eta \widetilde{C}_t+Q_{w 4} \eta J_t,
	\end{aligned}
\end{equation}
where we have defined $z^{(t)}=v^{(t)}+\nicefrac{(1-\tau)}{\tau}\cdot (v^{(t)}-v^{(t-1)}),$ and $a, \omega_1, Q_{w 2}, Q_{w 3}, Q_{w 4}$ are positive. Then, the convergence of CIDERS  solving PBO (\ref{eq:PBO}) is:

\begin{theorem}[Ergodic convergence of CIDERS]
	\label{thm:ciders-main}
	Under regular assumptions and moderate conditions, let the step size satisfies $\eta=\mathcal{O}({T}^{\nicefrac{-1}{2}})$, then	the ergodic joint gap satisfies where $V_0^{(0)}=(1-4 a^2) \mathbb{E}\|z^{(0)}-v^*\|^2+\omega_1 E_0$ and $C_{\eta^2}>0$ is a constant from the expansion of $D_{\text {tot }}^{\prime}$. Then the ergodic average of the joint gap satisfies
	{\small\begin{equation}
		\frac{1}{T} \sum_{t=0}^{T-1} J_t \leq \frac{C_\eta}{\delta_{\Phi}}+\frac{2}{\delta_{\Phi}} \sqrt{\frac{V_0^{(0)} C_{\eta^2}}{T}}+\frac{V_1^{(0)}}{\delta_{\Phi} T}+\frac{C_{\eta^3} V_0^{(0)}}{\delta_{\Phi} C_{\eta^2} T}+\mathcal{O}\left(T^{-3 / 2}\right),
	\end{equation}}
	where $C_\eta, C_{\eta^2}, C_{\eta^3}$ are positive constants, $V_0^{(0)}=(1-4 a^2) \mathbb{E}\|z^{(0)}-v^*\|^2+\omega_1 E_0$ collects the intial client drift, and $V_1^{(0)}=Q_{w 2} C_0+Q_{w 3} \widetilde{C}_0+$ $Q_{w 4} J_0$ collects the initial consensus lags and the joint gap.
\end{theorem}
The detailed proof and discussions on the convergence result is moved to Appendix.

\subsection{Geometry of the Trajectory}
\label{sec:theoretical_tradeoff}

Consider a fixed communication round $t$, the client $i$ obtains the global messenger $\bar{v}:=w_{\mathcal{P}_M}^{(t)}$ from the server and performs multiple local updates, where the trajectory is simultaneously determined by three forces: the local meta-gradient provides a personalized direction, the   consensus-variate  correction pulls the local messenger $x_k:=w_{\mathcal{P}_{M, i,k}}^{(t)}$ towards the global consensus, and the proximal constraint resists the messenger's deviation from the starting point, which constitutes a tightly coupled nonlinear dynamical process, making it extremely difficult for an analysis.  Here, we develop an approximate geometric characterization that isolates the  core factors. Specifically,  let us define $\rho:=1-\eta \varepsilon$ and $S_K:=\sum_{j=0}^{K-1} \rho^{K-1-j}$, we can unroll  \eqref{eq:local_update}
\begin{equation}\label{eq:m}
	\bar{v}-x_K=\eta \sum\nolimits_{j=0}^{K-1} \rho^{K-1-j} g_{i, j}+\eta S_K \cdot \gamma\left({c}^{(t)}-{c}_i^{(t)}\right) .
\end{equation}
Then, we  define the along-path average meta-gradient $\widebar{g}_i^{(t)}:=\nicefrac{1}{S_K} \sum_{j=0}^{K-1} \rho^{K-1-j} g_{i, j}$. Then (\ref{eq:m}) can lead to
\begin{equation}\label{eq:delta}
	\Delta_{i}=\sigma_K\left(\bar{g}_i^{(t)}+\gamma\left({c}^{(t)}-{c}_i^{(t)}\right)\right),
\end{equation}
where $\sigma_K:=S_K / K \in(0,1]$. It can be seen the   correction shapes the direction and the  proximity  modulates the magnitude.

\subsubsection{{Convex Combination}} Let $\widetilde{g}_i^{(t)}:=\nabla \mathcal{L}_i(\bar{v})$ denote the local meta-gradient at the broadcast point and $\widetilde{G}^{(t)}:=\sum_j p_j \widetilde{g}_j^{(t)}$ the global average gradient. Define the heterogeneity gap $\delta_i^{(t)}:=\widetilde{g}_i^{(t)}-\widetilde{G}^{(t)}$. To render the directional dynamics analytically tractable and obtain a closed-form characterization of the   personalization--consensus  trade-off, we introduce the idealizing approximations:

\begin{itemize}
	\item[A1.] \textit{Short local step}: Consider sufficiently small $\eta$, which means short local step  so that $\widebar{g}_i^{(t)} \approx \tilde{g}_i^{(t)}$. The path error $\phi_i^{(t)}:=$ $\widebar{g}_i^{(t)}-\tilde{g}_i^{(t)}$ is bounded by $O(\eta )$.
	\item[A2.] \textit{Fixed heterogeneity}: The heterogeneity gap is constant across rounds when the messenger moves slowly, i.e., $\delta_i^{(t+1)}=\delta_i^{(t)}=\delta_i$.
	\item[A3.] \textit{No proximal damping}: We set $\varepsilon=0$, hence $\sigma_K=1$. The case $\sigma_K \neq 1$ is treated separately and shown to preserve span invariance of the convex combination.
	\item[A4.] \textit{Cold start}: Initializes  ${c}^{(0)}={c}_i^{(0)}=0$, thus $r_i^{(0)}=\delta_i$.
	\item[A5.] \textit{Full participation}: All clients participate in every round.
\end{itemize}

Under (A1-A5), the normalized local displacement is simplified to $\Delta_{i}\approx\widetilde{g}_i^{(t)}+\gamma({c}^{(t)}-{c}_i^{(t)})$, We now define the reference field, which is the effective total direction as follows
\begin{equation}
	u_i^{(t)}:=\widetilde{g}_i^{(t)}+\gamma({c}^{(t)}-{c}_i^{(t)}).
\end{equation}
To track how this reference field relates to the global consensus direction $\widetilde{G}^{(t)}$, we introduce the residual
\begin{equation}
	r_i^{(t)}:=u_i^{(t)}-\widetilde{G}^{(t)}=\delta_i^{(t)}+\gamma({c}^{(t)}-{c}_i^{(t)}).
\end{equation}
Under the idealizing assumptions, the consensus-gap recursion yields a closed-form recurrence for $r_i^{(t)}$:
\begin{equation}
	r_i^{(t+1)}=(1-\alpha) r_i^{(t)}+\alpha(1-\gamma) \delta_i \,\,\text{where}\, r_i^{(0)}=\delta_i.
\end{equation}
In fact, the residual remains strictly proportional to $\delta_i$:
\begin{equation}\label{eq:7}
	r_i^{(t)}=\psi_t \delta_i \,\text{where}\,\, \psi_t:=(1-\gamma)+\gamma(1-\alpha)^t.
\end{equation}
Substituting (\ref{eq:7}) back into the definition of $u_i^{(t)}$ yields:
\begin{prop}[Convex combination of local and global directions]\label{prop:1}
	Under approximations (A1)-(A5), the reference direction $u_i^{(t)}$ at client $i$ in round $t$ lies exactly on the line segment joining the local starting gradient $\widetilde{g}_i$ and the global average gradient $\widetilde{G}$ :
	\begin{equation}\label{equi}
		u_i^{(t)}=\psi_t \widetilde{g}_i+\left(1-\psi_t\right) \widetilde{G}, \quad \psi_t=(1-\gamma)+\gamma(1-\alpha)^t.
	\end{equation}
	Moreover, the displacement satisfies $\Delta_{i}\approx u_i^{(t)}$.
\end{prop}
Proposition \ref{prop:1} provides a  geometric characterization of the local messenger trajectory. The  global consensus weight $\left(1-\psi_t\right)=\gamma\left(1-(1-\alpha)^t\right)$ increases monotonically from 0 at $t=0$ to $\gamma$ as $t \rightarrow \infty$, with a half-life of approximately $\Theta(\log 2 /|\log (1-\alpha)|)$ rounds. The scalar $\psi_t$ thus interpolates between two regimes: at cold start $\left(\psi_0=1\right)$, the direction is purely local; at lock-in $\psi_{\infty}=1-\gamma$, the direction stabilizes at the convex combination specified by $\gamma$. The parameter $\alpha$ controls the speed of transition, while the pair $\left({c}, {c}_i\right)$ implements this transition through an exponential moving average that smoothly drives the translation from 0 to $-\gamma \delta_i$.

\subsubsection{Quantifying Deviations.}
We quantify the direction deviations from the idea segment under the idealizing assumptions.  Our first key insight is relaxing the short-inner-loop assumption (A1) introduces a path-averaging error. Note the alongpath  error $\phi_i^{(t)}:=\bar{g}_i^{(t)}-\tilde{g}_i^{(t)}$  is bounded by $\|\phi_i^{(t)}\| \leq L C_K \eta G_{\mathrm{bd}}$. Applying a variation of constants argument to the residual recurrence gives the uniform bound
\begin{equation}
	\|u_i^{(t)}-u_i^{(t)}|_{(\mathrm{E} 0)}\| \leq 2 \gamma R\left(1-(1-\alpha)^t\right) \leq 2 \gamma R,
\end{equation}
where $R:=L C_K \eta G_{\mathrm{bd}}$ and  $u_i^{(t)}|_{(\mathrm{E} 0)}$ denotes the ideal direction from Proposition \ref{prop:1}. Thus, the true direction lies within an $O\left(\gamma L K \eta G_{\mathrm{bd}}\right)$ neighborhood of the ideal segment. This deviation vanishes as $K \eta \rightarrow 0$ and scales linearly with $\gamma$, reflecting the fact that stronger consensus alignment amplifies the sensitivity of the consensus variates to gradient estimation error. Second, even when the consensus   variates are ideally aligned , i.e., ${c}\approx\widetilde{G}$ and ${c}_i\approx\widetilde{g}_i$, the executed field at an arbitrary point $x$ along the inner path contains a Hessian remainder. Specifically, the Taylor expansion around $\bar{v}$ gives
\begin{equation}
	\begin{aligned}
		\nabla \mathcal{L}_i(x)+\gamma&\left({c}-{c}_i\right)\approx(1-\gamma) \nabla \mathcal{L}_i(x)\\
		&+\gamma \nabla \overline{\mathcal{L}}(x)+\gamma\left(H_i-\widebar{H}\right)(x-\bar{v}),
	\end{aligned}
\end{equation}
where $H_i:=\nabla^2 \mathcal{L}_i(\bar{v})$ and $\bar{H}:=\sum_j p_j H_j$. The remainder is bounded by $\gamma \chi\|x-\bar{v}\|$, with $\chi:=$ $\max _i\left\|H_i-\bar{H}\right\|$ quantifying the heterogeneity of meta-Hessians across clients. This implies that when $\gamma$ is large to extent, long local step may accumulate Hessian mismatch rather than improve consensus.

\subsubsection{The Case $\sigma_K \neq 1$}\label{sigma}
When $\varepsilon>0$, it satisfies  $\sigma_K=S_K / K<1$. Under (A1), (A2), and (A4), the recurrence generalizes to
\begin{equation}\label{eq:r_i}
	r_i^{(t+1)}=\left[1-\alpha+\alpha \gamma\left(1-\sigma_K\right)\right] r_i^{(t)}+\alpha(1-\gamma) \delta_i.
\end{equation}
It can be seen the span-invariance property $r_i^{(t)} \in \operatorname{span}\left\{\delta_i\right\}$ is preserved, i.e., $r_i^{(t)}=\psi_t \delta_i$,  \textcolor{black}{ but $\psi_t$ is modified} as,
\begin{equation}
	\psi_{t+1}=\left[1-\alpha+\alpha \gamma\left(1-\sigma_K\right)\right] \psi_t+\alpha(1-\gamma),  \psi_0=1,
\end{equation}
whose steady-state solution can be derived as follows
\begin{equation}
	\psi_{\infty}=\frac{1-\gamma}{1-\gamma\left(1-\sigma_K\right)}.
\end{equation}
Note $\sigma_K<1$, then $\psi_{\infty}>1-\gamma$.  Since $u_i^{(\infty)}=\widetilde{G}+\psi_{\infty} \delta_i=\psi_{\infty} \widetilde{g}_i+(1-\psi_{\infty}) \widetilde{G}$, a larger $\psi_{\infty}$ means the reference field is closer to $\widetilde{g}_i$ than the undamped counterpart  \(1-\gamma\) in \eqref{equi}. The mechanism can be therefore concluded that the compression of the displacement \(\sigma_K<1\) attenuates \(\Delta_{i}\), so the consensus  variates  underestimate the gap \(\delta_i\) and  \(\gamma({c}-{c}_i)\) is weaker than in the undamped case, resulting in the proximal damping driving the EMA to a more personalized field.

\renewcommand{\figurename}{Fig.}
\renewcommand{\tablename}{Table}
\makeatletter
\def\CIDERS@tabletype{table}
\long\def\@makecaption#1#2{%
    \normalfont\fontsize{9}{10}\selectfont
    \setlength{\parindent}{0pt}%
    \ifx\@captype\CIDERS@tabletype
        \centering #1: #2\par
        \vskip6pt\relax
    \else
        \vskip8pt\relax
        \noindent #1: #2\par
    \fi
}
\makeatother
\begin{table*}[!htbp]
    \caption{Datasets, client partitions, and evaluation scope across Protocols A--C}
    \label{tab:protocol_dataset_audit}
    \fontsize{9}{10}\selectfont
    \renewcommand{\arraystretch}{0.95}
    \setlength{\tabcolsep}{2.0pt}
    \centering
        \begin{tabular}{l|l|r|c|l|l|l}
            \toprule
            \textbf{Protocol} &
            \textbf{Dataset} &
            \makecell[r]{\textbf{Samples}} &
            \textbf{Clients} &
            \textbf{Partition} &
            \textbf{Evaluation basis} &
            \textbf{Measurements} \\
            \midrule
             & GSM8K~\cite{cobbe2021gsm8k} & 7,458 &  &
            & \makecell[l]{GSM8K: 1,319 questions}
            & Exact-match accuracy \\
            A & Dolly-15K~\cite{conover2023dolly} & 14,980 & 5 & IID
            & \makecell[l]{PPL: 6,698 target tokens}
            & PPL and ROUGE-L~\cite{lin2004rouge} \\
             & Rosetta~\cite{wu2024fedbiot} & 7,953 & &
            & \makecell[l]{HumanEval-X~\cite{zheng2023codegeex}: 164/language}
            & pass@1 \\
            \midrule
            B & \makecell[l]{Five-domain\\(this study)} & 5,000 & 5
            & \makecell[l]{One domain/client}
            & \makecell[l]{1,500 test examples; 300/domain}
            & \makecell[l]{Own-domain NLL;\\ public-task scores} \\
            \midrule
            C & MetaMath~\cite{yu2024metamath} & 45,000 & 8 & MetaMath LDA
            & \makecell[l]{16,912 client-round observations}
            &   \makecell[l]{Loss; consensus-state statistics}  \\
            \bottomrule
        \end{tabular}
\end{table*}

\subsubsection{  Quantitative  Metric Design.} According to above, we define two complementary  metrics that measure where a given trajectory lies on the segment between. Let  $\widetilde{g}_{i,\perp} = \widetilde{g}_i - \text{Proj}_{\widetilde{G}}(\widetilde{g}_i)$ be the component of $\widetilde{g}_i^{(t)}$ orthogonal to the global consensus direction. This axis captures updates that have zero directional derivative on the global average objective while contributing to client $i$ 's personalization. We define the drift reduction (DR) as $\mathrm{DR}_i^{(t)}:=1-\nicefrac{\|r_i^{(t)}\|}{\|\delta_i^{(t)}\|}$, which measures the fraction of the initial heterogeneity gap $\delta_i^{(t)}$
that has been cancelled toward the consensus direction.  When $r_i=0$, we have $\mathrm{DR}=1$, indicating full drift elimination; when $r_i=\delta_i$, we have $\mathrm{DR}=0$, indicating no drift reduction. Similarly, we define the personalization preservation (PPR) as $\mathrm{PPR}_i^{(t)}:=\nicefrac{\langle u_i^{(t)}, \tilde{g}_{i, \perp}^{(t)}\rangle}{\|\tilde{g}_{i, \perp}^{(t)}\|^2}$, which measures the projection of the reference field onto the purely local axis. When $u_i$ is fully aligned with $\tilde{g}_i$, we have $\operatorname{PPR}=1$; when $u_i$ is fully aligned with $\widetilde{G}$, we have $\operatorname{PPR}=0$. Substituting $r_i^{(t)}=\psi_t \delta_i$ and $u_i^{(t)}=\psi_t \widetilde{g}_i+(1-\psi_t) \widetilde{G}$ into the definitions of DR and PPR yields $\mathrm{DR}_i^{(t)}=1-\psi_t$ and $\mathrm{PPR}_i^{(t)}=\psi_t$.

\section{Experiments}
\label{sec:experiments}

\subsection{Experiment Setup}
\label{subsec:experimental_setup}

The experiments examine three parts of CIDERS. Protocol A evaluates downstream performance on two backbones. The recurring Output-KL ablation uses Qwen2.5-3B. Protocol B evaluates trained client models on five domains and shared tasks. Protocol C tracks training dynamics and consensus diagnostics under two MetaMath partitions. A separate compression sweep varies messenger depth and nominal drop ratio, while a common-probe run compares training progress and final global-model responses. Table~\ref{tab:protocol_dataset_audit} summarizes datasets, partitions, sample counts, and evaluation units. Unless stated otherwise, CIDERS-FO denotes the first-order solver and CIDERS-FD the finite-difference solver. Section~\ref{subsec:evaluation_diagnostics} defines the evaluation metrics and diagnostics.

\para{Model preparation.}
Protocols A and B use a six-layer messenger. Protocol C also evaluates two- and four-layer messengers. Clients train only messenger LoRA parameters, following the parameter-isolation principle of adapter-based tuning~\cite{houlsby2019adapter}. We use LoRA rank 4, scaling $\alpha_{\mathrm{LoRA}}=16$, and dropout 0.05. A six-layer messenger contains 84 trainable LoRA tensors across attention and MLP projections. Server-side task-aware knowledge distillation (KD~\cite{hinton2015distilling}) trains the learnable backbone LoRA separately. The messenger--learnable-backbone deployment path (MsgLB) combines both components for inference. Each consensus state has the same dimensionality as the messenger LoRA state. The Qwen2.5~\cite{qwen2024qwen25} and Qwen2~\cite{yang2024qwen2} backbones follow their technical reports. Table~\ref{tab:model_payload} reports model depths, LoRA parameter counts, and consensus-state sizes.

\para{Data and configuration.}
Protocol A uses an independent and identically distributed (IID) training-data partition across five clients for GSM8K, Dolly-15K, and Rosetta. The loader uses train/validation/test proportions $[0.998,0.001,0.001]$; the training counts in Table~\ref{tab:protocol_dataset_audit} exclude the held-out portions. GSM8K is evaluated on its benchmark test set, and Dolly on the prompt--response corpus. Code models are trained on the Python, Java, and Go subsets of Rosetta and evaluated on HumanEval-X, which contains 164 problems per language.

\begin{table}[!htbp]
    \caption{Model composition and FP32 LoRA state sizes}
    \label{tab:model_payload}
    \fontsize{9}{10}\selectfont
\renewcommand{\arraystretch}{1}
\setlength{\tabcolsep}{10.5pt}
    \centering
        \begin{tabular}{l|r|r}
            \toprule
            \textbf{Backbone} & \textbf{Qwen2.5-3B} & \textbf{Qwen2-1.5B} \\
            \midrule
              \makecell[l]{Full/LB/Msg. layers}  & 36/24/6 & 28/18/6 \\
            messenger LoRA params & 1.247M & 0.989M \\
            LB  LoRA params & 4.989M & 2.968M \\
            Consensus state (MiB) & 4.76 & 3.77 \\
            \bottomrule
        \end{tabular}
\end{table}

Protocol B assigns medical, finance, customer-support, code, and general-domain data to five clients. Each domain has 1,000 private training examples, 200 public distillation examples, and 300 evaluation examples. The client-model comparison uses Qwen2-1.5B with an 18-layer learnable backbone and a six-layer messenger; compressed methods share this composition and the client splits. Each model is evaluated on its own 300-example test set and on shared public tests: 92 Dolly examples, 1,319 GSM8K problems, and 164 HumanEval-X problems per language. Full Qwen and untrained pruned models serve as capability references, not trained methods. The auxiliary all-to-all evaluation tests each of the five source messengers on all five target domains. Its client and evaluation splits match, but its server-side distillation sources differ.
In \texttt{MODELPREP}, $\beta$ denotes the retained fraction of the non-messenger backbone. The compression sweep reports the nominal deletion setting $\rho_{\mathrm{drop}}$ instead. Its implementation applies the discrete rule $n_L=\operatorname{round}(n_T(1-\rho_{\mathrm{drop}}))$ to the teacher-region depth $n_T=n-d$ and then selects the retained layers by the recorded endpoint-spaced index map. The realized retention can be summarized as $\beta_{\mathrm{eff}}=n_L/n_T$. For Qwen2.5-3B, the recorded $n=36,d=6$ configuration uses $n_T=30$ and $n_L=24$; for Qwen2-1.5B, $n=28$ and messenger depths $d=2,4,6$ give $n_T=26,24,22$ and, at $\rho_{\mathrm{drop}}=0.2$, $n_L=21,19,18$, respectively. This experiment-level rule is distinct from the floor expression in Algorithm~\ref{alg:preparation}.

Protocol C uses 45,000 MetaMath training examples and two held-out splits of 2,500 examples each. The LDA partitioner assigns training data to eight clients at Dirichlet parameters $\alpha=0.1$ and $0.4$, labeled groups A and B. Temporal summaries use 14 complete FO/FD trajectories across two messenger depths, two nominal drop ratios, and rounds 0--150. The coefficient-response analysis adds one configuration.

\para{Training.}
Protocol A compares CIDERS with FedBiOT~\cite{wu2024fedbiot} and the FedOT and LocalOT~\cite{xiao2023offsitetuning} variants used in our evaluation. FedOT follows the federated Offsite-Tuning implementation pattern; LocalOT omits cross-client aggregation. Federated methods run for 150 rounds with $K=10$ local optimizer updates per participating client per round. LocalOT runs 150 outer iterations without cross-client communication. The client-model comparison in Table~\ref{tab:ciders-client-algorithm-results} matches the compressed methods at 150 rounds, microbatch size 2, gradient accumulation 1, and a configured maximum of $K=20$ local updates per participating client per round. The effective batch size is 2 examples per optimizer update.

The separate common-probe comparison uses the same Qwen checkpoint, compressed architecture, LoRA layout, data partitions, and initialization for CIDERS and FedBiOT. FedBiOT first aligns the emulator, then trains clients with cross-entropy. Both runs use five clients, seed 42, 150 rounds, batch size 2, gradient accumulation 1, and a configured maximum of 20 local updates per round. Each method processes 30,000 training examples, 6,294,780 input tokens, and 3,022,134 target tokens. The logs record 14,970 successful CIDERS updates and 15,000 successful FedBiOT updates. A fixed set of 32 training examples per client supplies the common probe.

CIDERS uses AdamW~\cite{loschchilov2019decoupled} with base learning rate $10^{-4}$, weight decay 0.01, global-norm clipping at 1.0, and a warmup-cosine schedule. The lookahead step size $\eta_{\mathrm{in}}$ follows the client learning-rate schedule. In Protocol A, both backbones use a gradient accumulation factor of 1 for all methods. The server uses aggregation step $\tau=0.02$, consensus-variate update coefficient $\alpha_c=0.002$, and local anchoring weight $\lambda_{\mathrm{prox}}=0.01$. Server-side teacher-assistant distillation follows the teacher--student objective and the Teacher Assistant Knowledge Distillation (TAKD~\cite{mirzadeh2020takd}) construction. The coefficient $\alpha_c$ is distinct from Protocol C's Dirichlet parameter $\alpha$.

\subsection{Evaluation and Diagnostic Metrics}
\label{subsec:evaluation_diagnostics}

\paragraph{Task metrics}
Task metrics evaluate model outputs; solver and consensus-state diagnostics describe gradient fields along the training trajectory and do not enter the training objective. For Protocol A, higher accuracy, ROUGE-L, and pass@1 are better, whereas lower PPL is better. ROUGE-L measures sequence overlap. Dolly PPL is the exponential of the corpus-level response NLL, weighted by the 6,698 scored target tokens. Prompt tokens provide context and are masked from scoring.

\paragraph{Client-model metrics}
Own-domain response NLL evaluates each client model on its domain test set. For Protocol B, each response NLL first averages negative log probabilities over valid reference-response tokens; the client metric then averages response NLLs equally within a domain. Prompt positions are masked, and causal shifting and truncation determine scored response positions. Client PPL is $\mathrm{PPL}_i=\exp(\mathrm{NLL}_i)$. For $N=5$ clients, the reported means are $\overline{\mathrm{NLL}}=\nicefrac{1}{N}\sum_i\mathrm{NLL}_i$ and $\overline{\mathrm{PPL}}=\nicefrac{1}{N}\sum_i\mathrm{PPL}_i$; mean PPL is the arithmetic mean of client PPL values, not the exponential of mean NLL. Public-task scores use equal client weights. HumanEval-X uses greedy pass@1, with one deterministic completion per problem.

\paragraph{Common training-probe NLL}
For each client $i$, let $\mathcal{P}_i$ contain the fixed training-probe examples, and let $\mathcal{V}(z)$ denote the valid shifted response positions of example $z$. We evaluate negative log-likelihood (NLL) at the actual post-local parameters $\theta_{i,t}^{\mathrm{post}}$ after completed communication round $t$:
\begin{equation}
\ell_{i,t}^{\mathrm{probe}}=-\frac{\sum_{z\in\mathcal{P}_i}\sum_{j\in\mathcal{V}(z)}\log p_{\theta_{i,t}^{\mathrm{post}}}(y_j\mid x,y_{<j})}{\sum_{z\in\mathcal{P}_i}|\mathcal{V}(z)|}.
\label{eq:common_probe_nll}
\end{equation}
The five-client macro average is $\bar{\ell}_{t}^{\mathrm{probe}}=\nicefrac{1}{5}\sum_{i=1}^{5}\ell_{i,t}^{\mathrm{probe}}$. Both methods use identical probe examples, response masks, and scoring rules. Prompt and padding labels are excluded; valid response EOS tokens are included. The evaluator uses the model's actual parameters rather than lookahead weights and excludes auxiliary training penalties. Round zero denotes a separate evaluation before the first local fit; rounds 1--150 denote completed communication rounds. These training-probe values are distinct from held-out response NLL.

\begin{table*}
    \centering
    \caption{Downstream task performance of compressed MsgLB under Protocol A. Each cell gives the score followed by its signed relative change (\%) versus FedBiOT within the same backbone.}
    \label{tab:main_results}
\centering
    \small
    \renewcommand{\arraystretch}{1.18}
    \setlength{\tabcolsep}{1.5pt}
    \renewcommand{\metricbarfill}{metricfillblue}
    \setlength{\metricbarwidth}{22mm}
    \begin{tabular*}{\textwidth}{@{\extracolsep{\fill}}lcccccc@{}}
        \toprule
        & \textbf{GSM8K} & \multicolumn{2}{c}{\textbf{Dolly}} & \multicolumn{3}{c}{\textbf{HumanEval-X pass@1}} \\
        \cmidrule(lr){2-2}\cmidrule(lr){3-4}\cmidrule(l){5-7}
        \textbf{Method} & \textbf{Acc. $\uparrow$} & \textbf{PPL $\downarrow$} & \textbf{R-L $\uparrow$} & \textbf{Py $\uparrow$} & \textbf{Java $\uparrow$} & \textbf{Go $\uparrow$} \\
        \midrule
        \rowcolor{panelshade}\multicolumn{7}{@{}l@{}}{\textbf{Qwen2-1.5B / MsgLB}\hspace{1em}\footnotesize 18-layer learnable backbone + messenger} \\
        \textbf{CIDERS (ours)} & \comparehigh{7.88}{2.58}{2.05}{7.88}{\textbf{7.88\%}}{2}{\,\mathrm{pp}} & \comparelow{9.53}{10.52}{9.53}{11.52}{\textbf{9.53}}{2}{} & \comparehigh{0.2176}{0.1812}{0.1612}{0.2176}{\textbf{0.2176}}{4}{} & \comparehigh{7.9}{3.0}{1.2}{7.9}{\textbf{7.9\%}}{2}{\,\mathrm{pp}} & \comparehigh{3.0}{3.7}{1.8}{3.7}{\underline{3.0\%}}{2}{\,\mathrm{pp}} & \comparehigh{4.9}{1.8}{1.2}{4.9}{\textbf{4.9\%}}{2}{\,\mathrm{pp}} \\
        \addlinespace[1.5pt]
        FedBiOT & \comparehigh{2.58}{2.58}{2.05}{7.88}{\underline{2.58\%}}{2}{\,\mathrm{pp}} & \comparelow{10.52}{10.52}{9.53}{11.52}{\underline{10.52}}{2}{} & \comparehigh{0.1812}{0.1812}{0.1612}{0.2176}{\underline{0.1812}}{4}{} & \comparehigh{3.0}{3.0}{1.2}{7.9}{\underline{3.0\%}}{2}{\,\mathrm{pp}} & \comparehigh{3.7}{3.7}{1.8}{3.7}{\textbf{3.7\%}}{2}{\,\mathrm{pp}} & \comparehigh{1.8}{1.8}{1.2}{4.9}{\underline{1.8\%}}{2}{\,\mathrm{pp}} \\
        FedOT & \comparehigh{2.05}{2.58}{2.05}{7.88}{2.05\%}{2}{\,\mathrm{pp}} & \comparelow{11.52}{10.52}{9.53}{11.52}{11.52}{2}{} & \comparehigh{0.1639}{0.1812}{0.1612}{0.2176}{0.1639}{4}{} & \comparehigh{1.2}{3.0}{1.2}{7.9}{1.2\%}{2}{\,\mathrm{pp}} & \comparehigh{1.8}{3.7}{1.8}{3.7}{1.8\%}{2}{\,\mathrm{pp}} & \comparehigh{1.2}{1.8}{1.2}{4.9}{1.2\%}{2}{\,\mathrm{pp}} \\
        LocalOT & \comparehigh{2.20}{2.58}{2.05}{7.88}{2.20\%}{2}{\,\mathrm{pp}} & \comparelow{11.52}{10.52}{9.53}{11.52}{11.52}{2}{} & \comparehigh{0.1612}{0.1812}{0.1612}{0.2176}{0.1612}{4}{} & \comparehigh{1.2}{3.0}{1.2}{7.9}{1.2\%}{2}{\,\mathrm{pp}} & \comparehigh{2.4}{3.7}{1.8}{3.7}{2.4\%}{2}{\,\mathrm{pp}} & \comparehigh{1.2}{1.8}{1.2}{4.9}{1.2\%}{2}{\,\mathrm{pp}} \\
        \addlinespace[3pt]
        \noalign{\global\def\metricbarfill{metricfillgreen}}
        \rowcolor{panelshade}\multicolumn{7}{@{}l@{}}{\textbf{Qwen2.5-3B / MsgLB}\hspace{1em}\footnotesize 24-layer learnable backbone + messenger} \\
        \textbf{CIDERS (ours)} & \comparehigh{16.38}{5.23}{4.40}{16.38}{\textbf{16.38\%}}{2}{\,\mathrm{pp}} & \comparelow{8.60}{9.42}{8.60}{10.50}{\textbf{8.60}}{2}{} & \comparehigh{0.2384}{0.2307}{0.2151}{0.2384}{\textbf{0.2384}}{4}{} & \comparehigh{9.1}{6.7}{6.7}{9.1}{\textbf{9.1\%}}{2}{\,\mathrm{pp}} & \comparehigh{11.0}{6.7}{3.0}{11.0}{\textbf{11.0\%}}{2}{\,\mathrm{pp}} & \comparehigh{6.7}{3.0}{1.2}{6.7}{\textbf{6.7\%}}{2}{\,\mathrm{pp}} \\
        \addlinespace[1.5pt]
        FedBiOT & \comparehigh{5.23}{5.23}{4.40}{16.38}{\underline{5.23\%}}{2}{\,\mathrm{pp}} & \comparelow{9.42}{9.42}{8.60}{10.50}{\underline{9.42}}{2}{} & \comparehigh{0.2307}{0.2307}{0.2151}{0.2384}{\underline{0.2307}}{4}{} & \comparehigh{6.7}{6.7}{6.7}{9.1}{\underline{6.7\%}}{2}{\,\mathrm{pp}} & \comparehigh{6.7}{6.7}{3.0}{11.0}{\underline{6.7\%}}{2}{\,\mathrm{pp}} & \comparehigh{3.0}{3.0}{1.2}{6.7}{\underline{3.0\%}}{2}{\,\mathrm{pp}} \\
        FedOT & \comparehigh{4.70}{5.23}{4.40}{16.38}{4.70\%}{2}{\,\mathrm{pp}} & \comparelow{10.50}{9.42}{8.60}{10.50}{10.50}{2}{} & \comparehigh{0.2234}{0.2307}{0.2151}{0.2384}{0.2234}{4}{} & \comparehigh{6.7}{6.7}{6.7}{9.1}{\underline{6.7\%}}{2}{\,\mathrm{pp}} & \comparehigh{4.3}{6.7}{3.0}{11.0}{4.3\%}{2}{\,\mathrm{pp}} & \comparehigh{1.2}{3.0}{1.2}{6.7}{1.2\%}{2}{\,\mathrm{pp}} \\
        LocalOT & \comparehigh{4.40}{5.23}{4.40}{16.38}{4.40\%}{2}{\,\mathrm{pp}} & \comparelow{10.50}{9.42}{8.60}{10.50}{10.50}{2}{} & \comparehigh{0.2151}{0.2307}{0.2151}{0.2384}{0.2151}{4}{} & \comparehigh{6.7}{6.7}{6.7}{9.1}{\underline{6.7\%}}{2}{\,\mathrm{pp}} & \comparehigh{3.0}{6.7}{3.0}{11.0}{3.0\%}{2}{\,\mathrm{pp}} & \comparehigh{1.2}{3.0}{1.2}{6.7}{1.2\%}{2}{\,\mathrm{pp}} \\
        \bottomrule
    \end{tabular*}
\end{table*}

\paragraph{Final global model evaluation}
We evaluate the server's final compressed MsgLB model, comprising the messenger and learnable backbone, on the same 1,500 held-out inputs with a common tokenizer, prompt template, response mask, and maximum sequence length of 800. Two inputs have no valid response tokens after truncation; both methods therefore score the same 1,498 responses and 155,211 valid response tokens. Domain and pooled NLL use valid-response-token weighting: pooled NLL divides the summed held-out response loss by the total valid-token count, and pooled PPL is its exponential. These pooled metrics differ from the equal-domain mean NLL and arithmetic mean of domain PPL values. Per-sample comparisons instead use each response's token-averaged NLL.

\paragraph{Auxiliary cross-domain metric}
The auxiliary Protocol B evaluation tests every source-domain messenger on every target domain. Let $L_{ij}$ denote the response NLL of source-domain messenger $i$ on target domain $j$. The average matched-domain advantage is
\begin{equation}
    \Delta_{\mathrm{cross}}
    =\nicefrac{1}{20}\sum\nolimits_{\substack{1\leq i,j\leq5\\i\ne j}}L_{ij}
    -\nicefrac{1}{5}\sum\nolimits_{i=1}^{5}L_{ii}.
    \label{eq:cross_domain_gap}
\end{equation}
This statistic gives each target domain equal weight by comparing its matched messenger with the mean of its four unmatched messengers. Positive $\Delta_{\mathrm{cross}}$ indicates lower NLL for the matched messenger on average.

\paragraph{Solver and consensus-state diagnostics}
For the finite-difference analysis, we measure the curvature correction relative to the first-order direction:
\begin{equation}
    r_{\mathrm{FD}}
    =\nicefrac{\|\lambda_{\mathrm{HVP}}\widehat{Hv}\|_2}
    {(\|v\|_2+\epsilon_{\mathrm{num}})},
    \label{eq:fd_relative_correction}
\end{equation}
where $v$ is the first-order direction, $\widehat{Hv}$ its finite-difference Hessian--vector estimate, and $\lambda_{\mathrm{HVP}}$ the applied coefficient. Both vectors include only trainable LoRA parameters. Here, $\epsilon_{\mathrm{num}}=10^{-12}$ is a numerical stabilizer, distinct from the proximal coefficient and the FD perturbation scale. An optimizer-step observation is small when $r_{\mathrm{FD}}\leq0.1$, moderate when $0.1<r_{\mathrm{FD}}\leq0.2$, and large when $r_{\mathrm{FD}}>0.2$.

For Protocol C, let $s=(i,t,k)$ denote local step $k$ of client $i$ in round $t$, and let $\mathcal{S}_{i,t}$ contain its recorded steps. Let $g_s$ denote the recorded LoRA gradient after the consensus correction and proximal term, before clipping and the AdamW step. We define an auxiliary diagnostic vector by adding another consensus term:
\begin{equation}
    g_{\mathrm{diag},s}
    =g_s+\gamma_{t(s)}
    \bigl(c_{\mathrm{global},s}-c_{\mathrm{local},s}\bigr).
    \label{eq:diagnostic_effective_gradient}
\end{equation}
The additional correction is used only for measurement and is not written back to the training gradients. The statistics below characterize this auxiliary field, rather than the optimizer's parameter update after clipping and AdamW.

Define $e_s=g_{\mathrm{diag},s}-c_{\mathrm{global},s}$ and $q_s=\|g_s\|_2^2+\|c_{\mathrm{global},s}\|_2^2+\epsilon_{\mathrm{num}}$. The normalized squared-residual statistic is
\begin{equation}
    \mathrm{E\mbox{-}RAE}_{i,t}
    =\nicefrac{\sum_{k\in\mathcal{S}_{i,t}}\|e_s\|_2^2}
    {\sum_{k\in\mathcal{S}_{i,t}}q_s}.
    \label{eq:erae_diagnostic}
\end{equation}
Each value is a ratio of sums over one client-round. Zero denotes an exact match to the global consensus variate at every recorded step; a value near one does not denote a vanishing residual.

Sample-based drift reduction (S-DR) measures the change in squared residuals. Sample-based personalization preservation (S-PPR) measures retention of the recorded field's local projection. With numerical regularization, this projection is
\begin{equation}
    \widetilde g_{\perp,s}
    =g_s-\frac{\langle g_s,c_{\mathrm{global},s}\rangle}
    {\lVert c_{\mathrm{global},s}\rVert_2^2+\epsilon_{\mathrm{num}}}
    c_{\mathrm{global},s}.
    \label{eq:empirical_local_component}
\end{equation}
We compute the client-round ratios
\begin{align}
    &\mathrm{S\mbox{-}DR}_{i,t}
    =1-\frac{\sum_{k\in\mathcal{S}_{i,t}}
    \lVert g_{\mathrm{diag},s}-c_{\mathrm{global},s}\rVert_2^2}
    {\sum_{k\in\mathcal{S}_{i,t}}
    \lVert g_s-c_{\mathrm{global},s}\rVert_2^2},
    \label{eq:empirical_sdr}\\*
    &\mathrm{S\mbox{-}PPR}_{i,t}
    =\frac{\sum_{k\in\mathcal{S}_{i,t}}
    \langle g_{\mathrm{diag},s},\widetilde g_{\perp,s}\rangle}
    {\sum_{k\in\mathcal{S}_{i,t}}
    \langle g_s,\widetilde g_{\perp,s}\rangle}.
    \label{eq:empirical_sppr}
\end{align}
When a denominator is nonpositive, the recorded diagnostic is zero. Negative S-DR means that $g_{\mathrm{diag},s}$ has a larger summed squared residual relative to the global consensus variate than the $g_s$ baseline. Because this baseline already includes consensus correction, S-DR measures the effect of an additional diagnostic correction, rather than the algorithm's net correction of an uncorrected local gradient. It is a squared-residual diagnostic inspired by theoretical DR, whose norm ratio uses the global average gradient at the broadcast point. The empirical and theoretical quantities therefore differ in both reference and numerical scale.
\paragraph{Geometric interpretation of the diagnostics}
For geometric interpretation, let $b_s=g_s-c_{\mathrm{global},s}$ and $d_s=g_{\mathrm{diag},s}-g_s$. Expanding Eq.~\eqref{eq:empirical_sdr}, with all sums over $k\in\mathcal{S}_{i,t}$ and a positive denominator, gives
\begin{equation}
    \mathrm{S\mbox{-}DR}_{i,t}
    =-\frac{2\sum_k\langle b_s,d_s\rangle
    +\sum_k\lVert d_s\rVert_2^2}
    {\sum_k\lVert b_s\rVert_2^2}.
    \label{eq:diagnostic_residual_decomposition}
\end{equation}
Negative S-DR occurs when the alignment term fails to offset the correction's added squared norm. High S-PPR and negative S-DR can thus coexist: the auxiliary field can retain its local projection while its total residual grows. The idealized geometry links drift reduction and local-component retention through $\mathrm{DR}=1-\psi_t$ and $\mathrm{PPR}=\psi_t$. The coefficient endpoints also depend on the metric construction. For the unsmoothed ratios, if $\gamma_t=0$ throughout a recorded client-round, then $g_{\mathrm{diag},s}=g_s$, $\mathrm{S\mbox{-}DR}=0$, and $\mathrm{S\mbox{-}PPR}=1$ whenever the corresponding denominators are positive. These identities apply to raw ratios; three-observation rolling means can include neighboring rounds. In the coefficient plots, dashed curves provide coordinate guides $\mathrm{S\mbox{-}DR}=\gamma_t$ and $\mathrm{S\mbox{-}PPR}=1-\gamma_t$. E-RAE uses the unsmoothed client-round ratio.

\begin{figure}
\centering
    \includegraphics[width=\columnwidth]
    {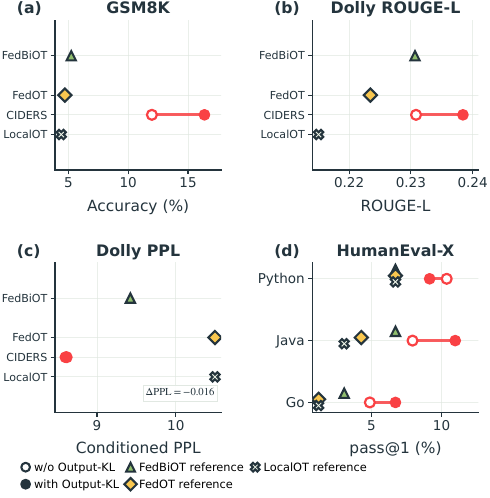}
    \Description{Four panels compare CIDERS-FO configurations with and
        without recurring Output-KL on GSM8K, Dolly, and HumanEval-X. Additional
        method markers show FedBiOT, FedOT, and LocalOT MsgLB results for
        reference.}
    \caption{\small\textbf{CIDERS-FO output alignment on Qwen2.5-3B.} Open and filled circles show results without and with recurring Output-KL; other markers indicate Protocol A baselines.}
    \label{fig:output_kl}
\end{figure}

\subsection{Task Performance and Output Alignment}
\label{subsec:main_performance}
\label{subsec:output_kl}

Protocol A evaluates the complete MsgLB deployment path on Qwen2-1.5B and Qwen2.5-3B. CIDERS leads FedBiOT, FedOT, and LocalOT in 11 of the 12 backbone--metric comparisons in Table~\ref{tab:main_results}. The clearest gains over FedBiOT occur on GSM8K. CIDERS improves exact-match accuracy by 5.30 percentage points on Qwen2-1.5B and 11.15 points on Qwen2.5-3B. It also lowers Dolly perplexity and improves ROUGE-L on both backbones. These Dolly results show gains in both reference-response likelihood and text overlap alongside the gains in mathematical answer accuracy. Java on Qwen2-1.5B is the exception. The pass@1 scores of 3.0\% for CIDERS and 3.7\% for FedBiOT differ by about one solved problem on the 164-problem benchmark. Within CIDERS, the Qwen2.5-3B configuration scores 8.50 percentage points higher on GSM8K and 8.0 points higher on Java than the Qwen2-1.5B configuration. It retains 24 learnable-backbone layers compared with 18 in the smaller configuration. Its messenger also contains 26.1\% more LoRA parameters. At the same numerical precision, that increase carries over to the size of each messenger-sized communication tensor. The comparison thus shows how task performance and per-message payload change together across the two tested configurations. Pretrained backbone, retained depth, and messenger capacity all vary between them.

The recurring Output-KL ablation examines a specific part of the server update. During each distillation step in Eq.~\eqref{eq:emulator_obj}, the teacher and compressed paths use the same current messenger. The KL term gives the learnable backbone a distribution to match under that shared messenger. It conveys the teacher's relative probabilities across tokens in addition to the supervision from the reference response. Fig.~\ref{fig:output_kl} and Table~\ref{tab:fd_output_kl} show how this signal affects downstream performance on Qwen2.5-3B. Adding recurring Output-KL raises GSM8K accuracy from 11.98\% to 16.38\% with either the FO or FD solver. On Dolly, the FD ablation changes perplexity only from 8.616 to 8.602, even though ROUGE-L improves. The code results reveal a different pattern. Java and Go pass@1 rise with both solvers, and Python pass@1 falls by 1.22 percentage points in both cases. The mean pass@1 across the three languages improves, but the mean hides the loss on Python. This distinction matters because the KL term aligns token distributions during training, whereas HumanEval-X credits a completion only when the complete program passes its tests. The ablation therefore shows a consistent GSM8K gain across the two tested solvers and a language-dependent trade-off on executable code.

\begin{figure}
\centering
    \includegraphics[width=\columnwidth]
    {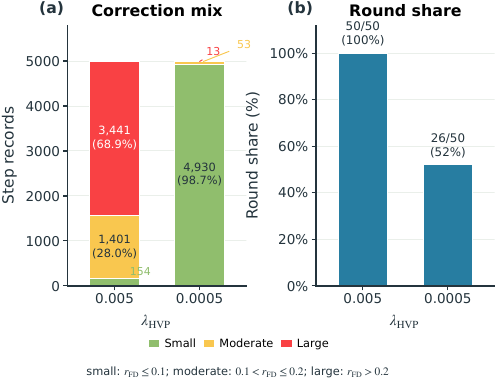}
    \Description{Two panels compare finite-difference correction profiles at
        $\lambda_{\mathrm{HVP}}=0.005$ and $0.0005$. The left panel shows small,
        moderate, and large counts among 4,996 optimizer-step observations. The
        right panel shows the fraction of 50 rounds containing at least one
        moderate or large correction event.}
    \caption{\small\textbf{FD correction profiles under independent scaling.} Step-category counts (left) and the share of rounds with a moderate or large correction (right) summarize 4,996 recorded steps over 50 rounds at two $\lambda_{\mathrm{HVP}}$ settings.}
    \label{fig:solver_sensitivity}
\end{figure}

\begin{figure*}
    \centering
\includegraphics[width=\textwidth]
    {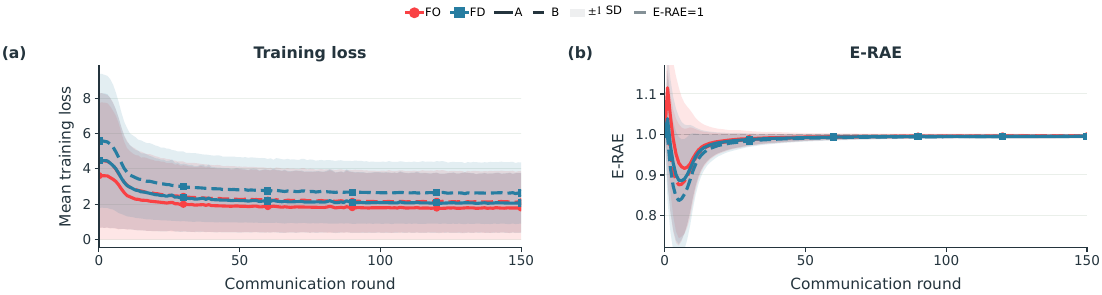}
    \Description{Two panels show CIDERS-FO and CIDERS-FD training
        loss and E-RAE over rounds 0--150. Colors and markers distinguish the
        solvers, line styles distinguish trace groups A ($\alpha=0.1$)
        and B ($\alpha=0.4$), and translucent
        bands show one configuration standard deviation.}
    \caption{\small\textbf{Loss and E-RAE under two MetaMath partitions.} FO/FD denote the first-order and finite-difference solvers; A/B denote $\alpha=0.1/0.4$; shading shows $\pm1$ configuration standard deviation; the dashed line marks E-RAE $=1$.}
    \label{fig:bgp_temporal_dynamics}
\end{figure*}

\begin{figure}
\centering
    \includegraphics[width=0.94\columnwidth]
    {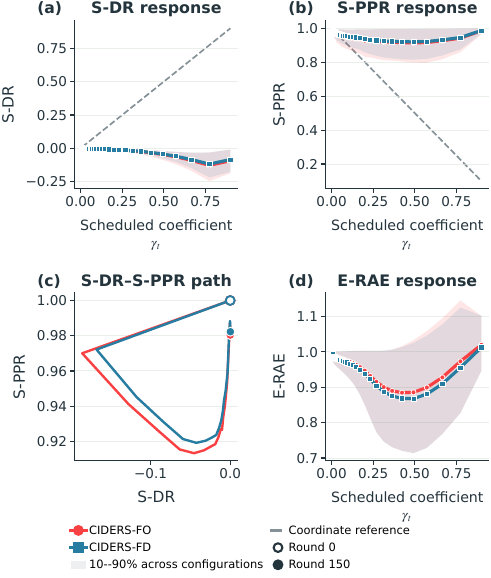}
    \Description{Four panels report the CIDERS-FO and CIDERS-FD
        S-DR, S-PPR, and E-RAE responses to the scheduled coefficient,
        together with their mean S-DR/S-PPR trajectories.
        Open and filled endpoints mark the actual round-0 and round-150
        configuration aggregates.}
    \caption{\small\textbf{Consensus diagnostics along the coefficient schedule.} Shading spans the 10th--90th configuration percentiles; open and filled endpoints mark rounds 0 and 150.}
    \label{fig:control_interpolation}
\end{figure}

\subsection{Finite-Difference Solver Diagnostics}
\label{subsec:fd_diagnostics}

The canonical FD update uses the inner learning rate as the weight of the Hessian--vector correction, $\lambda_{\mathrm{HVP}}=\eta_{\mathrm{in}}$. In our experiments, we set the correction weight independently. It is fixed at 0.001 in Protocol~A and 0.0005 in Protocol~B. The inner learning rate still follows the client schedule with a base rate of $10^{-4}$. This setting keeps the weight of the curvature term fixed as the inner learning rate changes. Equation~\eqref{eq:fd_relative_correction} uses the coefficient applied in each run to measure the size of the correction relative to the first-order direction. For a fixed first-order direction and Hessian--vector estimate, this ratio scales linearly with $\lambda_{\mathrm{HVP}}$. The ratio quantifies the size of the change to the raw direction before the optimizer step. Its geometric effect also depends on orientation. A small correction parallel to the first-order direction mainly changes its magnitude. A perpendicular component changes its direction.

The separate 50-round sweep in Fig.~\ref{fig:solver_sensitivity} shows how often the correction is large enough to matter at the step level. At $\lambda_{\mathrm{HVP}}=0.005$, 68.9\% of the 4,996 recorded steps have large corrections and 28.0\% have moderate corrections. At 0.0005, only 1.3\% of steps have moderate or large corrections. Yet at least one such step occurs in 26 of 50 rounds, compared with all 50 rounds at 0.005. A single qualifying step is enough for a round to enter this count. Thus, the smaller coefficient produces mostly small corrections, with occasional larger corrections spread across more than half of the rounds. The Protocol~A benchmark comparison puts these diagnostics in context. At the reported precision, FO and FD agree in nine of the twelve comparisons. The largest pass@1 difference is 0.6 percentage points on Qwen2.5-3B Java, which corresponds to about one solved problem. The Dolly ROUGE-L scores differ by 0.0045. Given the additional gradient evaluations required by FD, these results provide limited evidence of a task-level benefit at the tested Protocol~A setting.

\begin{table}
    \caption{Recurring Output-KL ablation for CIDERS-FD on Qwen2.5-3B}
    \label{tab:fd_output_kl}
    \fontsize{9}{10}\selectfont
    \renewcommand{\arraystretch}{1}
    \setlength{\tabcolsep}{6.5pt}
    \centering
        \begin{tabular}{l|r|r|r}
            \toprule
            \textbf{Metric} & \textbf{Without KL} & \textbf{With KL} & \textbf{Change} \\
            \midrule
            GSM8K accuracy & 11.98  & 16.38  & $+4.40$   \\
            Dolly PPL      & 8.616  & 8.602  & $-0.014$  \\
            Dolly ROUGE-L  & 0.2271 & 0.2339 & $+0.0068$ \\
            Python pass@1  & 10.37  & 9.15   & $-1.22$   \\
            Java pass@1    & 7.93   & 11.59  & $+3.66$   \\
            Go pass@1      & 4.27   & 6.71   & $+2.44$   \\
            \bottomrule
        \end{tabular}
\end{table}

\begin{table*}
    \caption{Own-domain and shared-task performance of deployed client models under Protocol B. Capability references show scores only. Bars rank the matched methods within each metric row; parenthetical changes are relative to FedBiOT for NLL, PPL, and ROUGE-L, and percentage-point (pp) differences for accuracy and pass@1.}
    \label{tab:ciders-client-algorithm-results}
    \centering
    \small
    \renewcommand{\arraystretch}{1.16}
    \setlength{\tabcolsep}{2pt}
    \renewcommand{\metricbarfill}{metricfillblue}
    \setlength{\metricbarwidth}{26mm}
    \newcommand{\scorecell}[2]{\makebox[24mm][c]{\makebox[10mm][r]{#1}\hspace{0.5mm}{\scriptsize #2}}}
    \newcommand{\metricbarbase}[3]{%
        \makebox[\metricbarwidth][l]{%
            \makebox[0pt][l]{\raisebox{-0.16ex}{\color{metrictrack}\rule{\metricbarwidth}{1.5ex}}}%
            \makebox[0pt][l]{\raisebox{-0.16ex}{\color{\metricbarfill}\rule{\fpeval{#1}\metricbarwidth}{1.5ex}}}%
            \makebox[\metricbarwidth][c]{\strut\makebox[10mm][r]{#2}\hspace{0.5mm}{\scriptsize$#3$}}%
        }%
    }
    \newcommand{\matchhigh}[5]{\metricbarannotated{(#1-#2)/(#3-#2)}{#4}{#5}}
    \newcommand{\matchlow}[5]{\metricbarannotated{(#3-#1)/(#3-#2)}{#4}{#5}}
    \newcommand{\matchbasehigh}[4]{\metricbarbase{(#1-#2)/(#3-#2)}{#4}{\phantom{+2.20\,\mathrm{pp}}}}
    \newcommand{\matchbaselow}[4]{\metricbarbase{(#3-#1)/(#3-#2)}{#4}{\phantom{-21.5\%}}}
    \begin{tabular*}{\textwidth}{@{\extracolsep{\fill}}lrrrrr@{}}
        \toprule
        \multirow{2}{*}{\textbf{Evaluation / metric}}
        & \multicolumn{2}{c}{\textbf{Capability references}}
        & \multicolumn{3}{c}{\textbf{Matched methods}}\\
        \cmidrule(lr){2-3}\cmidrule(l){4-6}
        & \textbf{Full Qwen} & \textbf{Untrained pruned}
        & \textbf{FedOT} & \textbf{FedBiOT (base)} & \cellcolor{panelshade}\textbf{CIDERS (ours)}\\
        \midrule
        \rowcolor{panelshade}\multicolumn{6}{@{}l@{}}{\textbf{Own-domain response NLL $\downarrow$ (relative change vs. FedBiOT)}}\\
        C1 / Medical        & \scorecell{2.2034}{} & \scorecell{2.4715}{} & \matchlow{2.0258}{2.0022}{2.0258}{2.0258}{+0.0\%} & \matchbaselow{2.0249}{2.0022}{2.0258}{2.0249} & \matchlow{2.0022}{2.0022}{2.0258}{\textbf{2.0022}}{-1.1\%}\\
        C2 / Finance        & \scorecell{2.6630}{} & \scorecell{2.9791}{} & \matchlow{2.9237}{2.9182}{2.9311}{2.9237}{-0.3\%} & \matchbaselow{2.9311}{2.9182}{2.9311}{2.9311} & \matchlow{2.9182}{2.9182}{2.9311}{\textbf{2.9182}}{-0.4\%}\\
        C3 / Cust.\ support & \scorecell{1.4725}{} & \scorecell{1.9649}{} & \matchlow{1.5159}{1.1615}{1.5159}{1.5159}{+2.4\%} & \matchbaselow{1.4799}{1.1615}{1.5159}{1.4799} & \matchlow{1.1615}{1.1615}{1.5159}{\textbf{1.1615}}{-21.5\%}\\
        C4 / Code           & \scorecell{0.8897}{} & \scorecell{1.3519}{} & \matchlow{1.1008}{1.0546}{1.1008}{1.1008}{+3.8\%} & \matchbaselow{1.0608}{1.0546}{1.1008}{1.0608} & \matchlow{1.0546}{1.0546}{1.1008}{\textbf{1.0546}}{-0.6\%}\\
        C5 / General        & \scorecell{2.1882}{} & \scorecell{2.4751}{} & \matchlow{2.3875}{2.3730}{2.4198}{2.3875}{-1.3\%} & \matchbaselow{2.4198}{2.3730}{2.4198}{2.4198} & \matchlow{2.3730}{2.3730}{2.4198}{\textbf{2.3730}}{-1.9\%}\\
        Mean                & \scorecell{1.8834}{} & \scorecell{2.2485}{} & \matchlow{1.9907}{1.9019}{1.9907}{1.9907}{+0.4\%} & \matchbaselow{1.9833}{1.9019}{1.9907}{1.9833} & \matchlow{1.9019}{1.9019}{1.9907}{\textbf{1.9019}}{-4.1\%}\\
        \midrule
        \rowcolor{panelshade}\multicolumn{6}{@{}l@{}}{\textbf{Client-mean response quality: PPL $\downarrow$ and ROUGE-L $\uparrow$ (relative change vs. FedBiOT)}}\\
        Own-domain PPL $\downarrow$
                           & \scorecell{7.8217}{} & \scorecell{10.8780}{} & \matchlow{8.9277}{8.5417}{8.9697}{8.9277}{-0.5\%} & \matchbaselow{8.9697}{8.5417}{8.9697}{8.9697} & \matchlow{8.5417}{8.5417}{8.9697}{\textbf{8.5417}}{-4.8\%}\\
        Dolly ROUGE-L $\uparrow$
                           & \scorecell{0.2469}{} & \scorecell{0.0964}{} & \matchhigh{0.1852}{0.1848}{0.1951}{0.1852}{+0.2\%} & \matchbasehigh{0.1848}{0.1848}{0.1951}{0.1848} & \matchhigh{0.1951}{0.1848}{0.1951}{\textbf{0.1951}}{+5.6\%}\\
        Dolly PPL $\downarrow$
                           & \scorecell{8.2303}{} & \scorecell{11.0434}{} & \matchlow{10.7402}{10.5328}{11.1370}{10.7402}{-3.6\%} & \matchbaselow{11.1370}{10.5328}{11.1370}{11.1370} & \matchlow{10.5328}{10.5328}{11.1370}{\textbf{10.5328}}{-5.4\%}\\
        \midrule
        \rowcolor{panelshade}\multicolumn{6}{@{}l@{}}{\textbf{Shared-task accuracy and pass@1 (\%; change vs. FedBiOT in pp)}}\\
        GSM8K accuracy (\%) $\uparrow$
                           & \scorecell{28.20}{} & \scorecell{2.58}{} & \matchhigh{2.08}{2.08}{2.99}{2.08}{-0.91\,\mathrm{pp}} & \matchbasehigh{2.99}{2.08}{2.99}{\textbf{2.99}} & \matchhigh{2.90}{2.08}{2.99}{2.90}{-0.09\,\mathrm{pp}}\\
        Python greedy pass@1 (\%) $\uparrow$
                           & \scorecell{34.76}{} & \scorecell{1.83}{} & \matchhigh{2.56}{2.56}{6.10}{2.56}{-1.59\,\mathrm{pp}} & \matchbasehigh{4.15}{2.56}{6.10}{4.15} & \matchhigh{6.10}{2.56}{6.10}{\textbf{6.10}}{+1.95\,\mathrm{pp}}\\
        Java greedy pass@1 (\%) $\uparrow$
                           & \scorecell{34.15}{} & \scorecell{0.61}{} & \matchhigh{0.61}{0.61}{2.93}{0.61}{-0.12\,\mathrm{pp}} & \matchbasehigh{0.73}{0.61}{2.93}{0.73} & \matchhigh{2.93}{0.61}{2.93}{\textbf{2.93}}{+2.20\,\mathrm{pp}}\\
        Go greedy pass@1 (\%) $\uparrow$
                           & \scorecell{23.17}{} & \scorecell{0.61}{} & \matchhigh{1.34}{1.34}{2.44}{1.34}{-0.12\,\mathrm{pp}} & \matchbasehigh{1.46}{1.34}{2.44}{1.46} & \matchhigh{2.44}{1.34}{2.44}{\textbf{2.44}}{+0.98\,\mathrm{pp}}\\
        \bottomrule
    \end{tabular*}
\end{table*}

\begin{table}
\centering
\caption{Final global model response quality.}
\label{tab:controlled_response_quality}
\small
\setlength{\tabcolsep}{3pt}
\renewcommand{\arraystretch}{1.08}
\begin{tabular*}{\columnwidth}{@{\extracolsep{\fill}}lrrr@{}}
\toprule
Domain / metric & FedBiOT & CIDERS & Reduction \\
\midrule
Medical & 1.9562 & \textbf{1.9058} & 2.58\% \\
Finance & 2.9235 & \textbf{2.9197} & 0.13\% \\
Customer support & 1.5188 & \textbf{1.2817} & 15.61\% \\
Code & 1.0886 & \textbf{1.0253} & 5.82\% \\
General & 2.3801 & \textbf{2.3694} & 0.45\% \\
\midrule
Pooled NLL & 2.2630 & \textbf{2.1921} & 3.13\% \\
Pooled PPL & 9.6120 & \textbf{8.9540} & 6.85\% \\
\bottomrule
\end{tabular*}
\end{table}
\subsection{PBO Dynamics under Data Heterogeneity}
\label{subsec:auxiliary_convergence}

Protocol C tracks training loss and the diagnostics defined in Section~\ref{subsec:evaluation_diagnostics}. Fig.~\ref{fig:bgp_temporal_dynamics} summarizes rounds 0--150 by averaging client-round values within each configuration and reporting the mean and standard deviation across configurations. The curves account for 74.0--76.0\% of the total loss decrease by round 15 and 86.8--88.5\% by round 30. After round 30, E-RAE settles between 0.995 and 0.996 while training loss continues to decrease. Under its stated conditions, the convergence theorem bounds the round-averaged joint gap $\frac{1}{T}\sum_{t=0}^{T-1}J_t$. The empirical curves show how training loss and E-RAE evolve over the same rounds. At round 150, the loss gaps between the two realized LDA partitions are 19.3\% for CIDERS-FO and 26.9\% for CIDERS-FD, with lower loss at $\alpha=0.1$ in both cases. These gaps exceed the differences between FO and FD despite similar late-stage E-RAE. The two partitions also differ in sample composition and difficulty, so the loss ordering describes these particular runs.

Fig.~\ref{fig:control_interpolation} shows that S-PPR remains high as S-DR moves from near zero at small coefficients to negative values at intermediate and larger coefficients. Proposition~\ref{prop:1} links higher DR to lower PPR under its idealized model. The empirical diagnostics apply an offline correction to an already corrected field and measure the local projection and the total residual against different references, as described in Section~\ref{subsec:evaluation_diagnostics}. At a fixed client-round, write $h_s=c_{\mathrm{global},s}-c_{\mathrm{local},s}$ and $d_s=\gamma_t h_s$. In Eq.~\eqref{eq:diagnostic_residual_decomposition}, the alignment term scales linearly with $\gamma_t$, while the squared norm of the added correction scales quadratically. A small correction can reduce the residual through favorable alignment, but the quadratic term can outweigh that reduction as $\gamma_t$ grows. The total residual can therefore increase while the local projection changes little. As $\gamma_t$ approaches zero, the added correction vanishes, which explains the near-zero S-DR endpoint. In the training runs, $\gamma_t$ decreases with communication round and eventually stays at its floor. The intermediate E-RAE minimum thus identifies a phase of the scheduled trajectory. Comparing fixed coefficients at matched rounds would separate the effect of coefficient choice from that of training stage.

\subsection{Training Progress Across Domains}
\label{subsec:controlled_training_progress}

Fig.~\ref{fig:common_probe_nll} compares CIDERS and FedBiOT on the same fixed training probe for each of the five clients. Both runs start from the same compressed initialization and process the same numbers of training examples and tokens over 150 communication rounds. We score each client at its actual post-local parameters. The macro average gives the five clients equal weight, so clients with more scored response tokens do not dominate the summary. CIDERS reaches FedBiOT's round-150 macro probe NLL after 31 completed communication rounds and stays at or below that threshold through round 150. The sustained crossing matters because it shows that the earlier attainment of FedBiOT's final probe-fit level persists through the rest of training.

\begin{figure}[!htbp]
\centering
\includegraphics[width=\columnwidth]{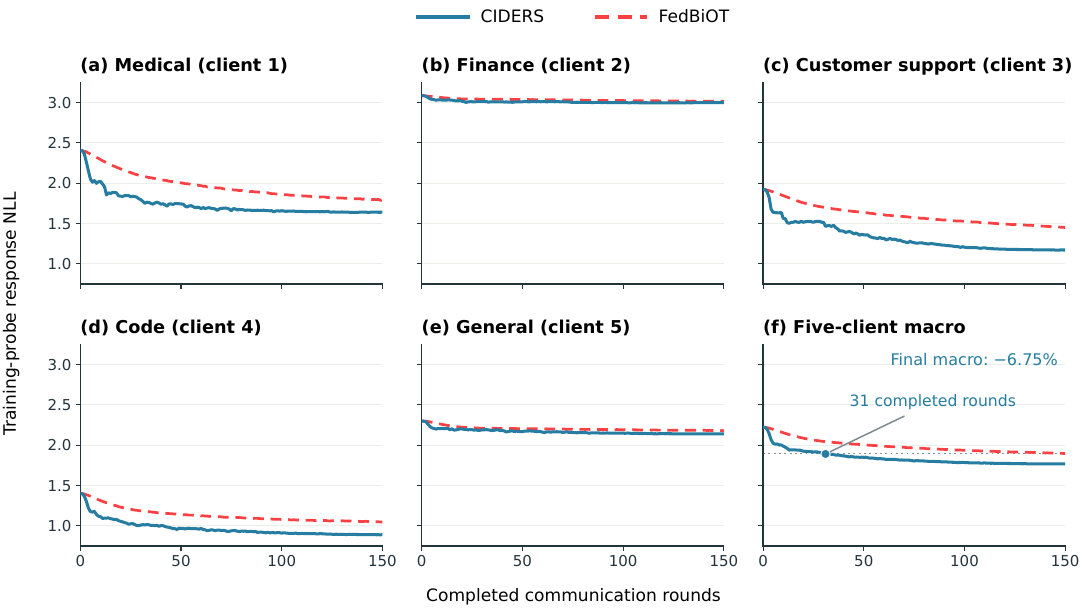}
\caption{\small\textbf{Common training-probe NLL across five clients.} Both methods use the same 32 training examples per client; the macro average weights all five clients equally.}
\label{fig:common_probe_nll}
\end{figure}

Table~\ref{tab:controlled_response_quality} evaluates the final server models on held-out responses from the five domains. CIDERS achieves 3.13\% lower pooled response NLL than FedBiOT. The pooled score weights all valid response tokens equally, so longer responses contribute more to the aggregate. To examine how widely the gains are distributed, we also compare the token-averaged NLL of each response between the two models and count each response once. Customer support has the largest domain-level NLL reduction, and CIDERS assigns lower NLL to 299 of its 300 responses. Across the other four domains combined, CIDERS assigns lower NLL to 842 of 1,198 responses, or about 70\%. The largest reduction is concentrated in customer support, with lower NLL still extending to a majority of responses elsewhere.

\begin{table*}
    \caption{MsgLB response quality and stored checkpoint size across compression settings}
    \label{tab:compression_grid}
    \fontsize{9}{10}\selectfont
    \renewcommand{\arraystretch}{1}
    \setlength{\tabcolsep}{11pt}
    \centering
        \begin{tabular}{cccccccc}
            \toprule
            \multirow{2}{*}{\textbf{Drop ratio $\rho_{\mathrm{drop}}$}} &
            \multirow{2}{*}{\textbf{Depth $d$}} &
            \multicolumn{2}{c}{\textbf{CIDERS (ours)}} &
            \multicolumn{2}{c}{\textbf{FedBiOT}} &
            \multicolumn{2}{c}{\textbf{FedOT}} \\
            \cmidrule(lr){3-4}\cmidrule(lr){5-6}\cmidrule(l){7-8}
            & & \textbf{NLL $\downarrow$} & \textbf{Size (MiB)}
            & \textbf{NLL $\downarrow$} & \textbf{Size (MiB)}
            & \textbf{NLL $\downarrow$} & \textbf{Size (MiB)} \\
            \midrule
            \multirow{3}{*}{0.2}
            & 2 & \textbf{2.0657} & 14.570 & 2.3576 & 14.568 & 2.2375 & 14.568 \\
            & 4 & \textbf{2.0600} & 14.570 & 2.4090 & 14.568 & 2.2671 & 14.568 \\
            & 6 & \textbf{1.7537} & 15.203 & 1.9454 & 15.201 & 1.8845 & 15.201 \\
            \midrule
            \multirow{3}{*}{0.5}
            & 2 & \textbf{6.5222} & 9.502 & 7.4351 & 9.501 & 7.0028 & 9.501 \\
            & 4 & \textbf{5.1666} & 10.136 & 6.2889 & 10.134 & 6.3932 & 10.134 \\
            & 6 & \textbf{4.1151} & 10.769 & 5.2901 & 10.768 & 5.6596 & 10.768 \\
            \bottomrule
        \end{tabular}
\end{table*}

\begin{table}
    \centering
    \caption{FP32 messenger and consensus communication payload over 150 rounds and five clients}
    \label{tab:communication_cost}
    \fontsize{8.5}{10}\selectfont
    \renewcommand{\arraystretch}{1.12}
    \setlength{\tabcolsep}{1.4pt}
    \begin{tabular*}{\columnwidth}{
        @{\extracolsep{\fill}}llrrr@{}
    }
        \toprule
        \textbf{Backbone} & \textbf{Method}
        & \textbf{Uplink} & \textbf{Downlink} & \textbf{Total} \\
        & & \textbf{(MiB)} & \textbf{(MiB)} & \textbf{(GiB)} \\
        \midrule
        Qwen2.5-3B
        & CIDERS
        & 4.758 & 9.516 & 10.45
        \\
        Qwen2.5-3B
        & FedBiOT / FedOT
        & 4.758 & 4.758 & 6.97
        \\
        \midrule
        Qwen2-1.5B
        & CIDERS
        & 3.773 & 7.547 & 8.29
        \\
        Qwen2-1.5B
        & FedBiOT / FedOT
        & 3.773 & 3.773 & 5.53
        \\
        \bottomrule
    \end{tabular*}
\end{table}

\subsection{Personalized Messenger Evaluation}
\label{subsec:native_personalization}

Protocol~B evaluates deployed client models on their own domains and on shared public tasks. Table~\ref{tab:ciders-client-algorithm-results} compares three methods under the same five-domain compressed configuration. Pretrained Full Qwen and its untrained pruned counterpart provide reference points before domain training. CIDERS lowers own-domain NLL relative to FedBiOT on all five clients. The size of the gain varies substantially across domains. Customer support accounts for 78.2\% of the sum of the five client-level NLL reductions. The comparison with code helps put this concentration in context. Before domain training, pruning raises NLL by similar amounts on customer support and code, at 0.4924 and 0.4622, respectively. After training, CIDERS reduces NLL relative to FedBiOT by 0.3184 on customer support and by 0.0062 on code. The initial pruning gap therefore provides little guidance about where CIDERS gains most over FedBiOT in this comparison. The public-task results also show the limits of improved domain fit. CIDERS reaches 2.90\% accuracy on GSM8K, close to FedBiOT at 2.99\% and well below Full Qwen at 28.20\%. Better fit to in-domain responses can thus coexist with a substantial accuracy gap on a shared reasoning benchmark.

The messenger-swap test examines a more specific question. For each target domain, it compares the matched messenger with the four messengers trained on other domains. The matched messenger ranks first in four of the five domains. Its mean advantage over the unmatched messengers is only $3.08\times10^{-4}$ in NLL, equivalent to 0.0176\% of the mean NLL with the matched messenger. The ranking shows how often the matched messenger is preferred. The margin shows how much that choice changes response likelihood on average. This distinction also matters for interpreting the upper-level objective in Eq.~\eqref{upperlevel}. The swap test compares messengers after training, whereas the objective evaluates client loss after local adaptation from a global initialization. A direct estimate of adaptation gain would score each client immediately before and after adaptation from its global checkpoint. Both scores would use the same backbone and evaluation set, with the local data and update budget matched across methods.

\subsection{Compression and Deployment Tradeoffs}
\label{subsec:decomposition}
\label{subsec:cost_robustness}

Table~\ref{tab:compression_grid} compares response NLL and serialized adapter size on Qwen2-1.5B at messenger depths of two, four, and six and nominal drop ratios of 0.2 and 0.5. For CIDERS-FO and FedBiOT, we evaluate the final global AdapEmu/MsgLB checkpoint on each domain. We average the valid response-token losses within each example. We then give examples equal weight within each domain and average the domain scores equally. The original FedOT evaluator is unavailable, so we include its recorded scores for reference. Checkpoint size covers the serialized messenger and learnable-backbone LoRA tensors and their metadata. With a six-layer messenger, increasing the drop ratio from 0.2 to 0.5 reduces the learnable backbone from 18 to 11 layers. The stored adapter checkpoint shrinks by 29.2\%, and CIDERS's response NLL rises from 1.7537 to 4.1151. CIDERS achieves lower NLL than FedBiOT in all six matched configurations. At the lighter drop ratio, however, six-layer FedBiOT reaches an NLL of 1.9454. Heavier pruning can therefore outweigh the advantage CIDERS achieves at matched architectures. The joint-gap bound in Theorem~\ref{thm:ciders-main} applies to a fixed architecture and its corresponding lower-level reference solution. The compression sweep evaluates the response-quality consequences of changing that architecture. The sweep also distinguishes two types of depth comparisons. At a drop ratio of 0.2, the two- and four-layer messengers both retain 23 layers in total and use 3,791,872 LoRA parameters across the messenger and backbone. Their NLLs are close at 2.0657 and 2.0600. The six-layer configuration reaches 1.7537, but it also retains 24 layers in total. That comparison changes the allocation of layers and the total retained depth together.

Table~\ref{tab:communication_cost} reports the recurrent messenger and consensus tensor budget in Algorithm~\ref{alg:ciders_workflow}. This budget depends on the messenger parameter count $P_M$. Two configurations with the same total adapter size can therefore have different recurrent payloads if they allocate different numbers of LoRA parameters to the messenger. An FP32 messenger-sized state occupies $s=4P_M/2^{20}$ MiB. CIDERS uploads one such tensor, $\Delta c_i$, and the server reconstructs the local messenger update from it. The uplink therefore has the same tensor size as in FedBiOT and FedOT. The extra cost comes from the downlink, where CIDERS sends both the global messenger and the consensus state. CIDERS exchanges $3s$ per client-round in the counted tensors, compared with $2s$ for FedBiOT and FedOT. This adds 50\% to the counted bidirectional payload. With five client exchanges per round over 150 rounds, the cumulative budget is $150\times5\times(U+D)/1024$ GiB, where $U$ and $D$ are the unrounded uplink and downlink volumes in MiB. An end-to-end traffic estimate would also account for learnable-backbone transfers, initial distribution, and transport overhead.

\section{Conclusion}

This paper addresses the tension between global consensus and client personalization in cloud-edge LLM systems through a personalized bilevel optimization framework. CIDERS decouples each model into a learnable backbone and a messenger, and embeds consensus-variate correction in local updates. Theoretical analysis characterizes the local trajectory and provides a convergence guarantee. Across Qwen2.5-3B and Qwen2-1.5B, CIDERS improves most evaluated tasks relative to the compressed baselines and provides an explicit trade-off among response quality, stored adapter size, and the counted communication payload. In the five-domain evaluation, own-domain NLL improves for every client, with most of the gain concentrated in customer support. These improvements coexist with a substantial GSM8K gap to the full reference and a small mean matched-domain advantage. In the additional five-domain comparison, CIDERS reaches FedBiOT's final training-probe NLL earlier in communication rounds and achieves lower final held-out response NLL across all five domains. The diagnostics characterize local-component retention in the constructed gradient field. Future work will extend CIDERS to world models for embodied intelligence and physical-world simulation.

\bibliographystyle{IEEEtran}
\bibliography{reference}

\clearpage          
\onecolumn          
\appendix[Detailed Derivations of Theoretical Analysis]

\makeatletter
\fontsize{11}{16}\selectfont

\setlength{\textwidth}{6.8in}
\setlength{\oddsidemargin}{\dimexpr(\paperwidth-\textwidth)/2-1in\relax}
\setlength{\evensidemargin}{\oddsidemargin}

\setlength{\hsize}{\textwidth}
\setlength{\linewidth}{\textwidth}
\setlength{\columnwidth}{\textwidth}
\makeatother

\allowdisplaybreaks

\subsection{Detailed Derivations in the Section of \textit{Geometry of the Trajectory}}

We provide the complete algebraic derivations for the geometric characterization of the local messenger trajectory presented in Section \ref{sec:theoretical_tradeoff}.

\subsubsection{Derivation of Approximation A1}

Recall (A1) the total inner-loop displacement is sufficiently small, i.e., $K \eta \ll 1 / L$, so that $\bar{g}_i^{(t)} \approx \tilde{g}_i^{(t)}$. The path error $\phi_i^{(t)}:=\bar{g}_i^{(t)}-\tilde{g}_i^{(t)}$ is bounded by $O\left(L K \eta G_{\mathrm{bd}}\right)$.  We now derive this bound.
Let the meta-gradient map $g(x):=\nabla \mathcal{L}_i(x)$ be $L$-Lipschitz continuous $\|g(x)-g(y)\| \leq L\|x-y\|, \quad \forall x, y$. 
This is a standard smoothness assumption on the local meta-loss. Let $G_{\mathrm{bd}}:=\sup _x\|g(x)\|$ be a uniform bound on the gradient norm. From the inner update (1) in the main text, we have:
\begin{equation}
	x_{k+1}-\bar{v}=\rho\left(x_k-\bar{v}\right)-\eta g_{i, k}\left(x_k\right)-\eta \gamma\left({c}-{c}_i\right), \quad \rho:=1-\eta \varepsilon.
\end{equation}
For $\varepsilon \geq 0$, we have $\rho \in(0,1]$. Taking norms and applying the triangle inequality:	
\begin{equation}
	\begin{aligned}
		\left\|x_{k+1}-\bar{v}\right\| & \leq \rho\left\|x_k-\bar{v}\right\|+\eta\left\|g_{i, k}\left(x_k\right)\right\|+\eta \gamma\left\|{c}-{c}_i\right\| \\
		& \leq \rho\left\|x_k-\bar{v}\right\|+\eta G_{\mathrm{bd}}+\eta \gamma C_c=\rho\left\|x_k-\bar{v}\right\|+\eta G_{\mathrm{bd}}^{\prime} .
	\end{aligned}
\end{equation}	
where $C_c:=\sup _t\|{c}^{(t)}-{c}_i^{(t)}\|$ is finite under the bounded gradient assumption, we further unrolling to derive: 
\begin{equation}
	\left\|x_K-\bar{v}\right\| \leq \eta G_{\mathrm{bd}}^{\prime} \sum\nolimits_{j=0}^{K-1} \rho^{K-1-j}=\eta G_{\mathrm{bd}}^{\prime} S_K\leq K \eta G_{\mathrm{bd}}^{\prime}.
\end{equation}
Now we bound the path error $\phi_i^{(t)}:=\bar{g}_i^{(t)}-\tilde{g}_i^{(t)}$. Recall:
\begin{equation}
	\bar{g}_i^{(t)}=\frac{1}{S_K} \sum\nolimits_{j=0}^{K-1} \rho^{K-1-j} g_{i, j}\left(x_j\right), \quad \tilde{g}_i^{(t)}=g_{i, 0}(\bar{v}),
\end{equation}
and we can derive the path error $\phi_i^{(t)}$ as follows
\begin{equation}\label{eq:phi}
	\begin{aligned}
		\|\bar{g}_i^{(t)}-\tilde{g}_i^{(t)}\| & =\|\frac{1}{S_K} \sum_{j=0}^{K-1} \rho^{K-1-j}(g_{i, j}(x_j)-g_{i, 0}(\bar{v}))\|\leq \frac{1}{S_K} \sum_{j=0}^{K-1} \rho^{K-1-j}\|g_{i, j}(x_j)-g_{i, 0}(\bar{v})\|.
	\end{aligned}
\end{equation}		
Using the Lipschitz property, it results in 	$\|g_{i, j}(x_j)-g_{i, 0}(\bar{v})\| \leq L\|x_j-\bar{v}\|\leq  K \eta G_{\mathrm{bd}}^{\prime}$, which we	substitute into  \eqref{eq:phi}  and it leads to $\|\phi_i^{(t)}\| \leq L C_K \eta G_{\mathrm{bd}}$.

\subsubsection{Derivation of Approximation A2}

A2 posits that the heterogeneity gap remains constant approximately across rounds, i.e., $\delta_i^{(t+1)}\approx\delta_i^{(t)}\approx\delta_i$. While this holds exactly for quadratic meta-losses, for general nonconvex losses it requires justification. Here, we demonstrate this  locally valid via a second-order Taylor expansion of the loss function around the current broadcast point, and show that the approximation error decays naturally as the global model converges.

Assume that each local loss function $\mathcal{L}_i$ has an $M$-Lipschitz continuous Hessian, i.e., 
\begin{equation}
	\left\|\nabla^2 \mathcal{L}_i\left(w_1\right)-\nabla^2 \mathcal{L}_i\left(w_2\right)\right\| \leq M\left\|w_1-w_2\right\|, \quad \forall w_1, w_2 .
\end{equation}	
Expanding the gradient $\nabla \mathcal{L}_i$ around the broadcast point $\bar{v}^{(t)}$ and evaluating at an arbitrary nearby point $w$, we obtain:
\begin{equation}\label{eq:ri}
	\nabla \mathcal{L}_i(w)=\nabla \mathcal{L}_i(\bar{v}^{(t)})+H_i^{(t)}(w-\bar{v}^{(t)})+\mathcal{R}_i^{(t)}(w), 
\end{equation}	
where $H_i^{(t)}:=\nabla^2 \mathcal{L}_i(\bar{v}^{(t)})$ is the local Hessian at the broadcast point, and the remainder term satisfies
\begin{equation}
	\|\mathcal{R}_i^{(t)}(w)\| \leq \frac{M}{2}\|w-\bar{v}^{(t)}\|^2 . 
\end{equation}	
Similarly, for the global average gradient $\nabla \overline{\mathcal{L}}:=\sum_j p_j \nabla \mathcal{L}_j$, we have:
\begin{equation}\label{eq:rbar}
	\nabla \overline{\mathcal{L}}(w)=\nabla \overline{\mathcal{L}}(\bar{v}^{(t)})+\bar{H}^{(t)}(w-\bar{v}^{(t)})+\overline{\mathcal{R}}^{(t)}(w), 
\end{equation}	
where $\bar{H}^{(t)}:=\sum_j p_j H_j^{(t)}$ is the average Hessian, and the global remainder $\overline{\mathcal{R}}^{(t)}(w)$ satisfies
\begin{equation}
	\|\overline{\mathcal{R}}^{(t)}(w)\| \leq \frac{M}{2}\|w-\bar{v}^{(t)}\|^2.
\end{equation}
Our objective is to characterize the change in $\delta_i$ between rounds $t$ and $t+1$. Let $\bar{v}^{(t+1)}=\bar{v}^{(t)}+$ $\Delta \bar{v}^{(t)}$, where $\Delta \bar{v}^{(t)}$ denotes the global messenger update from round $t$ to $t+1$. Applying the expansions (\ref{eq:ri}) and (\ref{eq:rbar}) at $w=\bar{v}^{(t+1)}$, we obtain:	
\begin{equation}
	\begin{aligned}
		\delta_i^{(t+1)}= & \nabla \mathcal{L}_i(\bar{v}^{(t+1)})-\nabla \overline{\mathcal{L}}(\bar{v}^{(t+1)}) \\
		= & {[\nabla \mathcal{L}_i(\bar{v}^{(t)})+H_i^{(t)} \Delta \bar{v}^{(t)}+\mathcal{R}_i^{(t)}(\bar{v}^{(t+1)})] } \\
		& -[\nabla \overline{\mathcal{L}}(\bar{v}^{(t)})+\bar{H}^{(t)} \Delta \bar{v}^{(t)}+\overline{\mathcal{R}}^{(t)}(\bar{v}^{(t+1)})] \\
		= & \delta_i^{(t)}+(H_i^{(t)}-\bar{H}^{(t)}) \Delta \bar{v}^{(t)}+(\mathcal{R}_i^{(t)}(\bar{v}^{(t+1)})-\overline{\mathcal{R}}^{(t)}(\bar{v}^{(t+1)})) .
	\end{aligned}
\end{equation}	
Therefore, the cross-round variation of the heterogeneity gap is:
\begin{equation}
	\delta_i^{(t+1)}-\delta_i^{(t)}=(H_i^{(t)}-\bar{H}^{(t)}) \Delta \bar{v}^{(t)}+\mathcal{O}(\|\Delta \bar{v}^{(t)}\|^2).
\end{equation}
To justify treating $\delta_i$ as constant across rounds, we must show that the right-hand side of (B.18) is negligible at the scale of our analysis. 	First-order term (Hessian mismatch contribution): Define the worst-case Hessian heterogeneity as
\begin{equation}
	\chi^{(t)}:=\max _i\|H_i^{(t)}-\bar{H}^{(t)}\| . 
\end{equation}	
Then the first-order term is bounded by:
\begin{equation}
	\|(H_i^{(t)}-\bar{H}^{(t)}) \Delta \bar{v}^{(t)}\| \leq \chi^{(t)} \cdot\|\Delta \bar{v}^{(t)}\| .
\end{equation}
Second-order remainder can be simply bounded by $\mathcal{O}(\|\Delta \bar{v}^{(t)}\|^2)$. Therefore, we have $\|\delta_i^{(t+1)}-\delta_i^{(t)}\| \leq \mathcal{O}(\|\Delta \bar{v}^{(t)}\|^2)\rightarrow 0$   as $t \rightarrow \infty$.  

\subsubsection{Residual Recurrence}
We first derive the consensus-gap recursion, recall the  consensus updates are:
\begin{equation}\label{eq:candci}
	{c}^{(t+1)}=(1-\alpha) {c}^{(t)}+\alpha \sum\nolimits_{j \in \mathcal{S}_t} p_j \Delta_{j, t}, \;\text{and } {c}_i^{(t+1)}=[1+\alpha(\gamma-1)] {c}_i^{(t)}+\alpha \Delta_{i, t}-\alpha \gamma {c}^{(t)} .
\end{equation}
Under (A5) that  $\mathcal{S}_t=[N]$, we can derive according to \eqref{eq:candci}	
\begin{equation}
	\begin{aligned}
		{c}^{(t+1)}-{c}_i^{(t+1)}= & (1-\alpha) {c}^{(t)}+\alpha \Delta_{\mathrm{avg}}^{(t)}-[1+\alpha(\gamma-1)] {c}_i^{(t)}-\alpha \Delta_{i, t}+\alpha \gamma {c}^{(t)} \\
		= & (1-\alpha+\alpha \gamma) {c}^{(t)}-(1-\alpha+\alpha \gamma) {c}_i^{(t)}+\alpha(\Delta_{\mathrm{avg}}^{(t)}-\Delta_{i, t}) \\
		= & {[1-\alpha(1-\gamma)]({c}^{(t)}-{c}_i^{(t)})+\alpha(\Delta_{\mathrm{avg}}^{(t)}-\Delta_{i, t}), }
	\end{aligned}
\end{equation}	
where $\Delta_{\text {avg }}^{(t)}:=\sum_j p_j \Delta_{j, t}$. Thus,	the consensus-gap recursion can be derived 
\begin{equation}\label{eq:controlgap}
	{c}^{(t+1)}-{c}_i^{(t+1)}=[1-\alpha(1-\gamma)]({c}^{(t)}-{c}_i^{(t)})+\alpha(\Delta_{\mathrm{avg}}^{(t)}-\Delta_{i, t}).
\end{equation}
Next, we take average over $c_i$ in \eqref{eq:candci} and recall the definition $\overline{{c}}^{(t)}:=\sum_j p_j {c}_j^{(t)}$, we have
\begin{equation}
	\overline{{c}}^{(t+1)}=[1-\alpha(1-\gamma)] \overline{{c}}^{(t)}+\alpha \Delta_{\mathrm{avg}}^{(t)}-\alpha \gamma {c}^{(t)}=(1-\alpha) {c}^{(t)}+\alpha \Delta_{\text {avg }}^{(t)}.
\end{equation}
Subsequently, comparing to $c$ update in  \eqref{eq:candci}, we have $\overline{{c}}^{(t)}={c}^{(t)}$ if it initializes  $\overline{{c}}^{(0)}={c}^{(0)}$.  Next, we consider $\Delta_{\mathrm{avg}}^{(t)}$. From $\Delta_{i, t}=\tilde{g}_i^{(t)}+\gamma({c}^{(t)}-{c}_i^{(t)})$, the average displacement is
\begin{equation}
	\Delta_{\mathrm{avg}}^{(t)}=\sum\nolimits_j p_j(\tilde{g}_j^{(t)}+\gamma({c}^{(t)}-{c}_j^{(t)}))=\widetilde{G}^{(t)}.
\end{equation}	
Therefore, we obtain the following 
\begin{equation}\label{eq:deltaavg}
	\begin{aligned}
		\Delta_{\text {avg }}^{(t)}-\Delta_{i, t} & =\widetilde{G}^{(t)}-(\widetilde{g}_i^{(t)}+\gamma({c}^{(t)}-{c}_i^{(t)}))=-\delta_i^{(t)}-\gamma({c}^{(t)}-{c}_i^{(t)}) .
	\end{aligned}
\end{equation}	
Recall the reference field $u_i^{(t)}:=\tilde{g}_i^{(t)}+\gamma({c}^{(t)}-{c}_i^{(t)})$ and the residual relative to the global consensus direction satisfies $r_i^{(t)}:=u_i^{(t)}-\widetilde{G}^{(t)}=\delta_i^{(t)}+\gamma({c}^{(t)}-{c}_i^{(t)})$, which also lead to 	$
\gamma({c}^{(t)}-{c}_i^{(t)})=r_i^{(t)}-\delta_i^{(t)} .
$ 	Under (A2) that $\delta_i^{(t)}=\delta_i$, we  Substituting (A.16) into (A.14):
\begin{equation}
	\begin{aligned}
		&{c}^{(t+1)}-{c}_i^{(t+1)}  =[1-\alpha(1-\gamma)]\left({c}^{(t)}-{c}_i^{(t)}\right)-\alpha \delta_i-\alpha \gamma\left({c}^{(t)}-{c}_i^{(t)}\right) \\
		&\quad\quad =[1-\alpha(1-\gamma)-\alpha \gamma]\left({c}^{(t)}-{c}_i^{(t)}\right)-\alpha \delta_i =(1-\alpha)\left({c}^{(t)}-{c}_i^{(t)}\right)-\alpha \delta_i.
	\end{aligned}
\end{equation}	
Now compute $r_i^{(t+1)}$ with the result that $\gamma({c}^{(t)}-{c}_i^{(t)})=r_i^{(t)}-\delta_i^{(t)}$:
\begin{equation}\label{eq:r}
	\begin{aligned}
		r_i^{(t+1)} & =\delta_i+\gamma\left({c}^{(t+1)}-{c}_i^{(t+1)}\right) \\
		& =\delta_i+\gamma\left[(1-\alpha)\left({c}^{(t)}-{c}_i^{(t)}\right)-\alpha \delta_i\right] \\
		& =\gamma(1-\alpha)\left({c}^{(t)}-{c}_i^{(t)}\right)+(1-\alpha \gamma) \delta_i\\
		&=(1-\alpha)\left(r_i^{(t)}-\delta_i\right)+(1-\alpha \gamma) \delta_i\\
		&=(1-\alpha) r_i^{(t)}+\alpha(1-\gamma) \delta_i .
	\end{aligned}
\end{equation}
Since the initial condition follows from (A4) that ${c}^{(0)}={c}_i^{(0)}=0$, $r_i^{(0)}=\delta_i^{(0)}+\gamma(0-0)=\delta_i^{(0)}=\delta_i$ holds.  Thus $r_i^{(t+1)}=(1-\alpha) r_i^{(t)}+\alpha(1-\gamma) \delta_i, \quad r_i^{(0)}=\delta_i .$

\subsubsection{Derivation of Proposition \ref{prop:1}}
We first solve the recurrence for $\psi_t$.  Since the recurrence (\ref{eq:r}) expresses $r_i^{(t+1)}$ as an $\mathbb{R}$-linear combination of $r_i^{(t)}$ and $\delta_i$, and $r_i^{(0)}=$ $\delta_i$, it follows by induction that
\begin{equation}
	r_i^{(t)} \in \operatorname{span}\left\{\delta_i\right\} \quad \forall t \geq 0 .
\end{equation}	
Thus we can write $r_i^{(t)}=\psi_t \delta_i$ with a scalar $\psi_t$. Substituting into (\ref{eq:r}) leads to
\begin{equation}
	\psi_{t+1} \delta_i=(1-\alpha) \psi_t \delta_i+\alpha(1-\gamma) \delta_i .
\end{equation}	
Since $\delta_i \neq 0$, we can cancel it which results in
\begin{equation}
	\psi_{t+1}=(1-\alpha) \psi_t+\alpha(1-\gamma), \quad \psi_0=1 . 
\end{equation}	
We solve this nonhomogeneous first-order linear recurrence 	
and the constant steady state $\psi^*$
\begin{equation}\label{eq:psi}
	\psi_t=(1-\gamma)+\gamma(1-\alpha)^t\text{ with } \psi^*=1-\gamma. 
\end{equation}	
Substituting (\ref{eq:psi}) into the definition of the reference field we have
\begin{equation}\label{eq:ui}
	\begin{aligned}
		u_i^{(t)} & =\widetilde{G}^{(t)}+r_i^{(t)}=\psi_t \widetilde{g}_i+\left(1-\psi_t\right) \widetilde{G} .
	\end{aligned}
\end{equation}
This completes the derivation of the proposition \ref{prop:1}.

\subsubsection{Path-Averaging Error and Its Effect on the Residual}
We now analyze the impact of relaxing the short-inner-loop assumption~(A1).
In the implemented local update, the along-path average \(\bar{g}_i^{(t)}\) may differ from the starting oracle \(\widetilde{g}_i^{(t)}\), thus we  introduce the path error $\phi_i^{(t)}:=\bar{g}_i^{(t)}-\widetilde{g}_i^{(t)}$.	As established in Appendix~B.1, this error is uniformly bounded, i.e., $\|\phi_i^{(t)}\|\le L C_K\eta G_{\mathrm{bd}}$.
We quantify how this path error propagates into the residual dynamics and the reference direction. We retain (A2)--(A5) and thus \(\sigma_K=1\),  \(\varepsilon=0\) and  \(\phi_i^{(t)}\ne 0\), then
\begin{equation}
	\Delta_{i,t}
	=\bar{g}_i^{(t)}+\gamma\bigl({c}^{(t)}-{c}_i^{(t)}\bigr)
	=\widetilde{g}_i^{(t)}+\phi_i^{(t)}+\gamma\bigl({c}^{(t)}-{c}_i^{(t)}\bigr),
\end{equation}
and we average over clients, it yields $\Delta_{\mathrm{avg}}^{(t)}=\widetilde{G}^{(t)}+\bar{\phi}^{(t)}$,	where \(\bar{\phi}^{(t)}:=\sum_j p_j\phi_j^{(t)}\).	Subsequently, we evaluate the displacement gap as follows
\begin{equation}
	\begin{aligned}
		\Delta_{\mathrm{avg}}^{(t)}-\Delta_{i,t}
		&=\bigl(\widetilde{G}^{(t)}+\bar{\phi}^{(t)}\bigr)
		-\bigl(\widetilde{g}_i^{(t)}+\phi_i^{(t)}\bigr)
		+\gamma\bigl({c}^{(t)}-\bar{{c}}^{(t)}\bigr)
		-\gamma\bigl({c}^{(t)}-{c}_i^{(t)}\bigr) \\
		&=-\delta_i^{(t)}+\bar{\phi}^{(t)}-\phi_i^{(t)}
		-\gamma\bigl({c}^{(t)}-{c}_i^{(t)}\bigr)
	\end{aligned}
\end{equation}
where \(\widetilde{G}^{(t)}-\widetilde{g}_i^{(t)}=-\delta_i^{(t)}\) and  \(\bar{{c}}^{(t)}={c}^{(t)}\). Substituting into the consensus-gap recursion in \eqref{eq:controlgap} gives
\begin{equation}
	\begin{aligned}
		{c}^{(t+1)}-{c}_i^{(t+1)}
		&=\bigl[1-\alpha(1-\gamma)\bigr]\bigl({c}^{(t)}-{c}_i^{(t)}\bigr)
		+\alpha\bigl[-\delta_i^{(t)}+\bar{\phi}^{(t)}-\phi_i^{(t)}
		-\gamma\bigl({c}^{(t)}-{c}_i^{(t)}\bigr)\bigr] \\
		&=\bigl[1-\alpha(1-\gamma)-\alpha\gamma\bigr]\bigl({c}^{(t)}-{c}_i^{(t)}\bigr)
		+\alpha\bigl(-\delta_i^{(t)}+\bar{\phi}^{(t)}-\phi_i^{(t)}\bigr) \\
		&=(1-\alpha)\bigl({c}^{(t)}-{c}_i^{(t)}\bigr)
		+\alpha\bigl(-\delta_i^{(t)}+\bar{\phi}^{(t)}-\phi_i^{(t)}\bigr).
	\end{aligned}
\end{equation}
Now let consider the residual, which satisfies \(r_i^{(t+1)}=\delta_i^{(t+1)}+\gamma\bigl({c}^{(t+1)}-{c}_i^{(t+1)}\bigr)\).
Under~(A2), we have the approximation \(\delta_i^{(t+1)}=\delta_i^{(t)}=\delta_i\), and therefore it yields 
\begin{equation}
	\begin{aligned}
		r_i^{(t+1)}
		&=\delta_i+\gamma\Bigl[(1-\alpha)\bigl({c}^{(t)}-{c}_i^{(t)}\bigr)
		+\alpha\bigl(-\delta_i+\bar{\phi}^{(t)}-\phi_i^{(t)}\bigr)\Bigr] \\
		&=\gamma(1-\alpha)\bigl({c}^{(t)}-{c}_i^{(t)}\bigr)
		+(1-\alpha\gamma)\delta_i
		+\alpha\gamma\bigl(\bar{\phi}^{(t)}-\phi_i^{(t)}\bigr)\\
		&=(1-\alpha)r_i^{(t)}+\alpha(1-\gamma)\delta_i
		+\alpha\gamma\bigl(\bar{\phi}^{(t)}-\phi_i^{(t)}\bigr),
	\end{aligned}
\end{equation}
where we have used the identity \(\gamma\bigl({c}^{(t)}-{c}_i^{(t)}\bigr)=r_i^{(t)}-\delta_i\). Now, let us turn to the ideal recurrence when \(\phi_i^{(t)}=\bar{\phi}^{(t)}=0\) 
\begin{equation}
	r_i^{(t+1)}\big|_{\mathrm{(E0)}}
	=(1-\alpha)\,r_i^{(t)}\big|_{\mathrm{(E0)}}
	+\alpha(1-\gamma)\delta_i,
	\qquad
	r_i^{(0)}\big|_{\mathrm{(E0)}}=\delta_i.
\end{equation}
Subsequently, we define the path-induced deviation $\tilde{r}_i^{(t)}
:=r_i^{(t)}-r_i^{(t)}\big|_{\mathrm{(E0)}}$, which can directly produce the following linear system
\begin{equation}
	\tilde{r}_i^{(t+1)}
	=(1-\alpha)\tilde{r}_i^{(t)}
	+\alpha\gamma\bigl(\bar{\phi}^{(t)}-\phi_i^{(t)}\bigr),
	\text{ with }
	\tilde{r}_i^{(0)}=0.
	\label{eq:r-dev}
\end{equation}
Unrolling~\eqref{eq:r-dev} from $s=0$ to $t-1$, we can obtain the following
\begin{equation}
	\tilde{r}_i^{(t)}
	=\alpha\gamma\sum\nolimits_{s=0}^{t-1}(1-\alpha)^{t-1-s}
	\bigl(\bar{\phi}^{(s)}-\phi_i^{(s)}\bigr).
	\label{eq:r-unroll}
\end{equation}
Taking norms and applying the triangle inequality,
\begin{equation}
	\begin{aligned}
		\|\tilde{r}_i^{(t)}\|
		&\le 2\alpha\gamma R\sum\nolimits_{s=0}^{t-1}(1-\alpha)^{t-1-s}
		=2\alpha\gamma R\cdot\frac{1-(1-\alpha)^t}{\alpha}=2\gamma R\bigl(1-(1-\alpha)^t\bigr),
	\end{aligned}
	\label{eq:r-bound}
\end{equation}
where we have used  \(\|\phi_j^{(s)}\|\le R\) for every client and \(\bar{\phi}^{(s)}\) is a convex combination of $\phi_j$, this leads to 
\(\|\bar{\phi}^{(s)}-\phi_i^{(s)}\|\le 2R\). Therefore
\begin{equation}
	\begin{aligned}
		\|\tilde{r}_i^{(t)}\|
		&\le 2\alpha\gamma R\sum\nolimits_{s=0}^{t-1}(1-\alpha)^{t-1-s}
		=2\alpha\gamma R\cdot\frac{1-(1-\alpha)^t}{\alpha}=2\gamma R\bigl(1-(1-\alpha)^t\bigr).
	\end{aligned}
	\label{eq:r-bound}
\end{equation}
Finally, \(r_i^{(t)}=u_i^{(t)}-\widetilde{G}\) and \(r_i^{(t)}\big|_{\mathrm{(E0)}}=u_i^{(t)}\big|_{\mathrm{(E0)}}-\widetilde{G}\), so the same bound holds for the reference direction:
\begin{equation}
	\bigl\|u_i^{(t)}-u_i^{(t)}\big|_{\mathrm{(E0)}}\bigr\|
	\le 2\gamma R\bigl(1-(1-\alpha)^t\bigr)
	\le 2\gamma R.
	\label{eq:u-bound}
\end{equation}	
This proves the uniform bound 
and the deviation  vanishes as \(K\eta\to 0\), which shows that a small local update step makes the convex-combination characterization  hold with high accuracy.

\subsubsection{Hessian remainder along the local path.}
We now consider the effect of evaluating the meta-gradient at an intermediate point \(x\) along the local trajectory, rather than at the broadcast messenger \(\bar{v}\).
Even when the consensus variates are locked,
\({c}=\widetilde{G}\) and \({c}_i=\widetilde{g}_i\),
the executed field at \(x\) is
\begin{equation}
	\nabla\mathcal{L}_i(x)+\gamma({c}-{c}_i)
	=\nabla\mathcal{L}_i(x)+\gamma\bigl(\widetilde{G}-\widetilde{g}_i\bigr).
\end{equation}
We compare this with the designed convex combination $(1-\gamma)\nabla\mathcal{L}_i(x)+\gamma\nabla\bar{\mathcal{L}}(x).$
Subtracting yields
\begin{equation}
	\begin{aligned}
		&\nabla\mathcal{L}_i(x)+\gamma\bigl(\widetilde{G}-\widetilde{g}_i\bigr)
		-\bigl[(1-\gamma)\nabla\mathcal{L}_i(x)+\gamma\nabla\bar{\mathcal{L}}(x)\bigr]=\gamma\nabla\mathcal{L}_i(x)+\gamma\bigl(\widetilde{G}-\widetilde{g}_i\bigr)
		-\gamma\nabla\bar{\mathcal{L}}(x) \\
		&\qquad\quad=\gamma\bigl[\nabla\mathcal{L}_i(x)-\nabla\bar{\mathcal{L}}(x)
		-(\widetilde{g}_i-\widetilde{G})\bigr]=\gamma\bigl[\bigl(\nabla\mathcal{L}_i(x)-\widetilde{g}_i\bigr)
		-\bigl(\nabla\bar{\mathcal{L}}(x)-\widetilde{G}\bigr)\bigr].
	\end{aligned}
\end{equation}
Now we compute the Taylor expansion of  \(\nabla\mathcal{L}_i(x)\) and $\nabla\bar{\mathcal{L}}(x)$ respectively at \(\bar{v}\)
\begin{equation}
	\begin{aligned}
		\nabla\mathcal{L}_i(x)
		=\widetilde{g}_i+H_i(x-\bar{v})+O\bigl(\|x-\bar{v}\|^2\bigr),\text{ and } \nabla\bar{\mathcal{L}}(x)
		=\widetilde{G}+\bar{H}(x-\bar{v})+O\bigl(\|x-\bar{v}\|^2\bigr),
	\end{aligned}
\end{equation}
where \(H_i:=\nabla^2\mathcal{L}_i(\bar{v})\) and \(\bar{H}:=\sum_j p_j H_j\).
Substituting these expansions produces
\begin{equation}
	\begin{aligned}
		\nabla\mathcal{L}_i(x)+\gamma\bigl(\widetilde{G}-\widetilde{g}_i\bigr)
		&=(1-\gamma)\nabla\mathcal{L}_i(x)+\gamma\nabla\bar{\mathcal{L}}(x) \\
		&\qquad+\gamma\bigl[(H_i-\bar{H})(x-\bar{v})+O\bigl(\|x-\bar{v}\|^2\bigr)\bigr].
	\end{aligned}
\end{equation}
Thus the first-order remainder is exactly \(\gamma(H_i-\bar{H})(x-\bar{v})\). Recall \(\chi:=\max_i\|H_i-\bar{H}\|\) and the remainder is bounded by \(\gamma\chi\|x-\bar{v}\|+O(\gamma\|x-\bar{v}\|^2)\).
Therefore, the curvature mismatch introduces an error that scales with \(\gamma\) and with the local  displacement \(\|x-\bar{v}\|\).

\subsubsection{The case \(\sigma_K\neq 1\): detailed derivation of the renormalized recurrence}
We now relax assumption~(A3) and allow a proximal coefficient \(\varepsilon>0\), which implies \(\rho=1-\eta\varepsilon<1\), and consequently
\begin{equation}
	S_K=\sum\nolimits_{j=0}^{K-1}\rho^{K-1-j}=\frac{1-\rho^K}{1-\rho}<K,
\end{equation}
so that $\sigma_K:=\nicefrac{S_K}{K}\in(0,1).$
The displacement now carries the compression factor \(\sigma_K\):
\begin{equation}
	\Delta_{i,t}
	=\sigma_K\Bigl(\bar{g}_i^{(t)}+\gamma\bigl({c}^{(t)}-{c}_i^{(t)}\bigr)\Bigr)	=\sigma_K\Bigl(\widetilde{g}_i+\gamma\bigl({c}^{(t)}-{c}_i^{(t)}\bigr)\Bigr),
\end{equation}
where under the small local update approximation~(A1) and fixed heterogeneity~(A2), we have used
\(\bar{g}_i^{(t)}=\widetilde{g}_i^{(t)}=\widetilde{g}_i\).
The average displacement can be simply derived  $\Delta_{\mathrm{avg}}^{(t)}=\sigma_K\widetilde{G}$, this leads to
\begin{equation}
	\begin{aligned}
		\Delta_{\mathrm{avg}}^{(t)}-\Delta_{i,t}
		&=\sigma_K\Bigl(\widetilde{G}-\widetilde{g}_i
		-\gamma\bigl({c}^{(t)}-{c}_i^{(t)}\bigr)\Bigr)=\sigma_K\Bigl(-\delta_i-\gamma\bigl({c}^{(t)}-{c}_i^{(t)}\bigr)\Bigr),
	\end{aligned}
\end{equation}
where we substitute into the consensus gap recursion in \eqref{eq:controlgap} as follows 
\begin{equation}\label{eq:cminusci}
	\begin{aligned}
		{c}^{(t+1)}-{c}_i^{(t+1)}
		&=\bigl[1-\alpha(1-\gamma)\bigr]\bigl({c}^{(t)}-{c}_i^{(t)}\bigr)
		+\alpha\sigma_K\Bigl(-\delta_i-\gamma\bigl({c}^{(t)}-{c}_i^{(t)}\bigr)\Bigr) \\
		&=\bigl[1-\alpha(1-\gamma)-\alpha\gamma\sigma_K\bigr]
		\bigl({c}^{(t)}-{c}_i^{(t)}\bigr)
		-\alpha\sigma_K\delta_i.
	\end{aligned}
\end{equation}
Now we substitute \eqref{eq:cminusci} into  \(r_i^{(t+1)}=\delta_i+\gamma\bigl({c}^{(t+1)}-{c}_i^{(t+1)}\bigr)\) and subsequently we have 
\begin{equation}
	\begin{aligned}
		r_i^{(t+1)}
		&=\delta_i+\gamma\Bigl[
		\bigl(1-\alpha(1-\gamma)-\alpha\gamma\sigma_K\bigr)
		\bigl({c}^{(t)}-{c}_i^{(t)}\bigr)
		-\alpha\sigma_K\delta_i
		\Bigr] \\
		&=\gamma\bigl(1-\alpha(1-\gamma)-\alpha\gamma\sigma_K\bigr)
		\bigl({c}^{(t)}-{c}_i^{(t)}\bigr)
		+(1-\alpha\gamma\sigma_K)\delta_i\\
		&=\bigl(1-\alpha(1-\gamma)-\alpha\gamma\sigma_K\bigr)r_i^{(t)}+
		\Bigl[
		-\bigl(1-\alpha(1-\gamma)-\alpha\gamma\sigma_K\bigr)
		+1-\alpha\gamma\sigma_K
		\Bigr]\delta_i\\
		&	=\bigl[1-\alpha+\alpha\gamma(1-\sigma_K)\bigr]r_i^{(t)}
		+\alpha(1-\gamma)\delta_i,
		\qquad
		r_i^{(0)}=\delta_i.
		\label{eq:D9}
	\end{aligned}
\end{equation}
where we have also used  \(\gamma\bigl({c}^{(t)}-{c}_i^{(t)}\bigr)=r_i^{(t)}-\delta_i\) in \eqref{eq:D9}, and we have complete the derivation of (\ref{eq:r_i}) in subsection \ref{sigma}.  	When \(\sigma_K=1\),~\eqref{eq:D9} reduces to the undamped recurrence in proposition \ref{prop:1}, i.e.,  (\ref{eq:r}). Here, 
span invariance is also preserved with  \(r_i^{(0)}=\delta_i\in\operatorname{span}\{\delta_i\}\), since ~\eqref{eq:D9} expresses \(r_i^{(t+1)}\) as an \(\mathbb{R}\)-linear combination of \(r_i^{(t)}\) and \(\delta_i\), so
\begin{equation}
	r_i^{(t)}\in\operatorname{span}\{\delta_i\}
	\quad\text{for all }t\ge 0.
\end{equation}
Writing \(r_i^{(t)}=\psi_t\delta_i\) and substituting into~\eqref{eq:D9} yields
\begin{equation}
	\psi_{t+1}
	=\bigl[1-\alpha+\alpha\gamma(1-\sigma_K)\bigr]\psi_t
	+\alpha(1-\gamma),
	\qquad
	\psi_0=1.
	\label{eq:D10}
\end{equation}
From \eqref{eq:D10}, we can derive $\psi_t$ with $\psi_0=1$ as follows
\begin{equation}
	\begin{aligned}
		\psi_t
		&=\frac{1-\gamma}{1-\gamma(1-\sigma_K)}
		+\Biggl(1-\frac{1-\gamma}{1-\gamma(1-\sigma_K)}\Biggr)\cdot\bigl[1-\alpha+\alpha\gamma(1-\sigma_K)\bigr]^t,
	\end{aligned}
	\label{eq:D12}
\end{equation}
with the steady state of $\psi$ when $t\to\infty$ for \(\alpha>0\)
\begin{equation}
	\psi_\infty
	=\frac{1-\gamma}{1-\gamma(1-\sigma_K)}.
	\label{eq:D11}
\end{equation}
So the approach to~\eqref{eq:D11} is exponential of rate \(\zeta\). From \eqref{eq:D11} we can see $\psi_\infty>1-\gamma$,  thus the residual retains a larger fraction of \(\delta_i\) than the undamped lock-in \(1-\gamma\) in \eqref{eq:psi}.
Then we further derive the steady-state reference field as
\begin{equation}
	u_i^{(\infty)}
	=\widetilde{G}+\psi_\infty\delta_i
	=\psi_\infty\widetilde{g}_i+(1-\psi_\infty)\widetilde{G},
\end{equation}
which lies closer to \(\widetilde{g}_i\) than the undamped counterpart  \(1-\gamma\) in \eqref{eq:ui}.

\subsection{The Proof of Convergence}
We present a complete convergence analysis of the CIDERS algorithm under the personalized bilevel optimization. The analysis establishes a joint convergence rate on the upper-level meta-objective gap and the lower-level approximation error. 

\subsubsection{Problem Statement}
The bilevel optimization problem in CIDERS is simplified for our careful proof. Specifically, we denote the frozen full teacher backbone as $w_T$, learnable student backbone as $w_S$, global personalization messenger as  $v$ and the client $i$'s personalization messenger as $v_i$.  	Each client $i$ holds a private dataset $\mathcal{D}_i$ and the server has access to a public dataset $\mathcal{D}_{\rm pub}$. Based on the global messenger $v$ with a good meta-initialization for rapid client-specific adaptation, the upper-level objective is the meta-learning formulation
\begin{equation}
	\Phi(v) = \nicefrac{1}{N} \sum\nolimits_{i=1}^N \mathcal{L}_i \bigl( \widetilde{v}, \, w_S^*(v) \bigr),
\end{equation}
where $\widetilde{v} = v - \eta \nabla_v \mathcal{L}_i(v, w_S^*(v))$ is the one-step global personalized messenger, and the local client loss $\mathcal{L}_i$ is evaluated on its private data $\mathcal{D}_i$. On the other hand, the server solves the  lower-level problem to obtain a high-quality student backbone $w_S$ that is aligned with both the frozen teacher and the downstream task, conditioned on the current global messenger $v$:
\begin{equation}
	w_S^*(v) = \arg\min\nolimits_{w_S} \mathcal{L}_{\rm KD}(w_S; v),
\end{equation}
where the TAKD loss is evaluated on the public data $\mathcal{D}_{\rm pub}$. The overall problem is the following personalized bi-level optimization:
\begin{equation}
	\begin{aligned}
		\min\nolimits_v\; \Phi(v)+\nicefrac{\epsilon}{2}\left\|v-v^{-}\right\|^2,			\text{ s.t. } w_S = \arg\min\nolimits_{w_S} \mathcal{L}_{\rm KD}(w_S; v). 
	\end{aligned}
\end{equation}

\subsubsection{Algorithm Abstraction}

We simply illustrate the algorithm for further convergence analysis. At the $  t  $-th round, the server maintains the global messenger $v^{(t)}$, the teacher backbone $w_T^{(t)}$, the student backbone $w_S^{(t)}$, and the global consensus variate $  {c}^{(t)}  $. A subset $  \mathcal{S}_t  $ of clients performs local updates in parallel. Each client $  i \in \mathcal{S}_t  $ initializes its local messenger as $  v_{i,0}^{(t)} = v^{(t)}  $ and retrieves its local consensus variate ${c}_i^{(t)}$.
Over $K$ local steps, client $  i  $ computes the personalized meta-gradient with one-step look-ahead:
\begin{equation}
	\widehat{g}_{i,k} = \nabla_v \mathcal{L}_i\bigl( \widetilde{v}_{i,k}^{(t)}, w_S^{(t)}; \xi_{i,k} \bigr),
\end{equation}
where $\widetilde{v}_{i,k}^{(t)} = v_{i,k}^{(t)} - \eta_{\rm inner} \nabla_v \mathcal{L}_i\bigl( v_{i,k}^{(t)}, w_S^{(t)}; \xi_{i,k} \bigr)$. The local messenger is then updated via the consensus-corrected rule:
\begin{equation}
	v_{i,k+1}^{(t)} \leftarrow v_{i,k}^{(t)} - \eta \Bigl[ \widehat{g}_{i,k} + \gamma \bigl( c^{(t)} - c_i^{(t)} \bigr) \Bigr]-\varepsilon \eta \bigl( v_{i,k}^{(t)} - v^{(t)} \bigr).
\end{equation}
After $  K  $ local steps, client $  i  $ computes the averaged displacement $\Delta_{i,t} := \nicefrac{1}{K\eta} \bigl( v^{(t)} - v_{i,K}^{(t)} \bigr)$
and refreshes its local consensus variate by the exponential moving average:
\begin{equation}
	{c}_i^{(t+1)} = \bigl[1 + \alpha(\gamma - 1)\bigr] {c}_i^{(t)} + \alpha \Delta_{i,t} - \alpha \gamma {c}^{(t)}.
\end{equation}
Only the compact increment $  \Delta {c}_i = {c}_i^{(t+1)} - {c}_i^{(t)}  $ is transmitted to the server.
Upon receiving $  \{\Delta {c}_i\}_{i \in \mathcal{S}_t}  $, the server exactly reconstructs the displacements $  \{\Delta_{i,t}\}  $ by inverting the EMA update. It then performs global aggregation and solves the lower-level KD problem:
\begin{equation}
	({c}^{(t+1)},v^{(t+1)}, w_{S}^{(t+1)} )\leftarrow\left\{
	\begin{aligned}
		& (1 - \alpha) {c}^{(t)} + \alpha \sum\nolimits_{i \in \mathcal{S}_t} p_i \Delta_{i,t},\\
		&   v^{(t)} - \tau K \eta \sum\nolimits_{i \in \mathcal{S}_t}p_i \Delta_{i, t},\\
		& \arg\min\nolimits_{w_{S}} \mathcal{L}_{\rm KD}\bigl(w_S ; v^{(t+1)}\bigr).
	\end{aligned}
	\right.
\end{equation}	
After $  T  $ communication rounds, the algorithm outputs the final global messenger $v^{(t)}$ and the student backbone $w_S^{(T)}$. The framework coordinates bi-level optimization through consensus-variate correction, enabling low-communication personalized adaptation while mitigating client drift under heterogeneous data distributions.

\subsubsection{Definitions}
We define the following auxiliary quantities used throughout the analysis:
\begin{itemize}[leftmargin=*]
	\item 	Client drift $  E_t  $ at round $  t  $:
	client drift quantifies the average squared deviation of each client’s local messenger trajectory from the global messenger during the $  K  $ local steps. It is defined as
	\begin{equation}
		E_t := \frac{1}{KN} \sum\nolimits_{k=0}^{K-1} \sum\nolimits_{i=1}^{N} \mathbb{E} \bigl\| v_{i,k}^{(t)} - v^{(t)} \bigr\|^2,
	\end{equation}	
	where $  v^{(t)}  $ is the global messenger and $  v_{i,k}^{(t)}  $ denotes the local messenger trajectory of client $  i  $ after $  k  $ local steps in round $  t  $.
	\item Local consensus lag $  C_t  $ at round $  t  $:
	local consensus lag measures the average squared difference between the local consensus variates and the true personalized meta-gradients evaluated at the global messenger. It is defined as
	\begin{equation}
		C_t := \frac{1}{N} \sum\nolimits_{i=1}^N \mathbb{E} \bigl\| c_i^{(t)} - {\nabla} \Phi \bigl({v}^{(t)},w_S(v^{(t)}) \bigr) \bigr\|^2.
	\end{equation}		
	\item Glocal consensus lag $  \widetilde{C}_t  $ at round $  t  $:
	global consensus lag measures the  deviation of the global consensus variate ${c}^{(t)}$
	from the meta-gradient of the upper-level objective evaluated at the current global messenger and the current student backbone. It is defined as
	\begin{equation}
		\widetilde{C}_t := \frac{1}{N} \sum\nolimits_{i=1}^N \mathbb{E} \bigl\| c^{(t)} - {\nabla} \Phi \bigl({v}^{(t)},w_S(v^{(t)}) \bigr) \bigr\|^2.
	\end{equation}		
	\item Joint quantity of interest:
	this quantity that jointly tracks upper-level meta-suboptimality and lower-level distillation error:
	\begin{equation}
		\Phi(v) - \Phi^* + \bigl\| w_S - w_S^*(v) \bigr\|^2,
	\end{equation}
	where $  w_S^*(v)  $ denotes the exact minimizer of the lower-level task-aware knowledge distillation loss $  \mathcal{L}_{\rm KD}(w_S; v)  $ for a fixed messenger $  v  $.
	\item Full practical meta-gradient: in CIDERS, the full practical meta-gadient at the local messenger $v_{i,k}$ with the current student backbone $w_S$ on the client $i$ and the $k$th local iteration can be computed via
	\begin{equation}
		\widetilde{g}_{i,k}({v}_{i,k}, w_{S}):= \nabla_v \mathcal{L}_i \Bigl( 
		{v}_{i,k} - \eta_{\rm inner} \nabla_v \mathcal{L}_i({v}_{i,k}, w_S), \, w_S \Bigr),
	\end{equation}
	where $w_S$ is the approximate student backbone currently maintained by the server.
	\item Stochastic practical meta-gradient estimator: CIDERS actually compute the meta-gradient at the local messenger $v_{i,k}$ with the current student backbone $w_S$ on the client $i$ and the $k$th local iteration via
	\begin{equation}
		\widehat{{g}}_{i,k}({v}_{i,k}, w_{S}) = \nabla_{v} \mathcal{L}_i ( \widetilde{v}_{i,k}, w_{S}; \xi_{i,k} ), 
	\end{equation} 	
	where the one-step look-ahead point is $\widetilde{v}_{i,k} = v_{i,k} - \eta_{\rm inner} \nabla_{v} \mathcal{L}_i \Bigl(v_{i,k}, w_{S}; \xi_{i,k} \Bigr)$, and it  satisfies $\mathbb{E}\bigl[ \widehat{{g}}_{i,k}({v}_{i,k}, w_{S})\bigr] 
	= 
	\widetilde{g}_{i,k}({v}_{i,k}, w_{S})$. 
	\item Ideal meta-gradient: this is the meta-gradient that would be obtained if the lower-level problem are solved exactly with respect to the current global messenger $v$:
	\begin{equation}
		{g}_{i,k}^{\rm ideal}(v_{i,k}, w_S^*(v)) := \nabla_v \mathcal{L}_i \Bigl( 
		v_{i,k} - \eta_{\rm inner} \nabla_v \mathcal{L}_i(v_{i,k}, w_S^*(v)), \, w_S^*(v) 
		\Bigr),
	\end{equation}
	where $  w_S^*(v)  $ denotes the exact minimizer of the lower-level task-aware knowledge distillation objective:$$w_S^*(v) := \arg\min_{w_S} \mathcal{L}_{\rm KD}(w_S; v).$$
\end{itemize}

\subsubsection{Assumptions}
We state the complete set of assumptions for the convergence analysis. 
\begin{assumption}[Joint Smoothness]
	For each client $  i = 1, \dots, N  $, the loss function $  \mathcal{L}_i(v, w_S)  $ is jointly $  L  $-smooth with respect to the pair of variables $  (v, w_S)  $. That is, for any $  (v, w_S)  $ and $  (v', w_S')  $,
	\begin{equation}
		\|\nabla \mathcal{L}_i(v, w_S) - \nabla \mathcal{L}_i(v', w_S')\| \le L \bigl( \|v - v'\| + \|w_S - w_S'\| \bigr).
	\end{equation}
	Consequently, the upper-level meta-objective $  \Phi(v)  $ is $  \beta  $-smooth, where the effective smoothness constant $  \beta  $ is given by 
	\begin{equation}
		\beta = L(1 + \eta_{\rm inner} L)^2 + \eta_{\rm inner} \rho,
	\end{equation}
	and $  \rho  $ denotes the Lipschitz constant of the Hessian of $  \mathcal{L}_i  $ with respect to the messenger $  v  $. In addition, we assume the lower-level TAKD objective $  \mathcal{L}_{\rm KD}(w_S; v)  $ is $  L_y  $-smooth with respect to  $  w_S  $.
\end{assumption}

\begin{assumption}[Polyak-\L{}ojasiewicz Inequality]
	The upper-level objective satisfies
	$\|\nabla \Phi(v)\|^2 \ge 2\mu \bigl( \Phi(v) - \Phi^* \bigr)$
	for some $  \mu > 0  $.
\end{assumption}

\begin{assumption}[Meta-Gradient Dissimilarity]
	There exist constants $  G, B \geq 0  $ such that, for any local mesenger $v_{i,k}$ on the client $i$ $k$th local iteration, global messenger $v  $ and current student backbone $  w_{S}  $,
	\begin{equation}
		\mathbb{E} \bigl\| \widetilde{g}_{i,k}({v}_{i,k}, w_{S})- \nabla\Phi(v) \bigr\|^2  
		\leq G^2 + B^2 \|\nabla\Phi(v)\|^2.
	\end{equation}
\end{assumption}

\begin{assumption}[Gradient Heterogeneity]
	There exists $  \zeta \ge 0  $ such that
	$\nicefrac{1}{N} \sum\nolimits_{i=1}^N \bigl\| \nabla_v \mathcal{L}_i(v, w_S) - \nabla \Phi(v) \bigr\|^2 \le \zeta^2.$
\end{assumption}

\begin{assumption}[Gradient Variance]
	The stochastic meta-gradients used in the upper-level client updates have bounded variance for a constant $  \zeta^2 \ge 0  $ such that:
	\begin{equation}
		\mathbb{E}_{\xi_{i,k}}  \bigl\| \widetilde{g}_{i,k}\bigl(v_{i,k}^{(t)}, w_S^{(t)}; \xi_{i,k}\bigr) - \nabla_v \mathcal{L}_i\bigl(v_{i,k}^{(t)}, w_S^{(t)}\bigr) \bigr\|^2 \le \zeta^2,
	\end{equation}
	Also for the lower-level TAKD objective, there exists a constant $  \sigma_{\rm KD}^2 \ge 0  $ such that
	\begin{equation}
		\mathbb{E}_{\xi}  \bigl\| \nabla_{w_S} \mathcal{L}_{\rm KD}(w_S; v; \xi) - \nabla_{w_S} \mathcal{L}_{\rm KD}(w_S; v) \bigr\|^2  \le \sigma_{\rm KD}^2.
	\end{equation}
\end{assumption}

\begin{assumption}[Bounded gradients and Hessian]
	There exist nonnegative constants $  G, B  $ such that the following bounds hold uniformly for all clients and for all messenger and backbone parameters $  v, w_S  $, i.e., Gradient bound $\|\nabla_v \mathcal{L}_i(v, w_S)\| \le G$, and Hessian norm bound $\|\nabla_v^2 \mathcal{L}_i(v, w_S)\| \le B.$
\end{assumption}

\begin{assumption}[Hessian Lipschitz Continuity]
	There exists a nonnegative constant $  H  $ such that, for each client $i$ and for any two pairs of parameters $  (v, w_S)  $ and $  (v', w_S')  $, i.e., 
	$\| \nabla_v^2 \mathcal{L}_i(v, w_S) - \nabla_v^2 \mathcal{L}_i(v', w_S') \| \le H ( \|v - v'\| + \|w_S - w_S'\| ).$
\end{assumption}

\subsubsection{Basic Lemmas}

\begin{lemma}[Personalized Meta-Gradient Dissimilarity]\label{lemma:a}
	Under Assps. A1 and A5, for any client $  i  $, messenger $  v  $, and  student backbone  $  w_S  $, let $  g_i(v) = \nabla_v \mathcal{L}_i(v, w_S)  $ denote the primal gradient. Then,
	\begin{equation}
		\left\|\widetilde{g}_i(v)-g_i(v)\right\|^2 \leq L^2 \eta_{\text {inner }}^2\left\|g_i(v)\right\|^2, 
	\end{equation}
	where $  \beta  $ is the effective smoothness constant of the upper-level meta-objective defined in Assps. A1.
\end{lemma}
\begin{proof}
	By Assumption A1 the client loss $L_i$ is jointly $L$-smooth in the pair $\left(v, w_S\right)$. Fixing  $w_S$, it follows  that $L_i\left(\cdot, w_S\right)$ is $L$-smooth in $v$, i.e., $v \mapsto \nabla_v L_i\left(v, w_S\right)$ is $L$-Lipschitz continuous. Hence, we have
	\begin{equation}
		\left\|\nabla_v L_i\left(v-\eta_{\text {inner }} g_i(v), w_S\right)-\nabla_v L_i\left(v, w_S\right)\right\| \leq L\left\|\left(v-\eta_{\text {inner }} g_i(v)\right)-v\right\|\leq L \eta_{\text {inner }}\left\|g_i(v)\right\| .
	\end{equation}
	This completes the proof. 		
\end{proof}

\begin{lemma}[Upper-level Gap Inequality]\label{lemma:2beta}
	Under A2 and A3, for any messenger parameter $  v  $ and any admissible student backbone parameter $  w_S  $,
	\begin{equation}
		\frac{1}{N} \sum_{i=1}^N \mathbb{E} \|\widetilde{g}_i\left(v, w_S\right)\|^2 \leq  (8 L^2 \eta_{\text {inner }}^2+4) \zeta^2+(16 L^2 \eta_{\text {inner }}^2\beta+4\beta)\left(\Phi(v)-\Phi^*\right),
	\end{equation}
	where $  \Phi^* = \min_v \Phi(v)  $. 
\end{lemma}
\begin{proof}
	Fix an arbitrary messenger parameter (v) and an admissible student backbone parameter $w_S$. For notational convenience denote $\widetilde{g}_i:=\widetilde{g}_i\left(v, w_S\right)$ and $g_i:=\nabla_v L_i\left(v, w_S\right)$. All expectations are taken with respect to any stochasticity appearing in the gradient estimators. We first perform the decomposition via the triangle inequality: 
	\begin{equation}
		\begin{aligned}
			\frac{1}{N} \sum\nolimits_{i=1}^N& \mathbb{E}\left\|\widetilde{g}_i\right\|^2 \leq 2 \cdot \frac{1}{N} \sum\nolimits_{i=1}^N \mathbb{E}\left\|\widetilde{g}_i-\nabla \Phi(v)\right\|^2+2\|\nabla \Phi(v)\|^2\\
			&\leq \frac{4}{N} \sum\nolimits_{i=1}^N \mathbb{E}\left\|\widetilde{g}_i-g_i\right\|^2+\frac{4}{N} \sum\nolimits_{i=1}^N \mathbb{E}\left\|g_i-\nabla \Phi(v)\right\|^2+2\|\nabla \Phi(v)\|^2\\
			&\leq L^2 \eta_{\text {inner }}^2 \cdot \frac{4}{N} \sum\nolimits_{i=1}^N E\left\|g_i\right\|^2+4 \zeta^2+2\|\nabla \Phi(v)\|^2\\
			&\leq L^2 \eta_{\text {inner }}^2 \cdot \frac{8}{N} \sum\nolimits_{i=1}^N  \left\|g_i-\nabla \Phi(v)\right\|^2+8L^2 \eta_{\text {inner }}^2\|\nabla \Phi(v)\|^2 +4 \zeta^2+2\|\nabla \Phi(v)\|^2\\ 
			&\leq (8 L^2 \eta_{\text {inner }}^2+4) \zeta^2+(8 L^2 \eta_{\text {inner }}^2+2)\|\nabla \Phi(v)\|^2\\
			&\leq  (8 L^2 \eta_{\text {inner }}^2+4) \zeta^2+(16 L^2 \eta_{\text {inner }}^2\beta+4\beta)\left(\Phi(v)-\Phi^*\right)
		\end{aligned}
	\end{equation}
	where we have performed the decomposition of the dissimilarity term in the second inequality, i.e., $\left\|\widetilde{g}_i-\nabla \Phi(v)\right\|^2 \leq 2\left\|\widetilde{g}_i-g_i\right\|^2+2\left\|g_i-\nabla \Phi(v)\right\|^2$, decomposed via $\left\|g_i\right\|^2 \leq 2\left\|g_i-\nabla \Phi(v)\right\|^2+2\|\nabla \Phi(v)\|^2$ in the fourth inequality, and we have used $\|\nabla \Phi(v)\|^2 \leq 2 \beta(\Phi(v)-\Phi^*)$ in the last inequality.

\end{proof}

\begin{lemma}[Lower-level Approximation Dynamics]\label{lemma:average}
	Suppose it satisfies Assp. A1,  after performing $  E  $ steps of gradient descent on the lower-level task-aware knowledge distillation objective with fixed messenger parameter $  v^{(t)}  $, the expected squared distance to the exact minimizer satisfies
	\begin{equation}
		E\|w_S^{(t)}-w_S^*(v^{(t)})\|^2 \leq 2(1-\mu \eta_{\mathrm{KD}})^E E\|w_S^{(t-1)}-w_S^*(v^{(t-1)})\|^2+\frac{\eta_{\mathrm{KD}} \sigma_{\mathrm{KD}}^2}{\mu}+8 R_w^2(1-\mu \eta_{\mathrm{KD}})^E. 
	\end{equation}			
\end{lemma}
\begin{proof}
	We now explain the proof  by starting from  the deterministic contraction, which can be the foundational inequality. Specifically, 
	because $L_{\mathrm{KD}}\left(\cdot ; v^{(t)}\right)$ is $L_y$-smooth and $\mu$-strongly convex, any deterministic gradient step with step size $\eta_{\mathrm{KD}} \leq 1 / L_y$ satisfies the contraction
	\begin{equation}\label{eq:tmp25}
		\|w_S-\eta_{\mathrm{KD}} \nabla_{w_S} L_{\mathrm{KD}}(w_S ; v^{(t)})-w_S^*(v^{(t)})\|^2 \leq(1-\mu \eta_{\mathrm{KD}})\|w_S-w_S^*(v^{(t)})\|^2.
	\end{equation}

	Next, we consider the one-step recursion with the SGD for the KD objective. Consider the SGD step by adding and subtracting the true gradient:
	\begin{equation}
		w_S^{k+1}-w_S^*(v^{(t)})=(w_S^k-\eta_{\mathrm{KD}} \nabla_{w_S} L_{\mathrm{KD}}(w_S^k ; v^{(t)})-w_S^*(v^{(t)}))-\eta_{\mathrm{KD}}(g(w_S^k ; \xi_k)-\nabla_{w_S} L_{\mathrm{KD}}(w_S^k ; v^{(t)})).
	\end{equation}		
	We take the $L_2-$norm and the conditional expectation with respect to the current mini-batch $\xi_k$ conditioned on all previous randomness, with the notice that the cross term vanishes by unbiasedness of the stochastic gradient. This immediately yields
	\begin{equation}\label{eq:tmp26}
		\mathbb{E}\|w_S^{k+1}-w_S^*(v^{(t)})\|^2 \mid \text { history up to } k \leq\|w_S^k-\eta_{\mathrm{KD}} \nabla_{w_S} L_{\mathrm{KD}}(w_S^k ; v^{(t)})-w_S^*(v^{(t)})\|^2+\eta_{\mathrm{KD}}^2 \sigma_{\mathrm{KD}}^2.
	\end{equation}
	Subsequently, by inserting the deterministic contraction from (\ref{eq:tmp25}) into (\ref{eq:tmp26}) produces the fundamental one-step inequality
	\begin{equation}\label{eq:tmp27}
		\mathbb{E}\|w_S^{k+1}-w_S^*(v^{(t)})\|^2 \mid \text { history up to } k \leq(1-\mu \eta_{\mathrm{KD}})\|w_S^k-w_S^*(v^{(t)})\|^2+\eta_{\mathrm{KD}}^2 \sigma_{\mathrm{KD}}^2.
	\end{equation}
	Unrolling the recurrence \eqref{eq:tmp27} over $E$ steps
	by taking the unconditional expectation  yields the linear recurrence
	\begin{equation}
		\mathbb{E}\|w_S^E-w_S^*(v^{(t)})\|^2 \leq(1-\mu \eta_{\mathrm{KD}})^E \mathbb{E}\|w_S^0-w_S^*(v^{(t)})\|^2+\eta_{\mathrm{KD}}^2 \sigma_{\mathrm{KD}}^2 \sum\nolimits_{j=0}^{E-1}(1-\mu \eta_{\mathrm{KD}})^j .
	\end{equation}		
	Under the assumption that $0<\mu \eta_{\mathrm{KD}}<1$, the geometric sum is bounded by	$\sum_{j=0}^{E-1}(1-\mu \eta_{\mathrm{KD}})^j=\nicefrac{(1-(1-\mu \eta_{\mathrm{KD}})^E)}{\mu \eta_{\mathrm{KD}}} \leq \nicefrac{1}{\mu \eta_{\mathrm{KD}}} .$		Therefore, we have
	\begin{equation}\label{eq:tmp30}
		E\|w_S^{(t)}-w_S^*(v^{(t)})\|^2 \leq(1-\mu \eta_{\mathrm{KD}})^E E\|w_S^{(t-1)}-w_S^*(v^{(t)})\|^2+\frac{\eta_{\mathrm{KD}} \sigma_{\mathrm{KD}}^2}{\mu} .
	\end{equation}		
	At this point we have a bound expressed in terms of the distance to the current minimizer $w_S^*\left(v^{(t)}\right)$, but the warm-start $w_S^{(t-1)}$ was produced with respect to the previous minimizer. By performing the decomposition, we can simply handle the change of target minimizer with the  bi-level coupling
	\begin{equation}\label{eq:tmp28}
		\begin{aligned}
			\|w_S^{(t-1)}-w_S^*(v^{(t)})\|^2 &\leq 2\|w_S^{(t-1)}-w_S^*(v^{(t-1)})\|^2+2\|w_S^*(v^{(t-1)})-w_S^*(v^{(t)})\|^2\\
			&\leq 2 E\left\|w_S^{(t-1)}-w_S^*\left(v^{(t-1)}\right)\right\|^2+8 R_{w}^2
		\end{aligned}
	\end{equation}		
	where for brevity, we assume the uniform boundedness  $\left\|w_S^*(v)\right\| \leq R_w$ that is valid for every messenger, and it converts the movement of the target into an additive term that still contracts geometrically. Finally, we substitute the inequality \eqref{eq:tmp28} into \eqref{eq:tmp30}:
	\begin{equation}
		E\|w_S^{(t)}-w_S^*(v^{(t)})\|^2 \leq 2(1-\mu \eta_{\mathrm{KD}})^E E\|w_S^{(t-1)}-w_S^*(v^{(t-1)})\|^2+\frac{\eta_{\mathrm{KD}} \sigma_{\mathrm{KD}}^2}{\mu}+8 R_w^2(1-\mu \eta_{\mathrm{KD}})^E,
	\end{equation}
	which completes the proof. 
\end{proof}

\subsubsection{Key Lemmas}

With client drift and consensus lag defined, we now present four lemmas that analyze the behavior of the PERCE algorithm. These lemmas address the variance of the server update, the evolution of consensus lag, the bounding of client drift, and the progress made per round.

\begin{lemma}[Local Consensus Deviation]\label{lemma:controldeviation}
	For all $  t \geq 0  $ and all clients $  i  $, there exist positive values $(\widetilde{G}^2,\widetilde{B}^2)$ such that 
	\begin{equation}\label{eq:tmp12}
		\begin{aligned}
			\mathbb{E}\|{c}_i^{(t)} - \nabla\Phi(v^{(t)})\|^2& \leq\widetilde{G}^2+\widetilde{B}^2\left\|\nabla\Phi\left(v^{(t)}\right)\right\|^2
		\end{aligned}
	\end{equation}		
\end{lemma}
\begin{proof}
	Let us recall the local update. At round $t$, client $i$ initializes $v_{i,0}^{(t)} = v^{(t)}$, i.e., the global messenger received at the beginning of the round and performs $  K  $ local steps according to
	\begin{equation}
		v_{i,k+1}^{(t)} = v_{i,k}^{(t)} - \eta \Bigl[ \widetilde{g}_{i,k} + \gamma\bigl({c}^{(t)} - {c}_i^{(t)}\bigr) \Bigr] - \varepsilon\eta \bigl(v_{i,k}^{(t)} - v^{(t)}\bigr),
	\end{equation}
	for $  k = 0, \dots, K-1  $, where $  \widetilde{g}_{i,k}  $ is the practical meta-gradient evaluated at the current local point and the current fixed backbone $  w_{\mathcal{B}_L}^{(t)}$. Then, we sum both sides from $  k = 0  $ to $  k = K-1  $		
	\begin{equation}
		v^{(t)}-v_{i, K}^{(t)}=\eta \sum\nolimits_{k=0}^{K-1} \tilde{g}_{i, k}+\eta K \cdot \gamma\left({c}^{(t)}-{c}_i^{(t)}\right)+\eta \varepsilon \sum\nolimits_{k=0}^{K-1}\left(v_{i, k}^{(t)}-v^{(t)}\right).
	\end{equation}
	Define the averaged practical meta-gradient $\overline{\widetilde{g}}_i^{(t)} := \nicefrac{1}{K} \sum_{k=0}^{K-1} \widetilde{g}_{i,k}$, and the averaged proximal residual $r_{\rm prox,i}^{(t)} := \nicefrac{\varepsilon}{K} \sum_{k=0}^{K-1} (v_{i,k}^{(t)} - v^{(t)})$,  then according to $\Delta_{i, t}:=\nicefrac{1}{K \eta}(v^{(t)}-v_{i, K}^{(t)})$,  we have
	\begin{equation}
		\Delta_{i, t}=\overline{\widetilde{g}}_i^{(t)}+\gamma\left({c}^{(t)}-{c}_i^{(t)}\right)+r_{\mathrm{prox}, \mathrm{i}}^{(t)},\text{ and } {c}_i^{(t+1)}=(1-\alpha) {c}_i^{(t)}+\alpha(\overline{\widetilde{g}}_i^{(t)}+r_{\mathrm{prox}, \mathrm{i}}^{(t)}). 
	\end{equation}
	We further define the error vector $e_i^{(t)} := {c}_i^{(t)} - \nabla\Phi(v^{(t)})$,  the driving noise
	$\delta_i^{(t)} := \overline{\widetilde{g}}_i^{(t)} - \nabla\Phi(v^{(t)}) + r_{\rm prox,i}^{(t)}$, and $  \Delta\Phi^{(t)} := \nabla\Phi(v^{(t)}) - \nabla\Phi(v^{(t+1)})  $ for brevity. Subsequently, we can expand
	\begin{equation}\label{eq:tmp8}\begin{aligned}
			\mathbb{E}\|e_i^{(t+1)}\|^2 &= \mathbb{E}\|(1-\alpha)e_i^{(t)} + \alpha\delta_i^{(t)} + \Delta\Phi^{(t)}\|^2 \\
			&= (1-\alpha)^2 \mathbb{E}\|e_i^{(t)}\|^2 + 2\alpha(1-\alpha)\mathbb{E}\langle e_i^{(t)}, \delta_i^{(t)}\rangle+ \alpha^2 \mathbb{E}\|\delta_i^{(t)}\|^2 \\
			&\quad + 2(1-\alpha)\mathbb{E}\langle e_i^{(t)}, \Delta\Phi^{(t)}\rangle + 2\alpha\mathbb{E}\langle \delta_i^{(t)}, \Delta\Phi^{(t)}\rangle+ \mathbb{E}\|\Delta\Phi^{(t)}\|^2.
	\end{aligned}\end{equation}
	For the driving noise term $  \mathbb{E}\|\delta_i^{(t)}\|^2$, through triangle inequality we have
	\begin{equation}\label{eq:tmp9}
		\begin{aligned}
			\mathbb{E}\|\delta_i^{(t)}\|^2 &\leq 2 \mathbb{E}\|\overline{\widetilde{g}}_i^{(t)}-\nabla \Phi(v^{(t)})\|^2+2 \mathbb{E}\|r_{\text {prox }, \mathrm{i}}^{(t)}\|^2\\
			&\leq 2 G^2+2 B^2\|\nabla \Phi(v^{(t)})\|^2+ \frac{2\sigma^2}{K}+2 C_r \varepsilon^2,
		\end{aligned}
	\end{equation}
	where we have used the meta-gradient dissimilarity assumption,  variance of the average over the local $K$ steps, and the proximal term satisfies $  \mathbb{E}\|r_{\rm prox,i}^{(t)}\|^2 \leq C_r \varepsilon^2  $. Next, we bound the three cross terms in (\ref{eq:tmp8}) using Young's inequality with parameter $  \varepsilon_1 > 0  $ as follows:
	\begin{itemize}[leftmargin=*]
		\item For the term $  2\alpha(1-\alpha)\mathbb{E}\langle e_i^{(t)}, \delta_i^{(t)}\rangle  $:
		\begin{equation}
			2 \alpha(1-\alpha) \mathbb{E}\langle e_i^{(t)}, \delta_i^{(t)}\rangle \leq \alpha(1-\alpha) \varepsilon_1 \mathbb{E}\|e_i^{(t)}\|^2+\frac{\alpha(1-\alpha)}{\varepsilon_1} \mathbb{E}\|\delta_i^{(t)}\|^2, 
		\end{equation}
		\item For the term $  2(1-\alpha)\mathbb{E}\langle e_i^{(t)}, \Delta\Phi^{(t)}\rangle  $:
		\begin{equation}
			2(1-\alpha)\mathbb{E}\langle e_i^{(t)}, \Delta\Phi^{(t)}\rangle \leq (1-\alpha)\varepsilon_1 \mathbb{E}\|e_i^{(t)}\|^2 + \frac{1-\alpha}{\varepsilon_1} \mathbb{E}\|\Delta\Phi^{(t)}\|^2.
		\end{equation}			 
		\item For the term $  2\alpha\mathbb{E}\langle \delta_i^{(t)}, \Delta\Phi^{(t)}\rangle  $:
		\begin{equation}
			2\alpha\mathbb{E}\langle \delta_i^{(t)}, \Delta\Phi^{(t)}\rangle \leq \alpha\varepsilon_1 \mathbb{E}\|\delta_i^{(t)}\|^2 + \frac{\alpha}{\varepsilon_1} \mathbb{E}\|\Delta\Phi^{(t)}\|^2.
		\end{equation}		 
	\end{itemize}
	Note by $  \beta  $-smoothness of $  \Phi  $, we have $\mathbb{E}\|\Delta\Phi^{(t)}\|^2 \leq \beta^2 \mathbb{E}\|v^{(t+1)} - v^{(t)}\|^2$, it subsequently leads to
	\begin{equation}\label{eq:tmp10}
		\mathbb{E}\|\Delta\Phi^{(t)}\|^2 \leq\beta^2 \mathbb{E}\left\|v^{(t+1)}-v^{(t)}\right\|^2=\beta^2 \tau^2 K^2 \eta^2 \mathbb{E}\left\|\sum_{i \in \mathcal{S}_t} p_i \Delta_{i, t}\right\|^2 \leq \beta^2 \tau^2 K^2 \eta^2 \widetilde{\Delta}^2,
	\end{equation}
	where we have defined $\widetilde{\Delta^2}:=\sup _t \mathbb{E}\left\|\sum_{i \in \mathcal{S}_t} p_i \Delta_{i, t}\right\|^2$ for simplicity, since $\Delta_{i, t}=\nicefrac{1}{K \eta}(v^{(t)}-v_{i, K}^{(t)})$. Hence, we substitute (\ref{eq:tmp9})-(\ref{eq:tmp10}) into (\ref{eq:tmp8}) and it leads to		
	\begin{equation}\label{eq:tmp11}
		\begin{aligned}
			\mathbb{E}\|e_i^{(t+1)}\|^2 \leq & {\left[(1-\alpha)^2+\alpha(1-\alpha) \varepsilon_1+(1-\alpha) \varepsilon_1\right] \mathbb{E}\|e_i^{(t)}\|^2 } \\
			& +\left[\alpha^2+\frac{\alpha(1-\alpha)}{\varepsilon_1}+\alpha \varepsilon_1\right]\left(2 G^2+2 B^2\left\|\nabla \Phi\left(v^{(t)}\right)\right\|^2+\frac{2 \sigma^2}{K}+2 C_r \varepsilon^2\right) \\
			& +\left[\frac{1-\alpha}{\varepsilon_1}+\frac{\alpha}{\varepsilon_1}+1\right] \beta^2 \tau^2 K^2 \eta^2 \widetilde{\Delta}^2 .
		\end{aligned}
	\end{equation}
	Note (\ref{eq:tmp11}) has the form $V_{t+1} \leq \rho_1 V_t + \rho_2$ for the sequence $V_t$,  thus we can obtain the rsult as follows:
	\begin{equation}\label{eq:tmp12}
		\begin{aligned}
			\mathbb{E}\|{c}_i^{(t)} - \nabla\Phi(v^{(t)})\|^2& \leq \mathbb{E}\|{c}_i^{(0)} - \nabla\Phi(v^{(0)})\|^2+\frac{2 \left[\alpha^2+\frac{\alpha(1-\alpha)}{\varepsilon_1}+\alpha \varepsilon_1\right]B^2 \left\|\nabla \Phi\left(v^{(t)}\right)\right\|^2}{1-\left[(1-\alpha)^2+\alpha(1-\alpha) \varepsilon_1+(1-\alpha) \varepsilon_1\right]}\\
			&+\frac{\left[\alpha^2+\frac{\alpha(1-\alpha)}{\varepsilon_1}+\alpha \varepsilon_1\right]\left(2 G^2+\frac{2 \sigma^2}{K}+2 C_r \varepsilon^2\right)+\left(\frac{1}{\varepsilon_1}+1\right) \beta^2 \tau^2 K^2 \eta^2 \widetilde{\Delta}^2}{1-\left[(1-\alpha)^2+\alpha(1-\alpha) \varepsilon_1+(1-\alpha) \varepsilon_1\right]},
		\end{aligned}
	\end{equation} 
	where it can be seen that  $\widetilde{B}^2$  and $\widetilde{G}^2$ satisfies the following 
	\begin{equation}
		\begin{aligned}
			&\quad\quad\quad\quad\quad\quad\quad\quad\quad\widetilde{B}^2=\frac{2 \left[\alpha^2+\frac{\alpha(1-\alpha)}{\varepsilon_1}+\alpha \varepsilon_1\right]B^2}{1-\left[(1-\alpha)^2+\alpha(1-\alpha) \varepsilon_1+(1-\alpha) \varepsilon_1\right]}, \\
			&\widetilde{G}^2=\mathbb{E}\|{c}_i^{(0)} - \nabla\Phi(v^{(0)})\|^2+\frac{\left[\alpha^2+\frac{\alpha(1-\alpha)}{\varepsilon_1}+\alpha \varepsilon_1\right]\left(2 G^2+\frac{2 \sigma^2}{K}+2 C_r \varepsilon^2\right)+\left(\frac{1}{\varepsilon_1}+1\right) \beta^2 \tau^2 K^2 \eta^2 \widetilde{\Delta}^2}{1-\left[(1-\alpha)^2+\alpha(1-\alpha) \varepsilon_1+(1-\alpha) \varepsilon_1\right]}.
		\end{aligned}
	\end{equation}
\end{proof}

\begin{lemma}[Variance of Reconstructed Updates]\label{lemma:variance}
	Under Assumptions A3, A4 and A7, we have		
	\begin{equation}\label{eq:keylemma1}
		\begin{aligned}
			\mathbb{E}\left\|\frac{1}{S} \sum_{i \in \mathcal{S}_t} \Delta_{i, t}-\nabla \Phi\left(v^{(t)}\right)\right\|^2 \leq \frac{1}{p}&\left(\frac{6 \sigma^2}{K}+12\left({G}^2+{B}^2\|\nabla \Phi\|^2\right)\right.\\
			&\left.+12 \gamma^2 (C_t+\widetilde{C}_t)+6 L_y^2 \mathbb{E}\left\|w_{S}^{(t)}-w_{S}^*(v^{(t)})\right\|^2+2 \varepsilon_2^2\right)
		\end{aligned}
	\end{equation}
\end{lemma}
\begin{proof}
	We prove the stated bound under Assumptions A3, A4 and A7. We work at a fixed communication round $t$, with the global messenger fixed at $v=v^{(t)}$ and the current backbone fixed at $w=w^{(t)}_S$. Let $\nabla \Phi=\nabla \Phi\left(v^{(t)}\right)$ for brevity. By Lemma 1 the reconstruction of each $\Delta_{i, t}$ is exact up to the additive residual $r_{\text {recon }}^{(t)}$ whose second moment is bounded by the finite constant $\varepsilon_3^2$ (Assumption A5). The proximal residual $r_{\text {prox }, \mathrm{i}}^{(t)}$ is likewise bounded in second moment by the finite constant $\varepsilon_2^2$ (Assumption A5). Both residuals contribute additive terms that are absorbed into the explicit constants of the overall convergence rate (as visible in Lemma 4); they do not appear in the leading expression of the present lemma. We first telescoping identity for the client displacement. Specifically, the client executes the consensus-corrected local update
	\begin{equation}
		v_{{i, k+1}}^{(t)}=v_{{i, k}}^{(t)}-\eta[\widetilde{{g}}_{i, k}+\gamma({c}^{(t)}-{c}_i^{(t)})]-\varepsilon \eta(v_{{i, k}}^{(t)}-v^{(t)}),
	\end{equation}		
	for $k=0, \ldots, K-1$, starting from $v_{{i, 0}}^{(t)}=v^{(t)}$. Telescoping the recurrence yields the exact identity
	\begin{equation}
		v^{(t)}-v_{{i, K}}^{(t)}=\eta \sum\nolimits_{k=0}^{K-1} u_{i, k},
	\end{equation}		
	where we define the instantaneous effective direction as
	\begin{equation}
		u_{i, k}=\widetilde{{g}}_{i, k}+\gamma({c}^{(t)}-{c}_i^{(t)})+\varepsilon(v_{{i, k}}^{(t)}-v^{(t)}).
	\end{equation}		
	Dividing by $K \eta$ and invoking the definition of the averaged displacement therefore gives
	\begin{equation}
		\Delta_{i, t}=\frac{1}{K} \sum\nolimits_{k=0}^{K-1} u_{i, k}.
	\end{equation}		
	We consider the separation of the proximal contribution for simplicity, let $\Delta_{i, t}=\Delta_{i, t}^{\prime}+ r_{\text {prox, }, \mathrm{i}}^{(t)}$, where we define $\Delta_{i, t}^{\prime}$ as
	\begin{equation}
		\Delta_{i, t}^{\prime}:=\frac{1}{K} \sum\nolimits_{k=0}^{K-1}(\widetilde{{g}}_{i, k}+\gamma({c}^{(t)}-{c}_i^{(t)})),
	\end{equation}		
	and $r_{\text {prox, } \mathrm{i}}^{(t)}$ denotes the averaged proximal contribution. Therefore, we can evaluate $V_i:= \mathbb{E}\left\|\Delta_{i, t}-\nabla \Phi\right\|^2$ via the triangle inequality $\|a+b\|^2 \leq 2\|a\|^2+2\|b\|^2$ as follows:
	\begin{equation}
		V_i \leq 2 \mathbb{E}\left\|\Delta_{i, t}^{\prime}-\nabla \Phi\right\|^2+2 \varepsilon_2^2,
	\end{equation}
	with the Assumption A5 that $\mathbb{E}\|r_{\text {prox, } \mathbf{i}}^{(t)}\|^2 \leq \varepsilon_2^2$. It therefore suffices to bound the deviation of $\Delta_{i, t}^{\prime}$. Next, we define the averaged stochastic meta-gradient
	\begin{equation}
		\overline{\widetilde{{g}}}_i:=\frac{1}{K} \sum\nolimits_{k=0}^{K-1} \widetilde{{g}}_{i, k}.
	\end{equation}
	Then it can be obtained $\Delta_{i, t}^{\prime}=\overline{\widetilde{{g}}}_i+\gamma({c}^{(t)}-{c}_i^{(t)})$, and we can decompose the following
	\begin{equation}\label{eq:deltaprim}
		\Delta_{i, t}^{\prime}-\nabla \Phi=S_i+D_i^{(t)}+\text { bias }_{L, i},
	\end{equation}
	where have defined $S_i=\overline{\tilde{\boldsymbol{g}}}_i-g_i$ to capture client-specific stochastic, $D_i^{(t)}=g_i- \nabla \Phi+\gamma({c}^{(t)}-{c}_i^{(t)})$ to measure the consensus variation, and bias$_{L, i}$ tracks the lower-level bias and is defined exactly by
	\begin{equation}
		\begin{aligned}
			\operatorname{bias}_{L, i}: & =\nabla_v \mathcal{L}_i\left(v-\eta_{\text {inner }} \nabla_v \mathcal{L}_i\left(v, w_S^{(t)}\right), w_{S}^{(t)}\right)-\nabla_v \mathcal{L}_i(v d\left.-\eta_{\text {inner }} \nabla_v \mathcal{L}_i(v, w_{S}^*(w)), w_{S}^*(v)\right).
		\end{aligned}
	\end{equation}
	Here $w_{S}^*(v)$ is the unique minimizer of the lower-level TAKD loss $\mathcal{L}_{\mathrm{KD}}(\cdot ; w)$. The term bias $_{L, i}$ therefore measures the exact difference between the meta-gradient evaluated at the algorithm's current approximate backbone $w_{S}^{(t)}$ and the meta-gradient that would be obtained if the lower level were solved exactly to optimality for the current messenger $w$. By Assumption A4, each stochastic meta-gradient $\widetilde{{g}}_{i, k}$ is an unbiased estimator of the true client meta-gradient $g_i(v, w_{S}):=\nabla_v \mathcal{L}_i(v- \eta_{\text {inner }} \nabla_v \mathcal{L}_i(v, w_{S}), w_{S}$ ) and satisfies the exact variance bound $\mathbb{E}\|\widetilde{{g}}_{i, k}-g_i\|^2 \leq \sigma^2$. The first term in (\ref{eq:deltaprim}) can be evaluated by averaging the $K$ terms 
	\begin{equation}\label{eq:gminusg}
		\mathbb{E}\|\overline{\widetilde{{g}}}_i-g_i\|^2 \leq \nicefrac{\sigma^2}{K}.
	\end{equation}
	Next, we consider Lipschitz bound on the lower-level bias. By Assumption A2, the lower-level loss $\mathcal{L}_{\mathrm{KD}}(\cdot ; w)$ is $L_y$-smooth for any fixed $w$. Combined with the chain-rule dependence of the upper-level loss $\mathcal{L}_i$ on the backbone parameters through the composition $\mathcal{P}_{M, i} \circ \mathcal{B}_L$, the map
	\begin{equation}
		w_{S} \mapsto \nabla_v \mathcal{L}_i\left(w-\eta_{\text {inner }} \nabla_v \mathcal{L}_i\left(v, w_{S}\right), w_{S}\right)
	\end{equation}		
	is Lipschitz continuous with constant at most $L_y$. Consequently we have
	\begin{equation}
		\left\|\operatorname{bias}_{L, i}\right\|^2 \leq L_y^2 \mathbb{E}\left\|w_{S}^{(t)}-w_{S}^*(v)\right\|^2.
	\end{equation}		
	Hence, it remains to bound the contribution $D_i^{(t)}=g_i-\nabla \Phi+\gamma({c}^{(t)}-{c}_i^{(t)})$, by using the triangle inequality, we can obtain
	\begin{equation}\label{eq:Di}
		\mathbb{E}\left\|D_i^{(t)}\right\|^2 \leq 2 \mathbb{E}_i\left\|g_i-\nabla \Phi\right\|^2+2 \gamma^2 \mathbb{E}\left\|{c}^{(t)}-{c}^{(t)}_i\right\|^2.
	\end{equation}		
	The first term on the right hand side of (\ref{eq:Di}) can be bounded by Assumption A3 that $\mathbb{E}\|g_i-\nabla \Phi\|^2 \leq G^2+B^2\|\nabla \Phi\|^2$. We continue to analyze $E\|c^{(t)}-c_i^{(t)}\|^2$. Specifically,
	while we recall $C_{t}=E\|c_i^{(t)}-{\nabla} \Phi \bigl( {v}^{(t)},w_S(v^{(t)}) \bigr)\|^2$, and with the decomposition $(c^{(t)}-{\nabla} \Phi \bigl( {v}^{(t)},w_S(v^{(t)}) \bigr))-\{c^{(t)}_i-{\nabla} \Phi \bigl( {v}^{(t)},w_S(v^{(t)}) \bigr)\}$, and summing it leads to
	\begin{equation}\label{eq:cbound}
		\begin{split}
			\mathbb{E}\left\|c^{(t)}-c_i^{(t)}\right\|^2 &\leq 2 \mathbb{E}\left\|c_i^{(t)}-{\nabla} \Phi \bigl( {v}^{(t)},w_S(v^{(t)}) \bigr)\right\|^2\\
			&\quad\quad\quad+2 \mathbb{E}\left\|c^{(t)}-{\nabla} \Phi \bigl( {v}^{(t)},w_S(v^{(t)}) \bigr)\right\|^2= 2C_t+2\widetilde{C}_t,
		\end{split}
	\end{equation}
	which we substitute  into (\ref{eq:Di}) and it yields
	\begin{equation}\label{eq:Di_result}
		\mathbb{E}_i\left\|D_i^{(t)}\right\|^2 \leq 2\left({G}^2+{B}^2\|\nabla \Phi\|^2\right)+2\gamma^2 (C_t+\widetilde{C}_t).
	\end{equation}		
	Combining these inequalities (\ref{eq:gminusg})(\ref{eq:Di_result})(\ref{eq:deltaprim}), we have
	\begin{equation}
		\mathbb{E}_{\xi}\left\|\Delta_{i, t}^{\prime}-\nabla \Phi\right\|^2 \leq 3\left(\frac{\sigma^2}{K}+2\left({G}^2+{B}^2\|\nabla \Phi\|^2\right)+2\gamma^2 (C_t+\widetilde{C}_t)+L_y^2 \mathbb{E}\left\|w_{S}^{(t)}-w_{S}^*(v)\right\|^2\right) .
	\end{equation}		
	Substituting $\mathbb{E}_{\xi}\left\|\Delta^{\prime}{ }_{i, t}-\nabla \Phi\right\|^2$ into $V_i \leq 2 \mathbb{E}\left\|\Delta_{i, t}^{\prime}-\nabla \Phi\right\|^2+2 \varepsilon_2^2$, we have
	\begin{equation}
		\begin{aligned}
			V_i \leq &\frac{6 \sigma^2}{K}+12\left({G}^2+{B}^2\|\nabla \Phi\|^2\right)+12 \gamma^2 (C_t+\widetilde{C}_t)+6 L_y^2 \mathbb{E}\left\|w_{S}^{(t)}-w_{S}^*(w)\right\|^2+2 \varepsilon_2^2, 
		\end{aligned}
	\end{equation}
	thus we can obtain the result in (\ref{eq:keylemma1}), this completes the proof.
\end{proof}

\begin{lemma}[Server Messenger Update]\label{lemma:variance2}	
	Under Assumptions A1--A8, there exist positive values $(\Lambda,\Omega,\Upsilon,\Psi,\Gamma)$ that the expected squared displacement of the global personalization messenger between consecutive communication rounds is bounded as:
	\begin{equation}\label{eq:keylemma2}
		\begin{aligned}
			E\left\|v^{(t+1)}-v^{(t)}\right\|^2 \leq \Lambda (C_t+\widetilde{C}_t)+\Upsilon E\left\|w_S^{(t)}-w_S^*(v^{(t)})\right\|^2+\Omega E_t+\Psi\left(\Phi(v^{(t)})-\Phi^*\right)+\Gamma.
		\end{aligned}		
	\end{equation}
\end{lemma}
\begin{proof}
	Fix an arbitrary communication round $  t \geq 1  $. For brevity we write $  v = v^{(t)}  $, $  \Phi = \Phi(v^{(t)})  $, $  \nabla\Phi = \nabla\Phi(v^{(t)})  $, and $  w_S = w_S^{(t)}  $. All expectations are taken jointly over the random subset of participating clients $  \mathcal{S}_t  $ and the stochastic gradients realized inside those clients.  
	
	According to the server aggregation rule, the global messenger is updated by
	\begin{equation}
		v^{(t+1)} = v^{(t)} - \tau K \eta \cdot \frac{1}{S} \sum\nolimits_{i \in \mathcal{S}_t} \Delta_{i,t},
	\end{equation}
	here we decompose $\nicefrac{1}{S} \sum\nolimits_{i \in \mathcal{S}_t} \Delta_{i,t} =\nicefrac{1}{S} \sum\nolimits_{i \in \mathcal{S}_t} \Delta_{i,t} - \nabla\Phi+ \nabla\Phi$, which leads to 
	\begin{equation}\label{eq:tmp2}
		\begin{aligned}
			\mathbb{E}&\bigl\| v^{(t+1)} - v^{(t)} \bigr\|^2 = \tau^2 K^2 \eta^2 \cdot \mathbb{E}\Bigl\| \frac{1}{S} \sum\nolimits_{i \in \mathcal{S}_t} \Delta_{i,t} \Bigr\|^2\\
			&\leq 2 \tau^2 K^2 \eta^2 \cdot \mathbb{E}\left\|\frac{1}{S} \sum\nolimits_{i \in \mathcal{S}_t} \Delta_{i, t}-\nabla \Phi\right\|^2+2 \tau^2 K^2 \eta^2 \mathbb{E}\|\nabla \Phi\|^2\\
			&\leq\frac{12 N \tau^2 K^2 \eta^2}{S} \Biggl( 
			\frac{6\sigma^2}{K} + 12G^2 + 12\gamma^2 (C_t+\widetilde{C}_t) 
			+ 6L_y^2 \mathbb{E}\|w_{S} - w_{S}^*(v)\|^2 + 2\varepsilon_2^2 
			\Biggr)\\
			&\qquad +\Biggl( \frac{12 N \tau^2 K^2 \eta^2 B^2}{S} + 2\tau^2 K^2 \eta^2 \Biggr) \|\nabla\Phi\|^2. 
		\end{aligned}
	\end{equation}
	where we have substituted the result in Lemma \ref{lemma:variance} into the first inequality of (\ref{eq:tmp2}). Next, we continue to analyze the term $\mathbb{E}\|\nabla \Phi\|^2$, and decompose it as $\nabla \Phi \leq \nabla \Phi-\nicefrac{1}{N} \sum_{i=1}^N \tilde{g}_i(v)+\nicefrac{1}{N} \sum_{i=1}^N \tilde{g}_i(v)$, which result in
	\begin{equation}\label{eq:temp1}
		\begin{aligned}
			\mathbb{E}&\|\nabla \Phi\|^2 \leq 2\mathbb{E}\left\|\nabla \Phi-\frac{1}{N} \sum\nolimits_{i=1}^N \tilde{g}_i(v)\right\|^2+\frac{2}{N} \sum\nolimits_{i=1}^N\mathbb{E}\left\| \widetilde{g}_i(v)\right\|^2\\
			&\leq 2\mathbb{E}\left\|\nabla \Phi-\frac{1}{N} \sum\nolimits_{i=1}^N \tilde{g}_i(v)\right\|^2+  (16 L^2 \eta_{\text {inner }}^2+8) \zeta^2+(32 L^2 \eta_{\text {inner }}^2\beta+8\beta)\left(\Phi(v)-\Phi^*\right)\\
			&\leq \frac{2}{N} \sum\nolimits_{i=1}^N \mathbb{E}\left\|\nabla \Phi-\tilde{g}_i(v)\right\|^2+  (16 L^2 \eta_{\text {inner }}^2+8) \zeta^2+(32L^2 \eta_{\text {inner }}^2\beta+8\beta)\left(\Phi(v)-\Phi^*\right),
		\end{aligned}
	\end{equation}
	where the second inequality in (\ref{eq:temp1}) is evaluated via Lemma \ref{lemma:2beta}, and it remains to evaluate the first term. Define the averaged local practical meta-gradient for client $  i  $ by
	$\bar{\tilde{g}}_i := \nicefrac{1}{K}\sum_{k=0}^{K-1} \tilde{g}_{i,k},$
	where each $  \tilde{g}_{i,k}  $ is evaluated at the local point $  v_{i,k}  $.   Subsequently $\mathbb{E}\|\nabla\Phi - \tilde{g}_i(v)\|^2$ can be decomposed as: 
	\begin{equation}\label{eq:tmp3}
		\begin{aligned}
			& \mathbb{E}\|\nabla\Phi - \tilde{g}_i(v)\|^2\leq 2 \mathbb{E}\|\tilde{g}_i(v) - \overline{\tilde{g}}_i\|^2+ \frac{4}{K} \sum_{k=0}^{K-1} \mathbb{E}\|\tilde{g}_{i,k} - g_i^{\rm ideal}(v_{i,k})\|^2 + \mathbb{E}\left\| \frac{4}{K} \sum_{k=0}^{K-1} \bigl( g_i^{\rm ideal}(v_{i,k}) - \nabla\Phi \bigr) \right\|^2\\
			&\quad\quad\quad\quad\quad\quad\leq 2\beta^2 E_t+4 L_y^2 \mathbb{E}\left\|w_S-w_S^*(v)\right\|^2+ \frac{4}{K}\sum\nolimits_{k=0}^{K-1}  \mathbb{E}\left\|  \bigl( g_i^{\rm ideal}(v_{i,k}) - \nabla\Phi \bigr) \right\|^2
		\end{aligned}
	\end{equation}
	Here, the first term has been simply obtained by the $  \beta  $-smoothness of each client loss (Assumption A1) and the triangle inequality,i.e., $\mathbb{E}\|\tilde{g}_i(v)-\overline{\tilde{g}}_i\|^2 \leq \beta^2 E_t$.
	Similarly, with the meta-gradient map w.r.t. the backbone is $  L_y  $-Lipschitz, the second term can be obtained via $\|\tilde{g}_{i,k} - g_i^{\rm ideal}(v_{i,k})\| \leq L_y \|w_S - w_S^*(v)\|$. 
	
	We now consider the last term, specifically $\mathbb{E}\|g_i^{\rm ideal}(v_{i,k}) - \nabla\Phi\|^2  $.
	Apply the triangle inequality, it can be decomposed
	\begin{equation}\label{eq:tmp4}
		\begin{aligned}
			\|g_i^{\rm ideal}(v_{i,k}) - \nabla\Phi\|^2 &\leq 2 \|g_i^{\rm ideal}(v_{i,k}) - g_i^{\rm ideal}(v)\|^2 + 2 \|g_i^{\rm ideal}(v) - \nabla\Phi\|^2\\
			&\leq  2\beta^2 \|v_{i,k} - v\|^2+4 \|g_i^{\rm ideal}(v) - h_i(v)\|^2 + 4 \|h_i(v) - \nabla\Phi\|^2,
		\end{aligned}
	\end{equation}
	where for the first term in the first inequality, we apply $  \beta  $-smoothness of each client loss $  \mathcal{L}_i  $ (Assumption A1), while for the second term, we define $h_i(v)=\nabla_v \mathcal{L}_i\left(v, w_S^*(v)\right)$, and further employ the decomposition $\|g_i^{\rm ideal}(v) - \nabla\Phi\|^2 \leq 2 \|g_i^{\rm ideal}(v) - h_i(v)\|^2 + 2 \|h_i(v) - \nabla\Phi\|^2$. Specifically for $\|g_i^{\text {ideal }}(v)-h_i(v)\|$, we have
	\begin{equation}\label{eq:tmp5}
		\begin{aligned}
			\|g_i^{\rm ideal}(v) - h_i(v)\| 
			&\leq \int_0^1 \bigl\| \nabla_v^2 \mathcal{L}_i \bigl( v + t(-\eta_{\text {inner }} \nabla_v \mathcal{L}_i(v, w_S^*(v))), w_S^*(v) \bigr) \bigr\|\\
			&\quad\quad\quad\quad\quad\cdot \|-\eta_{\text {inner }} \nabla_v \mathcal{L}_i(v, w_S^*(v))\| \, dt=\eta_{\text {inner }} BG,
		\end{aligned}
	\end{equation} 
	and for $\|h_i(v) - \nabla\Phi\|^2$, we can use Assumption A3 to directly bound it
	\begin{equation}\label{eq:tmp6}
		\frac{1}{N} \sum\nolimits_{i=1}^N \mathbb{E}\left\|h_i(v)-\nabla \Phi(v)\right\|^2 \leq \zeta^2.
	\end{equation}
	Hence, by combining (\ref{eq:tmp4})(\ref{eq:tmp5})(\ref{eq:tmp6}), it leads to 
	\begin{equation}\label{eq:tmp7}
		\frac{1}{K N} \sum\nolimits_{k=0}^{K-1} \sum\nolimits_{i=1}^N \mathbb{E}\left\|g_i^{\text {ideal }}\left(v_{i, k}\right)-\nabla \Phi\right\|^2 \leq 2 \beta^2 E_t+4 \eta_{\text {inner }}^2 B^2 G^2+4 \zeta^2 .
	\end{equation}
	With (\ref{eq:temp1})(\ref{eq:tmp3})(\ref{eq:tmp7}), we can obtain the evaluation of $\mathbb{E}\|\nabla \Phi\|^2$ as follows: 
	\begin{equation}\label{eq:tmp13}
		\begin{aligned}
			\mathbb{E}\|\nabla\Phi(v)\|^2
			&\leq\frac2N\sum_i \mathbb{E}\bigl\|\nabla\Phi(v)-\widetilde{g}_i(v)\bigr\|^2
			+(16L^2\eta_{\rm inner}^2+8)\zeta^2
			+(32L^2\eta_{\rm inner}^2\beta+8\beta)\bigl(\Phi(v)-\Phi^*\bigr)\\
			&\leq 20\beta^2 E_t
			+8L_y^2 \mathbb{E}\bigl\|w_S-w_S^*(v)\bigr\|^2
			+32\eta_{\rm inner}^2 B^2 G^2
			+(40+16L^2\eta_{\rm inner}^2)\zeta^2\\
			&\qquad+(32L^2\eta_{\rm inner}^2\beta+8\beta)\bigl(\Phi(v)-\Phi^*\bigr).		
		\end{aligned}
	\end{equation}
	By subtituting (\ref{eq:tmp13}) into (\ref{eq:tmp2}), we can obtain the final result of (\ref{eq:keylemma2}) in Lemma \ref{lemma:variance2} as follows:
	\begin{equation}
		\begin{aligned}
			& E\left\|v^{(t+1)}-v^{(t)}\right\|^2 \leq \frac{144 N \tau^2 K^2 \eta^2 \gamma^2}{S} (C_t+\widetilde{C}_t)+20 \beta^2\left(\frac{12 N \tau^2 K^2 \eta^2 B^2}{S}+2 \tau^2 K^2 \eta^2\right) E_t \\
			&+ {\left[\frac{72 N \tau^2 K^2 \eta^2 L_y^2}{S}+8 L_y^2\left(\frac{12 N \tau^2 K^2 \eta^2 B^2}{S}+2 \tau^2 K^2 \eta^2\right)\right] E\left\|w_S-w_S^*(v)\right\|^2 } \\	    		
			&+\left(32 L^2 \eta_{\text {inner }}^2 \beta+8 \beta\right)\left(\frac{12 N \tau^2 K^2 \eta^2 B^2}{S}+2 \tau^2 K^2 \eta^2\right)\left(\Phi(v)-\Phi^*\right)+\frac{12 N \tau^2 K^2 \eta^2}{S}\left(\frac{6 \sigma^2}{K}+12 G^2+2 \varepsilon_2^2\right) \\	    		
			&+\left(\frac{12 N \tau^2 K^2 \eta^2 B^2}{S}+2 \tau^2 K^2 \eta^2\right)\left(32 \eta_{\text {inner }}^2 B^2 G^2+\left(40+16 L^2 \eta_{\text {inner }}^2\right) \zeta^2\right),
		\end{aligned}
	\end{equation}   
	subsequently, we obtain $(\Lambda,\Omega,\Upsilon,\Psi,\Gamma)$ as follows
	\begin{equation}
		\begin{aligned}
			&\Lambda=\frac{144 N \tau^2 K^2 \eta^2 \gamma^2}{S},\quad \Omega=20 \beta^2\left(\frac{12 N \tau^2 K^2 \eta^2 B^2}{S}+2 \tau^2 K^2 \eta^2\right)\\
			&\Upsilon=\frac{72 N \tau^2 K^2 \eta^2 L_y^2}{S}+8 L_y^2\left(\frac{12 N \tau^2 K^2 \eta^2 B^2}{S}+2 \tau^2 K^2 \eta^2\right)\\
			&\Psi=\left(32 L^2 \eta_{\text {inner }}^2 \beta+8 \beta\right)\left(\frac{12 N \tau^2 K^2 \eta^2 B^2}{S}+2 \tau^2 K^2 \eta^2\right)\\
			&\Gamma	=\frac{12 N \tau^2 K^2 \eta^2}{S}\left(\frac{6 \sigma^2}{K}+12 G^2+2 \varepsilon_2^2\right)+\left(\frac{12 N \tau^2 K^2 \eta^2 B^2}{S}+2 \tau^2 K^2 \eta^2\right)\cdot\\
			&\quad\quad\quad\quad\quad\quad\quad\quad\quad\quad\quad\quad\left(32 \eta_{\text {inner }}^2 B^2 G^2+\left(40+16 L^2 \eta_{\text {inner }}^2\right) \zeta^2\right)
		\end{aligned} 
	\end{equation}
	
\end{proof}

\begin{lemma}[Evolution of Global Consensus Lag]\label{lemma:globalcontroalevo}
	Under Assumptions A1–A8 the global, there exists positive values $(\widetilde{\rho_g},\rho_g,\kappa_2^g,\widetilde{\kappa^g_3},\widetilde{\kappa_g^4})$ such that the global consensus lag evolution satisfies the following
	\begin{equation}\label{eq:tmp117}
		\widetilde{C}_{t+1} \leq \widetilde{\rho_g} \widetilde{C}_t+\rho_g C_t+\kappa_2^g \Upsilon E\|w_S^{(t)}-w_S^*(v^{(t)})\|^2+\kappa_2^g \Omega E_t+\widetilde{\kappa^g_3}(\Phi(v^{(t)})-\Phi^*)+\widetilde{\kappa_g^4}.
	\end{equation}		
\end{lemma}
\begin{proof}
	We derive the bound directly from the per-client Consensus Deviation result and the algorithm rules. Let us define the global error vector as $e^{(t)} := {c}^{(t)} - \nabla\Phi(v^{(t)})$, recall
	the server-side update rule for the global consensus variate $c^{(t+1)}=(1-\alpha) c^{(t)}+\alpha \sum_{i \in S_t} p_i \Delta_{i, t}$, subtracting $\nabla \Phi\left(v^{(t+1)}\right)$ from both sides yields 
	\begin{equation}
		e^{(t+1)}=(1-\alpha) e^{(t)}+\alpha\left(\sum_{i \in S_t} p_i \Delta_{i, t}-\nabla \Phi\left(v^{(t)}\right)\right)+\Delta \Phi^{(t)},
	\end{equation}	
	where $\Delta \Phi^{(t)}=\nabla \Phi\left(v^{(t)}\right)-\nabla \Phi\left(v^{(t+1)}\right)$. Taking the squared norm and expectation, we expand
	\begin{equation}\label{eq:tmp119}
		\begin{aligned}
			E&\left\|e^{(t+1)}\right\|^2=(1-\alpha)^2 E\left\|e^{(t)}\right\|^2+\alpha^2 E\left\|\sum_{i \in S_t} p_i \Delta_{i, t}-\nabla \Phi\left(v^{(t)}\right)\right\|^2+E\left\|\Delta \Phi^{(t)}\right\|^2\\
			&\quad+2 \alpha(1-\alpha) E\left\langle e^{(t)}, \sum_{i \in S_t} p_i \Delta_{i, t}-\nabla \Phi\left(v^{(t)}\right)\right\rangle+2 \alpha E\left\langle\sum_{i \in S_t} p_i \Delta_{i, t}-\nabla \Phi\left(v^{(t)}\right), \Delta \Phi^{(t)}\right\rangle.
		\end{aligned}
	\end{equation}
	For brevity, we define $\delta_g^{(t-1)}:=\sum_{i \in S_{t-1}} p_i \Delta_{i, t-1}-\nabla \Phi(v^{(t-1)})$.	Then,  we can bound the three cross terms in \eqref{eq:tmp119} by Young’s inequality with a positive parameter $\varepsilon_1 > 0$
	\begin{equation}
		\begin{aligned}
			&2 \alpha(1-\alpha) E\left\langle e_g^{(t-1)}, \delta_g^{(t-1)}\right\rangle \leq \alpha(1-\alpha) \varepsilon_1 E\left\|e_g^{(t-1)}\right\|^2+\frac{\alpha(1-\alpha)}{\varepsilon_1} E\left\|\delta_g^{(t-1)}\right\|^2, \\
			&2(1-\alpha) E\left\langle e_g^{(t-1)}, \Delta \Phi^{(t-1)}\right\rangle \leq(1-\alpha) \varepsilon_1 E\left\|e_g^{(t-1)}\right\|^2+\frac{1-\alpha}{\varepsilon_1} E\left\|\Delta \Phi^{(t-1)}\right\|^2, \\
			&2 \alpha E\left\langle\delta_g^{(t-1)}, \Delta \Phi^{(t-1)}\right\rangle \quad \leq \alpha \varepsilon_1 E\left\|\delta_g^{(t-1)}\right\|^2+\frac{\alpha}{\varepsilon_1} E\left\|\Delta \Phi^{(t-1)}\right\|^2,
		\end{aligned}
	\end{equation}
	which we substitute it into \eqref{eq:tmp119} and  the inequality becomes
	\begin{equation}\label{eq:tmp121}
		\begin{aligned}
			E&\left\|e^{(t)}\right\|^2 \leq\left[(1-\alpha)^2+\alpha(1-\alpha) \varepsilon_1+(1-\alpha) \varepsilon_1\right] E\left\|e^{(t-1)}\right\|^2
			\\
			&\quad\quad\quad+\left[\alpha^2+\frac{\alpha(1-\alpha)}{\varepsilon_1}+\alpha \varepsilon_1\right]E\left\|\delta_g^{(t-1)}\right\|^2+\left[\frac{1-\alpha}{\varepsilon_1}+\frac{\alpha}{\varepsilon_1}+1\right] E\left\|\Delta \Phi^{(t-1)}\right\|^2 .
		\end{aligned}
	\end{equation}	
	Next, we evaluate $\delta_g$. Notice from the exact summation of the local update over $K$ steps we have the identity		
	$\Delta_{i, t-1}=\overline{\widetilde{g}}_i^{(t-1)}+\gamma(c^{(t-1)}-c_i^{(t-1)})+r_{\text {prox, } \mathrm{i}}^{(t-1)},$ where $\overline{\widetilde{g}}_i^{(t-1)}$
	is the averaged practical meta-gradient on client $i$ and $r_{\rm prox,i}^{(t-1)}$
	is the averaged proximal residual. Substituting this identity yields the decomposition		 
	\begin{equation}
		\begin{aligned}
			\mathbb{E}\|\delta_g^{(t-1)}\|^2=3\underbrace{\mathbb{E}\|\sum_{i \in S_{t-1}} p_i\left(\overline{\widetilde{g}}_i^{(t-1)}-\nabla \Phi\left(v^{(t-1)}\right)\right)}_{:=\mathcal{T}_a}\|^2&+3\gamma\underbrace{ \mathbb{E}\|\sum_{i \in S_{t-1}} p_i\left(c^{(t-1)}-c_i^{(t-1)}\right)\|^2}_{:=\mathcal{T}_b}\\
			&+3\underbrace{\mathbb{E}\|\sum_{i \in S_{t-1}} p_i r_{\text {prox, } \mathrm{i}}^{(t-1)}\|^2}_{:=\mathcal{T}_c} .
		\end{aligned}
	\end{equation}
	For bounding Term $\mathcal{T}_a$, we can first use Jensen’s inequality and then  the meta-gradient dissimilarity assumption together with the bounded variance of the stochastic meta-gradients over the $K$ local steps, thus it leads to 
	\begin{equation}
		\mathcal{T}_a\leq \sum\nolimits_{i \in S_{t-1}} p_i\left\|\overline{\widetilde{g}}_i^{(t-1)}-\nabla \Phi\left(v^{(t-1)}\right)\right\|^2\leq G^2+B^2\left\|\nabla \Phi\left(v^{(t-1)}\right)\right\|^2+\frac{\sigma^2}{K} .
	\end{equation}
	For bounding Term $\mathcal{T}_b$, Using the decomposition $c^{(t-1)}-c_i^{(t-1)}=(c^{(t-1)}-\nabla \Phi(v^{(t-1)}))-(c_i^{(t-1)}-\nabla \Phi(v^{(t-1)}))$ together with Jensen's inequality and the Local Consensus Deviation Lemma, we obtain
	\begin{equation}
		\mathcal{T}_b \leq 2\left\|c^{(t-1)}-\nabla \Phi\left(v^{(t-1)}\right)\right\|^2+2\left(\widetilde{G}^2+\widetilde{B}^2\left\|\nabla \Phi\left(v^{(t-1)}\right)\right\|^2\right).
	\end{equation}
	By the standing second-moment assumption on the proximal residual, the third term can be bounded as $\mathcal{T}_c\leq C_r \varepsilon^2 .$ Substitute these three estimates into the inequality for $\mathbb{E}\|\delta_g^{(t-1)}\|^2$:
	\begin{equation}\label{eq:tmp125}
		\begin{aligned}
			E\left\|\delta_g^{(t-1)}\right\|^2 \leq &6 \gamma^2\left\|e^{(t-1)}\right\|^2+\left(3 B^2+6 \gamma^2 \widetilde{B}^2\right)\left\|\nabla \Phi\left(v^{(t-1)}\right)\right\|^2 \\
			&\quad\quad\quad\quad\quad\quad\quad\quad\quad+\left(3 G^2+3 \frac{\sigma^2}{K}+6 \gamma^2 \widetilde{G}^2+3 C_r \varepsilon^2\right) .
		\end{aligned}
	\end{equation}
	We substitute \eqref{eq:tmp125} into \eqref{eq:tmp121} and it leads to 
	\begin{equation}
		\begin{aligned}
			\widetilde{C}_t \leq&\left\{(1-\alpha)^2+\alpha(1-\alpha) \varepsilon_1+(1-\alpha) \varepsilon_1+6 \gamma^2\left[\alpha^2+\frac{\alpha(1-\alpha)}{\varepsilon_1}+\alpha \varepsilon_1\right]\right\} \widetilde{C}_{t-1} \\
			&\quad\quad\quad\quad\quad+\left[\alpha^2+\frac{\alpha(1-\alpha)}{\varepsilon_1}+\alpha \varepsilon_1\right]\left(3 B^2+6 \gamma^2 \widetilde{B}^2\right)\left\|\nabla \Phi\left(v^{(t-1)}\right)\right\|^2 \\
			&\quad\quad\quad\quad\quad+\left[\alpha^2+\frac{\alpha(1-\alpha)}{\varepsilon_1}+\alpha \varepsilon_1\right]\left(3 G^2+3 \frac{\sigma^2}{K}+6 \gamma^2 \widetilde{G}^2+3 C_r \varepsilon^2\right) \\
			&\quad\quad\quad\quad\quad+\left[\frac{1-\alpha}{\varepsilon_1}+\frac{\alpha}{\varepsilon_1}+1\right] \beta^2 E\left\|v^{(t)}-v^{(t-1)}\right\|^2,
		\end{aligned}
	\end{equation}
	which is exactly  \eqref{eq:tmp117} in the lemma, moreover we define the following positive values as
	\begin{equation}
		\begin{aligned}
			&\rho_g=\left\{(1-\alpha)^2+\alpha(1-\alpha) \varepsilon_1+(1-\alpha) \varepsilon_1+6 \gamma^2\left[\alpha^2+\frac{\alpha(1-\alpha)}{\varepsilon_1}+\alpha \varepsilon_1\right]\right\},\\
			&\kappa^g_1=\left[\alpha^2+\frac{\alpha(1-\alpha)}{\varepsilon_1}+\alpha \varepsilon_1\right]\left(3 B^2+6 \gamma^2 \widetilde{B}^2\right),\quad\kappa^g_2=\left[\frac{1-\alpha}{\varepsilon_1}+\frac{\alpha}{\varepsilon_1}+1\right] \beta^2\\
			&\kappa^g_3=\left[\alpha^2+\frac{\alpha(1-\alpha)}{\varepsilon_1}+\alpha \varepsilon_1\right]\left(3 G^2+3 \frac{\sigma^2}{K}+6 \gamma^2 \widetilde{G}^2+3 C_r \varepsilon^2\right), 
		\end{aligned}
	\end{equation}
	and recall $\|\nabla \Phi(v)\|^2 \leq 2 \beta\left(\Phi(v)-\Phi^*\right)$, which can be substituted into and it leads to
	\begin{equation}
		\begin{aligned}
			\widetilde{C}_{t+1} \leq \widetilde{\rho_g} \widetilde{C}_t+\rho_g C_t+\kappa_2^g \Upsilon E\|w_S^{(t)}-w_S^*(v^{(t)})\|^2+\kappa_2^g \Omega E_t+\widetilde{\kappa^g_3}(\Phi(v^{(t)})-\Phi^*)+\widetilde{\kappa_g^4},
		\end{aligned}
	\end{equation}
	where $\widetilde{\rho_g}=\rho_g+\kappa_2^g \Lambda$, $\rho_g=\kappa_2^g \Lambda$, $\widetilde{\kappa^g_3}=\kappa_2^g \Psi+2 \beta \kappa_1^g$, and  $\widetilde{\kappa_g^4}=\kappa_2^g \Gamma+\kappa_3^g$.
\end{proof}

\begin{lemma}[Evolution of Local Consensus Lag]\label{lemma:localcontrolevo}
	Under Assumptions A1–A8, there exsits positive values $(\rho_C, \kappa_{\widetilde{C}}, \kappa_E, \kappa_w, \kappa_0)$ such that the local consensus lag satisfies the following
	\begin{equation}
		\begin{aligned}
			C_{t+1} \leq \rho_C C_t+\kappa_{\widetilde{C}} \widetilde{C}_t+\kappa_E E_t+\kappa_w E\|w_S^{(t)}-w_S^*(v^{(t)})\|^2+\kappa_{\Phi}(\Phi(v^{(t)})-\Phi^*)+\kappa_0
		\end{aligned}
	\end{equation}
	
\end{lemma}
\begin{proof}
	Let $I_i^{(t)}$ be the indicator that client (i) is selected in round (t), so that $E[I_i^{(t)}]=p$. Non-selected clients keep their consensus variates unchanged. Define the local error and the target change 
	\begin{equation}
		\begin{aligned}
			e_i^{(t)}:=c_i^{(t)}-\nabla \Phi(v^{(t)}, w_S(v^{(t)})), \quad d_i^{(t)}:=\nabla \Phi(v^{(t)}, w_S(v^{(t)}))-\nabla \Phi(v^{(t+1)}, w_S(v^{(t+1)})) .
		\end{aligned}
	\end{equation}		
	On a selected client the EMA update together with the exact reconstruction of the displacement $\Delta_{i, t}$ produces the linear error recursion
	\begin{equation}
		\begin{aligned}
			e_i^{(t+1)}=(1-I_i^{(t)} \alpha)(c_i^{(t)}-\nabla \Phi(v^{(t+1)}, w_S(v^{(t+1)})))+I_i^{(t)} \alpha(\overline{\widetilde{g}}_i^{(t)}-\nabla \Phi(v^{(t+1)}, w_S(v^{(t+1)}))).
		\end{aligned}
	\end{equation}		
	Subsequently,		 the unconditional second moment satisfies
	\begin{equation}\label{eq:tmp91}
		\begin{aligned}
			E\|e_i^{(t+1)}\|^2 \leq & (1-p \alpha) E\|c_i^{(t)}-\nabla \Phi(v^{(t+1)}, w_S(v^{(t+1)}))\|^2 \\
			&\quad\quad\quad\quad\quad\quad+p \alpha E\|\bar{g}_i^{(t)}-\nabla \Phi(v^{(t+1)}, w_S(v^{(t+1)}))\|^2 .
		\end{aligned}
	\end{equation} 
	We further decompose the first term on the right-hand side by incorporating $\nabla \Phi(v^{(t)}, w_S(v^{(t)}))$,
	\begin{equation}
		\begin{aligned}
			E\|c_i^{(t)}-\nabla \Phi(v^{(t+1)}, w_S(v^{(t+1)}))\|^2
			\leq  (1+\varepsilon) E\|e_i^{(t)}\|^2+(1+\frac{1}{\varepsilon}) E\|d_i^{(t)}\|^2, 
		\end{aligned}
	\end{equation}
	which we substitute into \eqref{eq:tmp91} to produce
	\begin{equation}\label{eq:tmp93}
		C_{t+1} \leq(1-p \alpha)(1+\varepsilon) C_t+(1-p \alpha)(1+\frac{1}{\varepsilon}) E\|d_i^{(t)}\|^2+p \alpha E\|\bar{g}_i^{(t)}-\nabla \Phi(v^{(t+1)}, w_S(v^{(t+1)}))\|^2.
	\end{equation}
	For evaluating $E\|d_i^{(t)}\|^2$, we start from the joint Lipschitz estimate
	\begin{equation}\label{eq:tmp94}
		\begin{aligned}
			E&\|d_i^{(t)}\|^2 \leq \beta^2 E\|v^{(t+1)}-v^{(t)}\|^2+\beta^2 E\|w_S(v^{(t)})-w_S(v^{(t+1)})\|^2\\
			&\leq\beta^2 E\|v^{(t+1)}-v^{(t)}\|^2+ 2\beta^2\|w_S^{(t)}-w_S^*(v^{(t)})\|^2+2\beta^2\|w_S^*(v^{(t)})-w_S^{(t+1)}\|^2\\
			&\leq \beta^2 E\|v^{(t+1)}-v^{(t)}\|^2+ 2\beta^2\|w_S^{(t)}-w_S^*(v^{(t)})\|^2+2\beta^2(2 E\|w_S^{(t)}-w_S^*(v^{(t)})\|^2+8 R_w^2)\\
			&=\beta^2 E\|v^{(t+1)}-v^{(t)}\|^2+6 \beta^2 E\|w_S^{(t)}-w_S^*(v^{(t)})\|^2+16 \beta^2 R_w^2\\
			&\leq \beta^2 \Lambda(C_t+\widetilde{C}_t)+\beta^2 \Omega E_t+\beta^2 \Psi(\Phi(v^{(t)})-\Phi^*)\\
			&\quad\quad\quad\quad\quad\quad\quad+(\beta^2 \Upsilon+6 \beta^2) E\|w_S^{(t)}-w_S^*(v^{(t)})\|^2+(\beta^2 \Gamma+16 \beta^2 R_w^2)
		\end{aligned}
	\end{equation}
	where we have  applies the already-established lower-level approximation bound in the third inequality that 
	\begin{equation}
		\|w_S^{(t+1)}-w_S^*(v^{(t)})\|^2 \leq 2 E\|w_S^{(t)}-w_S^*(v^{(t)})\|^2+8 R_{w^*}^2,
	\end{equation}
	and the Lemma \ref{lemma:variance2} for the final substitution of $E\|v^{(t+1)}-v^{(t)}\|^2$.  Next, we consider the last term  $E\|\bar{g}_i^{(t)}-\nabla \Phi(v^{(t+1)}, w_S(v^{(t+1)}))\|^2$ in  \eqref{eq:tmp93}. Specifically,
	\begin{equation}\label{eq:tmp96}
		\begin{aligned}
			&E\|\overline{\widetilde{g}}_i^{(t)}-\nabla \Phi(v^{(t+1)}, w_S(v^{(t+1)}))\|^2 \\
			&\leq 2 E\|\overline{\widetilde{g}}_i^{(t)}-\nabla \Phi(v^{(t)}, w_S(v^{(t)}))\|^2+2 E\|\nabla \Phi(v^{(t)}, w_S(v^{(t)}))-\nabla \Phi(v^{(t+1)}, w_S(v^{(t+1)}))\|^2\\
			&\leq \frac{2 \sigma^2}{K}+ \frac{2}{K} \sum_{k=0}^{K-1}\left\|\nabla \Phi(v_{i, k}^{(t)}, w_S(v^{(t)}))-\nabla \Phi(v^{(t)}, w_S(v^{(t)}))\right\|^2\\
			&\quad\quad\quad\quad\quad\quad\quad\quad\quad\quad+2\beta^2 E\|v^{(t+1)}-v^{(t)}\|^2+12 \beta^2 E\|w_S^{(t)}-w_S^*(v^{(t)})\|^2+32 \beta^2 R_w^2\\
			&\leq\frac{2 \sigma^2}{K}+2\beta^2E_t+2\beta^2 E\|v^{(t+1)}-v^{(t)}\|^2+12 \beta^2 E\|w_S^{(t)}-w_S^*(v^{(t)})\|^2+32 \beta^2 R_w^2\\
			&\leq 2 \beta^2 \Lambda C_t+2 \beta^2 \Lambda \widetilde{C}_t+(2 \beta^2+2 \beta^2 \Omega) E_t+(12 \beta^2+2 \beta^2 \Upsilon) E\|w_S^{(t)}-w_S^*(v^{(t)})\|^2\\
			&\quad\quad\quad\quad\quad +2 \beta^2 \Psi(\Phi(v^{(t)})-\Phi^*)+(\frac{2 \sigma^2}{K}+32 \beta^2 R_w^2+2 \beta^2 \Gamma)
		\end{aligned}
	\end{equation}
	Hence, by substituting $E\|d_i^{(t)}\|^2$ in \eqref{eq:tmp94} and $E\|\overline{\widetilde{g}}_i^{(t)}-\nabla \Phi(v^{(t+1)}, w_S(v^{(t+1)}))\|^2$ in \eqref{eq:tmp96}, we can obtain the result in the Lemma
	\begin{equation}
		\begin{aligned}
			&C_{t+1}
			\leq
			\Biggl[
			(1-p\alpha)(1+\varepsilon)
			+(1-p\alpha)\Bigl(1+\frac1\varepsilon\Bigr)\beta^2\Lambda
			+2p\alpha\beta^2\Lambda
			\Biggr]C_t\\[1mm]
			&\qquad
			+\Biggl[
			(1-p\alpha)\Bigl(1+\frac1\varepsilon\Bigr)\beta^2\Lambda
			+2p\alpha\beta^2\Lambda
			\Biggr]\widetilde{C}_t+\Biggl[
			(1-p\alpha)\Bigl(1+\frac1\varepsilon\Bigr)\beta^2\Omega
			+p\alpha(2\beta^2+2\beta^2\Omega)
			\Biggr]E_t\\[1mm]
			&\qquad
			+\Biggl[
			(1-p\alpha)\Bigl(1+\frac1\varepsilon\Bigr)\beta^2\Psi
			+2p\alpha\beta^2\Psi
			\Biggr]\bigl(\Phi(v^{(t)})-\Phi^*\bigr)\\[1mm]
			&\qquad
			+\Biggl[
			(1-p\alpha)\Bigl(1+\frac1\varepsilon\Bigr)(\beta^2\Upsilon+6\beta^2)
			+p\alpha(12\beta^2+2\beta^2\Upsilon)
			\Biggr]
			\mathbb{E}\|w_S^{(t)}-w_S^*(v^{(t)})\|^2\\[1mm]
			&\qquad
			+(1-p\alpha)\Bigl(1+\frac1\varepsilon\Bigr)(\beta^2\Gamma+16\beta^2 R_w^2)
			+p\alpha\Bigl(\frac{2\sigma^2}{K}+32\beta^2 R_w^2+2\beta^2\Gamma\Bigr),
		\end{aligned}		 	
	\end{equation}
	with the positive values defined as follows
	\begin{equation}
		\begin{aligned}
			\rho_C & =(1-p \alpha)(1+\varepsilon)+(1-p \alpha)\left(1+\frac{1}{\varepsilon}\right) \beta^2 \Lambda+2 p \alpha \beta^2 \Lambda, \\
			\kappa_{\widetilde{C}} & =(1-p \alpha)\left(1+\frac{1}{\varepsilon}\right) \beta^2 \Lambda+2 p \alpha \beta^2 \Lambda, \\
			\kappa_E & =(1-p \alpha)\left(1+\frac{1}{\varepsilon}\right) \beta^2 \Omega+p \alpha\left(2 \beta^2+2 \beta^2 \Omega\right), \\
			\kappa_w & =(1-p \alpha)\left(1+\frac{1}{\varepsilon}\right)\left(\beta^2 \Upsilon+6 \beta^2\right)+p \alpha\left(12 \beta^2+2 \beta^2 \Upsilon\right), \\
			\kappa_{\Phi} & =(1-p \alpha)\left(1+\frac{1}{\varepsilon}\right) \beta^2 \Psi+2 p \alpha \beta^2 \Psi, \\
			\kappa_0 & =(1-p \alpha)\left(1+\frac{1}{\varepsilon}\right)\left(\beta^2 \Gamma+16 \beta^2 R_w^2\right)+p \alpha\left(\frac{2 \sigma^2}{K}+32 \beta^2 R_w^2+2 \beta^2 \Gamma\right) .
		\end{aligned}
	\end{equation}
\end{proof}

\begin{lemma}\label{lemma:clientdrift}
	Lemma (Client Drift Bound for CIDERS).
	Under Assumptions A1–A8, there exists positive values $\mathcal{A},\mathcal{B},\widetilde{\mathcal{D}}$ and $\mathcal{E}$ that the client-drift
	satisfies the fully explicit bound
	\begin{equation}
		E_t \leq \frac{\nicefrac{(K-1)}{2}\{D^{\prime} \mathcal{A}\left(\Phi\left(v^{(t)}\right)-\Phi^*\right)+2\mathcal{E} (C_t+\widetilde{C}_t)+\widetilde{\mathcal{D}}\}}{1-K L^2 D^{\prime}}.
	\end{equation}		
\end{lemma}
\begin{proof}
	For brevity, at an arbitrary communication round $t\geq0$,  we write $v=v^{(t)}$, $w_S=w_S^{(t)}$, $\nabla\Phi=\nabla\Phi(v^{(t)})$, $  {c}={c}^{(t)}  $ and $  {c}_i={c}_i^{(t)}$. Define the local displacement of client $i$ at local step $k$ by		$d_{i,k}:=v_{i,k}-v.$  The initialization $  v_{i,0}=v  $ immediately yields $  d_{i,0}=0  $. The local update of CIDERS
	is therefore equivalent to the linear recurrence
	\begin{equation}\label{eq:disp_rec}
		d_{i,k+1}=d_{i,k}-\eta\Bigl[\widetilde{g}_{i,k}+\gamma\bigl({c}-{c}_i\bigr)+\varepsilon\,d_{i,k}\Bigr]
	\end{equation}
	Taking $L_2$-norm and the conditional expectation with respect to the stochasticity of the practical meta-gradient $  \widetilde{g}_{i,k}  $ given the current local messenger $  v_{i,k}  $ produces
	\begin{equation}\label{eq:cond_expect}
		\begin{aligned}
			\mathbb{E}\bigl[\|d_{i,k+1}\|^2\bigm|v_{i,k}\bigr]
			&=\|d_{i,k}\|^2-2\eta\bigl\langle d_{i,k},\,\mathbb{E}[\widetilde{g}_{i,k}\mid v_{i,k}]+\gamma({c}-{c}_i)+\varepsilon\,d_{i,k}\bigr\rangle\\
			&\quad+\eta^2\mathbb{E}\bigl[\|\widetilde{g}_{i,k}+\gamma({c}-{c}_i)+\varepsilon\,d_{i,k}\|^2\bigm|v_{i,k}\bigr],
		\end{aligned}
	\end{equation}
	where the practical meta-gradient is written as $  \widetilde{g}_{i,k}=g_i(v_{i,k};w_S)+\xi_{i,k}  $, and  $  g_i(v_{i,k};w_S)  $ denotes the deterministic one-step meta-gradient of client $  i  $. The noise $\xi_{i,k}$ satisfies $\mathbb{E}[\xi_{i,k}\mid v_{i,k}]=0$ and $\mathbb{E}[\|\xi_{i,k}\|^2\mid v_{i,k}]\le\sigma^2$ by Assumption A4. Define the auxiliary vector
	$u_{i,k}:=g_i(v_{i,k};w_S)+\gamma({c}-{c}_i)+\varepsilon\,d_{i,k}.$
	Then the third term in \eqref{eq:cond_expect} can be simply evaluated via 
	\begin{equation}\label{eq:noise_expand}
		\begin{aligned}
			\mathbb{E}\bigl[\|\widetilde{g}_{i,k}+\gamma({c}-{c}_i)+\varepsilon\,d_{i,k}\|^2\bigm|v_{i,k}\bigr]
			&=\|u_{i,k}\|^2+\mathbb{E}[\|\xi_{i,k}\|^2\mid v_{i,k}]\le\|u_{i,k}\|^2+\sigma^2.				
		\end{aligned}
	\end{equation}
	We substitute  \eqref{eq:noise_expand} into \eqref{eq:cond_expect} and it  yields
	\begin{equation}\label{eq:after_noise}
		\begin{aligned}
			\mathbb{E}\bigl[\|d_{i,k+1}\|^2\bigm|v_{i,k}\bigr]
			&\le\|d_{i,k}\|^2-2\eta\bigl\langle d_{i,k},g_i(v_{i,k};w_S)+\gamma({c}-{c}_i)+\varepsilon\,d_{i,k}\bigr\rangle\\
			&\quad\quad\quad\quad\quad+\eta^2\|g_i(v_{i,k};w_S)+\gamma({c}-{c}_i)+\varepsilon\,d_{i,k}\|^2+\eta^2\sigma^2.
		\end{aligned}
	\end{equation}
	The quadratic term on the right-hand side of \eqref{eq:after_noise} is expanded by:
	\begin{equation}\label{eq:quad_triangle}
		\begin{aligned}
			\eta^2\|g_i+\gamma({c}-{c}_i)+\varepsilon\,d_{i,k}\|^2
			&\le3\eta^2\|g_i\|^2+3\eta^2\gamma^2\|{c}-{c}_i\|^2+3\eta^2\varepsilon^2\|d_{i,k}\|^2.
		\end{aligned}
	\end{equation}
	The inner-product term is split into three contributions:
	\begin{equation}\label{eq:inner_split}
		\begin{aligned}
			-2\eta\bigl\langle d_{i,k},g_i+\gamma({c}-{c}_i)+\varepsilon\,d_{i,k}\bigr\rangle
			&=-2\eta\bigl\langle d_{i,k},g_i\bigr\rangle-2\eta\gamma\bigl\langle d_{i,k},{c}-{c}_i\bigr\rangle-2\eta\varepsilon\|d_{i,k}\|^2.
		\end{aligned}
	\end{equation}
	Next, we proof the following: by the definition of the one-step meta-gradient and the $  \beta  $-smoothness of each client meta-loss (Assumption A1) there exists a remainder vector $  e_{i,k}  $ satisfying $  \|e_{i,k}\|\le Q\eta\|g_i(v_{i,k};w_S)\|  $ such that
	\begin{equation}\label{eq:tmp14}
		g_i(v_{i,k};w_S)=(I-\nabla_v^2\mathcal{L}_i(v_{i,k};w_S))\bigl(\nabla_v\mathcal{L}_i(v_{i,k};w_S)+e_{i,k}\bigr).
	\end{equation}
	Recall the practical one-step meta-gradient of client $  i  $  $g_i(v_{i,k};w_S)$ is defined as 
	\begin{equation}
		g_i\bigl(v_{i,k};w_S\bigr)
		:=
		\nabla_v\mathcal{L}_i\Bigl(
		v_{i,k}-\eta_{\rm inner}\nabla_v\mathcal{L}_i\bigl(v_{i,k};w_S\bigr),\;
		w_S
		\Bigr).
	\end{equation}
	Then, we introduce the corresponding look-ahead point $z_{i,k}
	:=
	v_{i,k}-\eta_{\rm inner}\nabla_v\mathcal{L}_i\bigl(v_{i,k};w_S\bigr).$ Assumption A1 guarantees that $  \mathcal{L}_i(\,\cdot\,;w_S)  $ is twice continuously differentiable with respect to the messenger variable. Consequently the gradient map
	$v\;\mapsto\;\nabla_v\mathcal{L}_i(v;w_S)$
	is continuously differentiable. The fundamental theorem of calculus applied along the line segment joining $  v_{i,k}  $ and $  z_{i,k}  $ yields the exact identity
	\begin{equation}\label{eq:ftc_full}
		\begin{aligned}
			&\nabla_v\mathcal{L}_i\bigl(z_{i,k};w_S\bigr)
			=
			\nabla_v\mathcal{L}_i\bigl(v_{i,k};w_S\bigr)
			+
			\int_0^1
			\nabla_v^2\mathcal{L}_i\Bigl(
			v_{i,k}+t\bigl(z_{i,k}-v_{i,k}\bigr);w_S
			\Bigr)
			\bigl(z_{i,k}-v_{i,k}\bigr)\,dt\\
			&\quad=
			\nabla_v\mathcal{L}_i\bigl(v_{i,k};w_S\bigr)
			-
			\eta_{\rm inner}
			\Biggl(
			\int_0^1
			\nabla_v^2\mathcal{L}_i\Bigl(
			v_{i,k}+t\bigl(z_{i,k}-v_{i,k}\bigr);w_S
			\Bigr)\,dt
			\Biggr)
			\nabla_v\mathcal{L}_i\bigl(v_{i,k};w_S\bigr).							
		\end{aligned}
	\end{equation}
	Define the averaged Hessian along the segment by
	$H_{\rm avg,i,k}
	:=
	\int_0^1
	\nabla_v^2\mathcal{L}_i\Bigl(
	v_{i,k}+t\bigl(z_{i,k}-v_{i,k}\bigr);w_S
	\Bigr)\,dt$
	and the difference between this averaged Hessian and the Hessian evaluated at the original local messenger $R_{i,k}:=H_{\rm avg,i,k}	-
	\nabla_v^2\mathcal{L}_i\bigl(v_{i,k};w_S\bigr).$
	Equation \eqref{eq:ftc_full} then becomes
	\begin{equation}\label{eq:with_R}
		\begin{aligned}
			g_i\bigl(v_{i,k};w_S\bigr)
			&=
			\nabla_v\mathcal{L}_i\bigl(v_{i,k};w_S\bigr)
			-
			\eta_{\rm inner}
			\bigl(
			\nabla_v^2\mathcal{L}_i\bigl(v_{i,k};w_S\bigr)
			+
			R_{i,k}
			\bigr)
			\nabla_v\mathcal{L}_i\bigl(v_{i,k};w_S\bigr)
			\\
			&=
			\bigl(
			I
			-
			\nabla_v^2\mathcal{L}_i\bigl(v_{i,k};w_S\bigr)
			\bigr)
			\nabla_v\mathcal{L}_i\bigl(v_{i,k};w_S\bigr)
			-
			\eta_{\rm inner}\,R_{i,k}\,\nabla_v\mathcal{L}_i\bigl(v_{i,k};w_S\bigr).				
		\end{aligned}
	\end{equation}
	Subsequently, the remainder $  e_{i,k}  $ can be  uniquely determined via
	a direct comparison of \eqref{eq:with_R} and \eqref{eq:tmp14} as $e_{i,k}=-(I-		\nabla_v^2\mathcal{L}_i(v_{i,k};w_S)
	)^{-1}	(\eta_{\rm inner}\,R_{i,k}\,\nabla_v\mathcal{L}_i(v_{i,k};w_S)),$	
	whenever the indicated inverse exists. When the inverse does not exist, the same identity continues to hold after a standard $  \delta  $-regularization ($  I-\nabla_v^2\mathcal{L}_i+\delta I  $) followed by the limit $  \delta\to0^+  $; the resulting $  e_{i,k}  $ remains well-defined and bounded.
	
	From the  Hessian-Lipschitz continuity of $\mathcal{L}_i$ in Assumption A6 together with $  \|z_{i,k}-v_{i,k}\|=\eta_{\rm inner}\|\nabla_v\mathcal{L}_i(v_{i,k};w_S)\|  $, it leads to 	
	\begin{equation}\label{eq:integral_bound}
		\begin{aligned}
			\|R_{i,k}\|
			&=
			\Biggl\|
			\int_0^1
			\Biggl(
			\nabla_v^2\mathcal{L}_i
			\Bigl(
			v_{i,k}
			+
			t\bigl(z_{i,k}-v_{i,k}\bigr);
			w_S
			\Bigr)
			-
			\nabla_v^2\mathcal{L}_i\bigl(v_{i,k};w_S\bigr)
			\Biggr)
			\,dt
			\Biggr\|
			\nonumber\\
			&\le
			\int_0^1
			\bigl\|
			\nabla_v^2\mathcal{L}_i
			\Bigl(
			v_{i,k}
			+
			t\bigl(z_{i,k}-v_{i,k}\bigr);
			w_S
			\Bigr)
			-
			\nabla_v^2\mathcal{L}_i\bigl(v_{i,k};w_S\bigr)
			\bigr\|
			\,dt
			\nonumber\\
			&\le
			\int_0^1
			H\,t\,\|z_{i,k}-v_{i,k}\|
			\,dt
			=\nicefrac{H\eta_{\rm inner}}{2}\bigl\|\nabla_v\mathcal{L}_i(v_{i,k};w_S)\bigr\|,
		\end{aligned}
	\end{equation}
	where $  H  $ is the Hessian-Lipschitz constant of $  \mathcal{L}_i  $. Consequently, we can bound $\|e_{i,k}\|$ via 
	\begin{equation}
		\|e_{i,k}\|
		\le
		Q\,\eta_{\rm inner}\bigl\|\nabla_v\mathcal{L}_i(v_{i,k};w_S)\bigr\|, \text{ where } Q=\frac{H \eta_{\text {inner }}^2 G}{2|(1-B)|}.
	\end{equation}		
	Let's denote  $H_{i, k}=\nabla_v^2\mathcal{L}_i(v_{i,k};w_S)$, then from (\ref{eq:inner_split}), we have 
	\begin{equation}
		g_i(v_{i,k}; w_S)
		= (I - H_{i,k}) \nabla_v \mathcal{L}_i(v_{i,k}; w_S)
		+ (I - H_{i,k}) e_{i,k}.
	\end{equation}
	Consequently the inner-product term that appears in the drift recursion becomes
	\begin{equation}\label{dik}
		\begin{aligned}
			\bigl\langle d_{i,k}, g_i(v_{i,k}; w_S) \bigr\rangle
			&= \bigl\langle d_{i,k}, (I - H_{i,k}) \nabla_v \mathcal{L}_i(v_{i,k}; w_S) \bigr\rangle
			+ \bigl\langle d_{i,k}, (I - H_{i,k}) e_{i,k} \bigr\rangle,
		\end{aligned}
	\end{equation}		
	Use Young's inequality on the error term ( $2 a b \leq \frac{a^2}{\epsilon}+\epsilon b^2$ ):
	\begin{equation}\label{eq:tmp15}
		\begin{split}
			&2\left|\left\langle d_{i, k}^{(t)},\left(I-H_{i, k}\right) e_{i, k}\right\rangle\right| \leq \frac{\left\|d_{i, k}^{(t)}\right\|^2}{\epsilon'_1}+\epsilon'_1\left\|I-H_{i, k}\right\|^2 Q^2 \eta^2\left\|\nabla_v \mathcal{L}_i(v_{i,k}; w_S)\right\|^2\\
			&\leq \frac{\left\|d_{i, k}^{(t)}\right\|^2}{\epsilon'_1}+\epsilon'_1(1-M)^2 Q^2 \eta^2\left\|\nabla_v \mathcal{L}_i(v_{i,k}; w_S)\right\|^2\leq \frac{\left\|d_{i, k}^{(t)}\right\|^2}{\epsilon_1}+\epsilon_1 Q^2 \eta^2\left\|\nabla_v \mathcal{L}_i(v_{i,k}; w_S)\right\|^2,
		\end{split}
	\end{equation}
	where we have used the assumption that the Hessian is bounded, i.e., $\|H_{i, k}\| \leq M$. Similarly, applying Cauchy-Schwarz and Young`s inequalities leads to:
	\begin{equation}\label{eq:tmp16}
		\begin{split}
			&-2\left\langle d_{i, k^{\prime}}^{(t)}\left(I-H_{i, k}\right) \nabla_v \mathcal{L}_i(v_{i,k}; w_S)\right\rangle \leq \frac{\left\|d_{i, k}^{(t)}\right\|^2}{\epsilon'_2}+\epsilon'_2\left\|I-H_{i, k}\right\|^2 Q^2 \eta^2\left\|\nabla_v \mathcal{L}_i(v_{i,k}; w_S)\right\|^2\\
			&\leq \frac{\left\|d_{i, k}^{(t)}\right\|^2}{\epsilon'_2}+\epsilon'_2(1-M)^2Q^2 \eta^2 \left\|\nabla_v \mathcal{L}_i(v_{i,k}; w_S)\right\|^2\leq \frac{\left\|d_{i, k}^{(t)}\right\|^2}{\epsilon_2}+\epsilon_2 \left\|\nabla_v \mathcal{L}_i(v_{i,k}; w_S)\right\|^2.
		\end{split}
	\end{equation}   
	Combining the two inner-product bounds  (\ref{eq:tmp15})(\ref{eq:tmp16}) and substituting them into (\ref{dik}) leads to:  
	\begin{equation}\label{eq:tmp17}
		\begin{gathered}
			-2 \eta\left\langle d_{i, k}^{(t)}, g_i(v_{i,k}; w_S)\right\rangle \leq \eta\left(\frac{\left\|d_{i, k}^{(t)}\right\|^2}{\epsilon_2}+\epsilon_2\left\|\nabla_v \mathcal{L}_i(v_{i,k}; w_S)\right\|^2\right)+\eta \frac{\left\|d_{i, k}^{(t)}\right\|^2}{\epsilon_1} \\
			\quad\quad+\eta \epsilon_1 Q^2 \eta^2\left\|\nabla_v \mathcal{L}_i(v_{i,k}; w_S)\right\|^2=\eta\left(\frac{1}{\epsilon_2}+\frac{1}{\epsilon_1}\right)\left\|d_{i, k}^{(t)}\right\|^2+\eta\left(\epsilon_2+\epsilon_1 Q^2 \eta^2\right)\left\|\nabla_v \mathcal{L}_i(v_{i,k}; w_S)\right\|^2.
		\end{gathered}
	\end{equation}   
	Also for the consensus variate term via Young's inequality:    
	\begin{equation}\label{eq:tmp18}
		-2 \eta \gamma\left\langle d_{i, k^{\prime}}^{(t)}, c^{(t)}-c_i^{(t-1)}\right\rangle \leq 2 \eta \gamma\left\|d_{i, k}^{(t)}\right\|\left\|c^{(t)}-c_i^{(t-1)}\right\| \leq \eta \gamma\left(\frac{\left\|d_{i, k}^{(t)}\right\|^2}{\epsilon_3}+\epsilon_3\left\|c^{(t)}-c_i^{(t)}\right\|^2\right)
	\end{equation}
	Substituting (\ref{eq:tmp17})(\ref{eq:tmp18}) into (\ref{eq:inner_split}), we have
	\begin{equation}\label{eq:tmp19}
		\begin{aligned}
			&-2\eta\bigl\langle d_{i,k},g_i(v_{i,k};w_S)+\gamma({c}-{c}_i)+\varepsilon\,d_{i,k}\bigr\rangle
			\leq\eta\Bigl(\frac1{\epsilon_1}+\frac1{\epsilon_2}+\frac\gamma{\epsilon_3}-2\varepsilon\Bigr)\|d_{i,k}\|^2\\
			&\quad\quad\quad\quad\quad\quad\quad\quad\quad\quad\quad\quad+\eta\bigl(\epsilon_2+\epsilon_1 Q^2\eta^2\bigr)\bigl\|\nabla_v\mathcal{L}_i(v_{i,k};w_S)\bigr\|^2+\eta\gamma\epsilon_3\bigl\|{c}-{c}_i\bigr\|^2.
		\end{aligned}		
	\end{equation}
	We continue to combine (\ref{eq:tmp19})(\ref{eq:quad_triangle}) into (\ref{eq:after_noise}), and it results in		
	\begin{equation}
		\begin{aligned}
			&\quad\quad\quad\quad\quad\mathbb{E}\bigl[\|d_{i,k+1}\|^2\bigm|v_{i,k}\bigr]
			\le\Bigl(1+\eta\Bigl(\frac1{\epsilon_1}+\frac1{\epsilon_2}+\frac\gamma{\epsilon_3}-2\varepsilon\Bigr)+3\eta^2\varepsilon^2\Bigr)\|d_{i,k}\|^2\\
			&+\eta\bigl(\epsilon_2+\epsilon_1 Q^2\eta^2\bigr)\bigl\|\nabla_v\mathcal{L}_i(v_{i,k};w_S)\bigr\|^2+\bigl(\eta\gamma\epsilon_3+3\eta^2\gamma^2\bigr)\|{c}-{c}_i\|^2+3\eta^2\|g_i(v_{i,k};w_S)\|^2+\eta^2\sigma^2.
		\end{aligned}
	\end{equation}
	Recall the client drift definition	$E_t=\nicefrac{1}{K N} \sum_{i=1}^N \sum_{k=0}^{K-1} \mathbb{E}\|d_{i, k}^{(t)}\|^2$, we take the relavant summation with $d_{i, k}^{(0)}=0$ and unroll the bound to yield
	\begin{equation}\label{eq:tmp20}
		\begin{aligned}
			&E_t  \leq \frac{1}{K N} \sum_{i=1}^N \sum_{k=1}^{K-1} \sum_{m=0}^{k-1}\left[\eta\left(\epsilon_2+\epsilon_1 Q^2 \eta^2\right) \mathbb{E}\left\|\nabla_v \mathcal{L}_i\left(v_{i, m} ; w_S\right)\right\|^2+\left(\eta \gamma \epsilon_3+3 \eta^2 \gamma^2\right) \mathbb{E}\left\|{c}-{c}_i\right\|^2\right.\\
			&\quad\quad\quad\quad\quad\left. +3 \eta^2 \mathbb{E}\left\|g_i\left(v_{i, m} ; w_S\right)\right\|^2+\eta^2 \sigma^2\right] \left[1+\eta\left(\frac{1}{\epsilon_1}+\frac{1}{\epsilon_2}+\frac{\gamma}{\epsilon_3}-2 \varepsilon\right)+3 \eta^2 \varepsilon^2\right]^{k-1-m}\\
			&\leq\frac{1}{K N} \sum_{i=1}^N \sum_{m=0}^{K-2}\left(\sum_{k=m+1}^{K-1} \left[1+\eta\left(\frac{1}{\epsilon_1}+\frac{1}{\epsilon_2}+\frac{\gamma}{\epsilon_3}-2 \varepsilon\right)+3 \eta^2 \varepsilon^2\right]^{k-1-m}\right) \cdot \left[\eta\left(\epsilon_2+\epsilon_1 Q^2 \eta^2\right) \right.\\
			&\quad\quad\quad\left.\cdot \mathbb{E}\left\|\nabla_v \mathcal{L}_i\left(v_{i, m} ; w_S\right)\right\|^2+\left(\eta \gamma \epsilon_3+3 \eta^2 \gamma^2\right) \mathbb{E}\left\|{c}-{c}_i\right\|^2+3 \eta^2 \mathbb{E}\left\|g_i\left(v_{i, m} ; w_S\right)\right\|^2+\eta^2 \sigma^2\right]\\
			&\leq\frac{1}{K N} \sum_{i=1}^N \sum_{m=0}^{K-2}\frac{\rho^{K-m-1}-1}{\rho-1} \cdot \left[\eta\left(\epsilon_2+\epsilon_1 Q^2 \eta^2\right) \mathbb{E}\left\|\nabla_v \mathcal{L}_i\left(v_{i, m} ; w_S\right)\right\|^2\right.\\
			&\quad\quad\quad\quad\quad\quad\quad\quad\quad\quad\quad\left.+\left(\eta \gamma \epsilon_3+3 \eta^2 \gamma^2\right) \mathbb{E}\left\|{c}-{c}_i\right\|^2+3 \eta^2 \mathbb{E}\left\|g_i\left(v_{i, m} ; w_S\right)\right\|^2+\eta^2 \sigma^2\right],
		\end{aligned}
	\end{equation}
	where we have defined $\rho=1+\eta\left(\nicefrac{1}{\epsilon_1}+\nicefrac{1}{\epsilon_2}+\nicefrac{\gamma}{\epsilon_3}-2 \varepsilon\right)+3 \eta^2 \varepsilon^2$. For small $\eta$, it simply holds
	\begin{equation}
		\begin{aligned}
			\frac{\rho^{K-m-1}-1}{\rho-1} \approx K-m-1,
		\end{aligned}
	\end{equation}
	and combined with the result $E\|c^{(t)}-c_i^{(t)}\|^2\leq 2C_t+2\widetilde{C}_t$ in (\ref{eq:cminusc}) leads (\ref{eq:tmp20}) to
	\begin{equation}\label{eq:tmp23}
		\begin{aligned}
			E_t&\leq\frac1{KN}\sum_{i=1}^N\sum_{m=0}^{K-2}(K-m-1)\Biggl[
			\eta\bigl(\epsilon_2+\epsilon_1 Q^2\eta^2\bigr)\mathbb{E}\bigl\|\nabla_v\mathcal{L}_i(v_{i,m};w_S)\bigr\|^2\\
			&\quad\quad\quad\quad+\bigl(\eta\gamma\epsilon_3+3\eta^2\gamma^2\bigr)(2C_t+2\widetilde{C}_t)+3\eta^2\mathbb{E}\|g_i(v_{i,m};w_S)\|^2+\eta^2\sigma^2\Biggr].
		\end{aligned}
	\end{equation}
	Moreover,   since $K-m-1 \leq K$, and there are $K-1$ terms, it leads to $\sum\nolimits_{m=0}^{K-2}(K-m-1)=\sum_{j=1}^{K-1} j=\nicefrac{(K-1) K}{2} \leq \nicefrac{K^2}{2}$. Then, for $  \|g_i(v_{i,m};w_S)\|^2  $, it can be related to $  \|\nabla_v\mathcal{L}_i(v_{i,m};w_S)\|^2  $:
	\begin{equation}\label{eq:tmp21}
		\begin{aligned}
			\|g_i(v_{i,m};w_S)\|^2 
			&\leq 2\|\nabla_v\mathcal{L}_i(v_{i,m};w_S)\|^2 
			+ 2\eta_{\rm inner}^2L^2\|\nabla_v\mathcal{L}_i(v_{i,m};w_S)\|^2 \\
			&= 2\bigl(1+\eta_{\rm inner}^2L^2\bigr)\|\nabla_v\mathcal{L}_i(v_{i,m};w_S)\|^2,
		\end{aligned}
	\end{equation}
	where we have decomposed $  g_i(v_{i,m};w_S)  $ with the triangle inequality, and used the bound
	\begin{equation}\label{eq:tmp22}
		\|g_i(v_{i,m};w_S)-\nabla_v\mathcal{L}_i(v_{i,m};w_S)\|^2 
		\leq \eta_{\rm inner}^2L^2\|\nabla_v\mathcal{L}_i(v_{i,m};w_S)\|^2.
	\end{equation}
	From \eqref{eq:tmp23}\eqref{eq:tmp21}, we now turn to the evaluation of $  \|\nabla_v \mathcal{L}_i(v_{i,m}; w_S)\|^2  $, which we aim to relate to the global meta-objective value via the algorithm’s progress. Specifically, we have the following by the standard descent inequality for a $L$-smooth function:
	\begin{equation}\label{eq:tmp100}
		\mathbb{E}\left\|\nabla_v \mathcal{L}_i\left(v_{i, m} ; w_S\right)\right\|^2 \leq 2 L \mathbb{E}\left[\mathcal{L}_i\left(v_{i, m} ; w_S\right)-\mathcal{L}_i^*\right]
	\end{equation}
	We can bound $\mathbb{E}\|\nabla_v\mathcal{L}_i(v_{i,m};w_S)\|^2$ by starting via the $L$-smoothness of $  \mathcal{L}_i(\,\cdot\,;w_S)  $ about the global messenger at the beginning of the round:
	\begin{equation}\label{eq:f}
		\begin{aligned}
			\mathbb{E}\bigl[\mathcal{L}_i(v_{i,m};w_S)\bigr] 
			&\leq \mathcal{L}_i(v^{(t)};w_S)
			+\bigl\langle\nabla_v\mathcal{L}_i(v^{(t)};w_S),\mathbb{E}[v_{i,m}-v^{(t)}]\bigr\rangle
			+\frac\beta2\,\mathbb{E}\|v_{i,m}-v^{(t)}\|^2.
		\end{aligned}
	\end{equation}
	Then the inner-product term can be bounded simply by Cauchy–Schwarz and Young’s inequalities subsequently:
	\begin{equation}\label{eq:tmp24}
		\begin{aligned}
			&\quad\quad\quad\bigl\langle\nabla_v\mathcal{L}_i(v^{(t)};w_S),\mathbb{E}[v_{i,m}-v^{(t)}]\bigr\rangle
			\leq\bigl\|\nabla_v\mathcal{L}_i(v^{(t)};w_S)\bigr\|\,\bigl\|\mathbb{E}[v_{i,m}-v^{(t)}]\bigr\|\\
			&\leq\frac1{2L}\bigl\|\nabla_v\mathcal{L}_i(v^{(t)};w_S)\bigr\|^2
			+\frac L2\bigl\|\mathbb{E}[v_{i,m}-v^{(t)}]\bigr\|^2
			\leq\frac1{2L}\bigl\|\nabla_v\mathcal{L}_i(v^{(t)};w_S)\bigr\|^2
			+\frac L2\,\mathbb{E}\|v_{i,m}-v^{(t)}\|^2,
		\end{aligned}
	\end{equation}
	where the last step uses Jensen’s inequality. Substituting \eqref{eq:tmp24} into \eqref{eq:f} and subtracting the minimum value $  \mathcal{L}_i^*  $ yields
	\begin{equation}
		\mathbb{E}\bigl[\mathcal{L}_i(v_{i,m};w_S)-\mathcal{L}_i^*\bigr]
		\leq\mathcal{L}_i(v^{(t)};w_S)-\mathcal{L}_i^*
		+\frac1{2L}\bigl\|\nabla_v\mathcal{L}_i(v^{(t)};w_S)\bigr\|^2
		+L\,\mathbb{E}\|v_{i,m}-v^{(t)}\|^2.
	\end{equation}
	Next we bound the gradient term at the global messenger by client heterogeneity. By Assumption A3, we have 
	\begin{equation}
		\begin{aligned}
			\mathbb{E}\|g_i(v_{i,m};w_S)-\nabla \Phi(v)\|^2 \leq G^2 + B^2 \|\nabla\Phi(v)\|^2,
		\end{aligned}
	\end{equation}
	Decomposing the squared norm therefore gives
	\begin{equation}
		\begin{aligned}
			&\mathbb{E}\bigl\|\nabla_v\mathcal{L}_i(v^{(t)};w_S)\bigr\|^2
			=\mathbb{E}\bigl\|\nabla_v\mathcal{L}_i(v^{(t)};w_S)-g_i(v_{i,m};w_S)+g_i(v_{i,m};w_S)-\nabla\Phi(v^{(t)})+\nabla\Phi(v^{(t)})\bigr\|^2\\
			&\leq 3\mathbb{E}\bigl\|\nabla_v\mathcal{L}_i(v^{(t)};w_S)-g_i(v_{i,m};w_S)\bigr\|^2+3\mathbb{E}\bigl\|g_i(v_{i,m};w_S)-\nabla\Phi(v^{(t)})\bigr\|+3\mathbb{E}\bigl\|\nabla\Phi(v^{(t)})\bigr\|^2\\
			&\leq 3L^2 \eta_{\text {inner }}^2\mathbb{E}\|\nabla_v\mathcal{L}_i(v^{(t)};w_S)\|^2+3 G^2 +(3+3 B^2) \mathbb{E}\|\nabla \Phi(v^{(t)})\|^2.
		\end{aligned}
	\end{equation}
	By letting $\widetilde{\mathcal{C}}=\nicefrac{3(1+B^2 )}{(1-3 L^2 \eta_{\text {inner }}^2)}$ and $\widetilde{\mathcal{L}}_y=\nicefrac{3 G^2}{(1-3 L^2 \eta_{\text {inner }}^2)}$, this will result in 
	\begin{equation}
		\mathbb{E}\|\nabla_v L_i(v^{(t)} ; w_S)\|^2 \leq \widetilde{\mathcal{C}} \mathbb{E}\|\nabla \Phi(v^{(t)})\|^2+\widetilde{\mathcal{L}}_y.
	\end{equation}
	We now bound the client gap $  \mathcal{L}_i(v^{(t)};w_S)-\mathcal{L}_i^*  $ that appears in \eqref{eq:q}. Define the one-step point from the global messenger by $z_i^{(t)}=v^{(t)}-\eta\nabla_v\mathcal{L}_i(v^{(t)};w_S).$ 	By $L$-smoothness we have
	\begin{equation}\label{eq:tmp107}
		\begin{aligned}
			\mathcal{L}_i&(v^{(t)};w_S)
			\leq\mathcal{L}_i(z_i^{(t)};w_S)
			+\bigl\langle\nabla_v\mathcal{L}_i(z_i^{(t)};w_S),v^{(t)}-z_i^{(t)}\bigr\rangle
			+\frac L2\|v^{(t)}-z_i^{(t)}\|^2\\
			&\leq\mathcal{L}_i(z_i^{(t)};w_S)
			+\eta\bigl\langle\nabla_v\mathcal{L}_i(z_i^{(t)};w_S),\nabla_v\mathcal{L}_i(v^{(t)};w_S)\bigr\rangle
			+\frac{L\eta^2}2\bigl\|\nabla_v\mathcal{L}_i(v^{(t)};w_S)\bigr\|^2\\
			&\leq\mathcal{L}_i(z_i^{(t)};w_S)
			+\eta\|\nabla_v\mathcal{L}_i(z_i^{(t)};w_S)\|\,\|\nabla_v\mathcal{L}_i(v^{(t)};w_S)\|
			+\frac{L\eta^2}2\bigl\|\nabla_v\mathcal{L}_i(v^{(t)};w_S)\bigr\|^2\\
			&\leq\mathcal{L}_i(z_i^{(t)};w_S)
			+\eta\Bigl(1+\frac{3L\eta}2\Bigr)\bigl\|\nabla_v\mathcal{L}_i(v^{(t)};w_S)\bigr\|^2\\
			&\leq\mathcal{L}_i(z_i^{(t)};w_S)
			+\eta\Bigl(1+\frac{3L\eta}2\Bigr)\widetilde{\mathcal{C}} \mathbb{E}\|\nabla \Phi(v^{(t)})\|^2+\eta\Bigl(1+\frac{3L\eta}2\Bigr)\widetilde{\mathcal{L}}_y.
		\end{aligned}
	\end{equation}	
	Because $\Phi(v^*)=\nicefrac1N\sum_{j=1}^N\mathcal{L}_j(z_j(v^*);w_S)\geq\nicefrac1N\sum_{j=1}^N\mathcal{L}_j^*  $,  we write the client gap at the one-step point $  z_i^{(t)}  $ as
	\begin{equation}
		\begin{aligned}
			\mathcal{L}_i\bigl(z_i^{(t)};w_S\bigr)-\mathcal{L}_i^*
			&=N\Phi(v^{(t)})
			-\sum_{j\neq i}\mathcal{L}_j\bigl(z_j^{(t)};w_S\bigr)
			-\mathcal{L}_i^*-\bigl(N\Phi(v^{(t)})-\mathcal{L}_i\bigl(z_i^{(t)};w_S\bigr)\bigr)\\
			&\leq N\Phi(v^{(t)})
			-\sum_{j\neq i}\mathcal{L}_j\bigl(z_j^{(t)};w_S\bigr)
			-\mathcal{L}_i^*
		\end{aligned}
	\end{equation}		
	Note that each remaining client loss is bounded from below by its own minimum $\mathcal{L}_j\bigl(z_j^{(t)};w_S\bigr)\ge\mathcal{L}_j^*\; (j\neq i).$ Therefore $-\sum_{j\neq i}\mathcal{L}_j\bigl(z_j^{(t)};w_S\bigr)\le-\sum_{j\neq i}\mathcal{L}_j^*,$ and we obtain
	\begin{equation}\label{eq:tmp109}
		\mathcal{L}_i\bigl(z_i^{(t)};w_S\bigr)-\mathcal{L}_i^*
		\le N\Phi(v^{(t)})-\sum\nolimits_{j=1}^N\mathcal{L}_j^*\leq N\bigl(\Phi(v^{(t)})-\Phi^*\bigr), 
	\end{equation}
	where we have adopted the comparison $  \sum_j\mathcal{L}_j^*\le N\Phi(v^*)  $. As for  $\left\|\nabla_v \mathcal{L}_i\left(v^{(t)} ; w_S\right)\right\|^2$, we can evaluate it through		
	\begin{equation}\label{eq:tmp110}
		\begin{aligned}
			\|\nabla_v \mathcal{L}_i(v^{(t)} ; w_S)\|^2 & =\|\nabla_v \mathcal{L}_i(v^{(t)} ; w_S)-\nabla \Phi(v^{(t)})+\nabla \Phi(v^{(t)})\|^2 \\
			& \leq 2\|\nabla_v \mathcal{L}_i(v^{(t)} ; w_S)-\nabla \Phi(v^{(t)})\|^2+2\|\nabla \Phi(v^{(t)})\|^2 \\
			& \leq 2 N \zeta^2+4 \beta(\Phi(v^{(t)})-\Phi^*) .
		\end{aligned}
	\end{equation}		
	By subsituting \eqref{eq:tmp107}\eqref{eq:tmp109}\eqref{eq:tmp110} into \eqref{eq:tmp100}, we can obtain
	\begin{equation}
		\begin{aligned}
			&E_t
			\le\frac1{KN}\sum_{i=1}^N\sum_{m=0}^{K-2}(K-m-1)
			\Biggl[
			D'\Bigl(
			\bigl(2LN+4\beta\bigr)\bigl(\Phi(v^{(t)})-\Phi^*\bigr)
			+2L\eta\Bigl(1+\tfrac{3L\eta}2\Bigr)\widetilde{\mathcal{C}}\,\mathbb{E}\bigl\|\nabla\Phi(v^{(t)})\bigr\|^2\\
			&\quad+2L\eta\Bigl(1+\tfrac{3L\eta}2\Bigr)\widetilde{\mathcal{L}}_y
			+2N\zeta^2
			+2L^2\,\mathbb{E}\|v_{i,m}-v^{(t)}\|^2
			\Bigr)+\bigl(\eta\gamma\epsilon_3+3\eta^2\gamma^2\bigr)(2C_t+2\widetilde{C}_t)
			+\eta^2\sigma^2
			\Biggr]\\
			&\quad\leq \frac{K-1}2\Biggl[
			D'\Bigl(
			\bigl(2LN+4\beta\bigr)\bigl(\Phi(v^{(t)})-\Phi^*\bigr)
			+2L\eta\Bigl(1+\tfrac{3L\eta}2\Bigr)\widetilde{\mathcal{C}}\,\mathbb{E}\bigl\|\nabla\Phi(v^{(t)})\bigr\|^2\\
			&\quad\quad\quad+2L\eta\Bigl(1+\tfrac{3L\eta}2\Bigr)\widetilde{\mathcal{L}}_y+2N\zeta^2
			\Bigr)
			+2L^2 D'\,E_t				+\bigl(\eta\gamma\epsilon_3+3\eta^2\gamma^2\bigr)(2C_t+2\widetilde{C}_t)
			+\eta^2\sigma^2
			\Biggr].
		\end{aligned}
	\end{equation}
	where recall we have defined the positive prefactor
	\begin{equation}
		D'=\eta\bigl(\epsilon_2+\epsilon_1 Q^2\eta^2\bigr)+6\eta^2\bigl(1+\eta_{\rm inner}^2 L^2\bigr),
	\end{equation}
	then move $E_t$ to the left hand side we can obtain
	\begin{equation}
		\begin{aligned}
			E_t\left(1-(K-1) L^2 D^{\prime}\right) \leq & \frac{K-1}{2} D^{\prime}\left(2 L N+4 \beta+4 \beta L \eta\left(1+\frac{3 L \eta}{2}\right) \tilde{\mathcal{C}}\right)\left(\Phi\left(v^{(t)}\right)-\Phi^*\right) \\
			& +\frac{K-1}{2} \cdot 2 L \eta\left(1+\frac{3 L \eta}{2}\right) \widetilde{\mathcal{L}}_y D^{\prime}+(K-1)\left(\eta \gamma \epsilon_3+3 \eta^2 \gamma^2\right) (C_t+\widetilde{C}_t) \\
			& +\frac{K-1}{2}(2 N \zeta^2 D^{\prime}+\eta^2 \sigma^2),
		\end{aligned}
	\end{equation}
	where the positive variables $\mathcal{A},\mathcal{B},\widetilde{\mathcal{D}}$ and $\mathcal{E}$ can be obtained
	\begin{equation}\label{lemma:phi-evo}
		\begin{aligned}
			&\mathcal{A}=2 L N+4 \beta+4 \beta L \eta\left(1+\frac{3 L \eta}{2}\right) \tilde{\mathcal{C}},\quad \mathcal{E}=\eta \gamma \epsilon_3+3 \eta^2 \gamma^2\\
			&\widetilde{\mathcal{D}}=2 N \zeta^2 D^{\prime}+\eta^2 \sigma^2+2 L \eta\left(1+\frac{3 L \eta}{2}\right) \widetilde{\mathcal{L}}_y D^{\prime}. 
		\end{aligned}
	\end{equation}		
\end{proof}	
\begin{lemma}\label{lemma:10}
	Under Assumptions A1-A8 and the step-size restriction $\tau K \eta \leq \min \{\frac{1}{8 \beta}, \frac{\mu}{4(\gamma^2 \beta+\varepsilon^2 \beta+\Psi /(2 \tau \eta K))}\},$		the expected one-round progress of the meta-objective satisfies
	\begin{equation}\label{eq:tmp139}
		\begin{aligned}
			E& [\Phi(v^{(t+1)})-\Phi(v^{(t)})] \leq -(\tau \eta K \mu-\frac{\beta \Psi}{2}-\tau^2 \eta^2 K \beta(\gamma^2+\varepsilon^2))(\Phi\left(v^{(t)}\right)-\Phi^*) \\
			& +(\frac{\tau \eta K \beta^2}{2}+\frac{\tau \eta \varepsilon}{2}+\frac{\beta \Omega}{2}) E_t+(\tau \eta \gamma+\frac{\beta \Lambda}{2})(C_t+\widetilde{C}_t) +\frac{\beta \Upsilon}{2} E\|w_S^{(t)}-w_S^*(v^{(t)})\|^2+\frac{\beta \Gamma}{2} .
		\end{aligned}
	\end{equation}		
	For simplicity but without loss of generality, the higher-order terms can be omitted when $\eta$ is sufficiently small,
	\begin{equation}
		\begin{aligned}
			E[\Phi(v^{(t+1)})-\Phi(v^{(t)})] & \leq-\frac{\tau \eta K \mu}{2}(\Phi(v^{(t)})-\Phi^*)
			+(\frac{\tau \eta K \beta^2}{2} +\frac{\tau \eta \varepsilon}{2}+\frac{\beta \Omega}{2}) E_t\\
			&\quad\quad\quad+(\tau \eta \gamma+\frac{\beta \Lambda}{2})(C_t+\widetilde{C}_t)+ \frac{\beta \Upsilon}{2} E\|w_S^{(t)}-w_S^*(v^{(t)})\|^2+\frac{\beta \Gamma}{2} .
		\end{aligned}
	\end{equation}
\end{lemma}
\begin{proof}
	By taking the Talor expansion on $E[\Phi(v^{(t+1)})]$ at $v^{(t)}$, we combine the $\beta-$smoothness of $\Phi$, it leads to
	\begin{equation}\label{eq:tmp141}
		\begin{aligned}
			E&[\Phi(v^{(t+1)})] \leq \Phi(v^{(t)})+E\langle\nabla \Phi(v^{(t)}), v^{(t+1)}-v^{(t)}\rangle+\frac{\beta}{2} E\|v^{(t+1)}-v^{(t)}\|^2. 
		\end{aligned}
	\end{equation}
	Note exact reconstruction of the client displacements the server update takes the form		
	\begin{equation}
		v^{(t+1)}=v^{(t)}-\tau \eta \cdot \frac{1}{\left|S_t\right|} \sum_{i \in S_t} \sum_{k=0}^{K-1}\left[\hat{g}_{i, k}\left(v_{i, k}^{(t)}, w_S^{(t)}\right)+\gamma\left(c^{(t)}-c_i^{(t)}\right)+\varepsilon\left(v_{i, k}^{(t)}-v^{(t)}\right)\right], 
	\end{equation}
	which we substitute into the cross term in \eqref{eq:tmp141} and it results in
	\begin{equation}\label{eq:tmp142}
		\begin{aligned}
			&E\langle\nabla \Phi(v^{(t)}), v^{(t+1)}-v^{(t)}\rangle=\underbrace{-\tau \eta \cdot \frac{1}{N} \sum_{i=1}^N \sum_{k=0}^{K-1} E\left\langle\nabla \Phi\left(v^{(t)}\right), \hat{g}_{i, k}\left(v_{i, k}^{(t)}, w_S^{(t)}\right)\right\rangle}_{:=\mathcal{T}_1}\\
			&+\underbrace{-\tau \eta \gamma \cdot \frac{1}{N} \sum_{i=1}^N \sum_{k=0}^{K-1} E\left\langle\nabla \Phi\left(v^{(t)}\right), c^{(t)}-c_i^{(t)}\right\rangle}_{:=\mathcal{T}_2}+\underbrace{-\tau \eta \varepsilon \cdot \frac{1}{N} \sum_{i=1}^N \sum_{k=0}^{K-1} E\left\langle\nabla \Phi\left(v^{(t)}\right), v_{i, k}^{(t)}-v^{(t)}\right\rangle}_{:=\mathcal{T}_3}
		\end{aligned}
	\end{equation}
	For $\mathcal{T}_{1}$,	we decompose each practical meta-gradient
	\begin{equation}
		\begin{aligned}
			\mathcal{T}_1=-\tau \eta  \frac{1}{N} \sum_{i, k} E[\langle\nabla \Phi(v^{(t)}), {g}_{i, k}(v^{(t)}, w_S^{(t)})\rangle+\langle\nabla \Phi(v^{(t)}), {g}_{i, k}(v_{i, k}^{(t)}, w_S^{(t)})-{g}_{i, k}(v^{(t)}, w_S^{(t)})\rangle] .
		\end{aligned}
	\end{equation}
	{ The first average exactly equals } $\|\nabla \Phi(v^{(t)})\|^2$, and the second average can be evaluated via the joint $\beta$-smoothness, namely $\|{g}_{i, k}(v_{i, k}^{(t)}, w_S^{(t)})-{g}_{i, k}(v^{(t)}, w_S^{(t)})\| \leq \beta\|v_{i, k}^{(t)}-v^{(t)}\|$, hence we have
	\begin{equation}\label{eq:T1}
		\begin{aligned}
			&\mathcal{T}_1 \leq-\tau \eta K\left\|\nabla \Phi\left(v^{(t)}\right)\right\|^2+\tau \eta \cdot \frac{1}{N} \sum_{i, k}\left\|\nabla \Phi\left(v^{(t)}\right)\right\| \cdot \beta\left\|v_{i, k}^{(t)}-v^{(t)}\right\| \\
			&\leq-\tau \eta K\left\|\nabla \Phi\left(v^{(t)}\right)\right\|^2+\frac{\tau \eta K}{2}\left\|\nabla \Phi\left(v^{(t)}\right)\right\|^2+\frac{\tau \eta \beta^2}{2} \cdot \frac{1}{N} \sum_{i, k} E\left\|v_{i, k}^{(t)}-v^{(t)}\right\|^2 \\
			&=-\frac{\tau \eta K}{2}\left\|\nabla \Phi\left(v^{(t)}\right)\right\|^2+\frac{\tau \eta K \beta^2}{2} E_t\leq-\tau \eta K \mu\left(\Phi\left(v^{(t)}\right)-\Phi^*\right)+\frac{\tau \eta K \beta^2}{2} E_t,
		\end{aligned}
	\end{equation}
	where we have used the Polyak-Łojasiewicz inequality $\|\nabla \Phi(v^{(t)})\|^2 \geq 2 \mu(\Phi(v^{(t)})-\Phi^*)$. For $\mathcal{T}_{2}$,	we first apply Young’s inequality
	\begin{equation}
		\begin{aligned}
			\left|E\left\langle-\tau \eta \gamma \cdot \nabla \Phi\left(v^{(t)}\right), c^{(t)}-c_i^{(t)}\right\rangle\right| \leq \frac{(\tau \eta \gamma)^2}{2}\left\|\nabla \Phi\left(v^{(t)}\right)\right\|^2+\frac{1}{2} E\left\|c^{(t)}-c_i^{(t)}\right\|^2,
		\end{aligned}
	\end{equation}
	and we take the averaging over the $N$ clients and the $K$ local steps and use the bound $E\|c^{(t)}-c_i^{(t)}\|^2 \leq 2(C_t+\widetilde{C}_t)$ leads to
	\begin{equation}\label{eq:T2}
		\begin{aligned}
			\mathcal{T}_2 &\leq \frac{(\tau \eta \gamma)^2 K}{2}\left\|\nabla \Phi\left(v^{(t)}\right)\right\|^2+\tau \eta \gamma\left(C_t+\widetilde{C}_t\right)\leq\tau^2 \eta^2 \gamma^2 K \beta\left(\Phi\left(v^{(t)}\right)-\Phi^*\right)+\tau \eta \gamma\left(C_t+\widetilde{C}_t\right),
		\end{aligned}
	\end{equation}
	where we have used the smoothness result $\|\nabla \Phi(v^{(t)})\|^2 \leq 2 \beta(\Phi(v^{(t)})-\Phi^*)$. For $\mathcal{T}_3$ we apply the identical rewriting and Young's inequality:
	\begin{equation}
		\left|E\left\langle-\tau \eta \varepsilon \cdot \nabla \Phi\left(v^{(t)}\right), v_{i, k}^{(t)}-v^{(t)}\right\rangle\right| \leq \frac{(\tau \eta \varepsilon)^2}{2}\left\|\nabla \Phi\left(v^{(t)}\right)\right\|^2+\frac{1}{2} E\left\|v_{i, k}^{(t)}-v^{(t)}\right\|^2.
	\end{equation}		
	Subsequently, we take the averaging over the $N$ clients and the $K$ local steps, then using the smoothness result, it leads to
	\begin{equation}\label{eq:T3}
		\mathcal{T}_3 \leq \frac{(\tau \eta \varepsilon)^2 K}{2}\left\|\nabla \Phi\left(v^{(t)}\right)\right\|^2+\frac{\tau \eta \varepsilon}{2} E_t\leq  \tau^2 \eta^2 \varepsilon^2 K \beta\left(\Phi\left(v^{(t)}\right)-\Phi^*\right)+\frac{\tau \eta \varepsilon}{2} E_t .
	\end{equation}	
	Therefore, by substituting the Lemma \ref{lemma:variance2}, $\mathcal{T}_1$ in \eqref{eq:T1}, $\mathcal{T}_2$ in \eqref{eq:T2} and $\mathcal{T}_3$ in \eqref{eq:T3} into \eqref{eq:tmp141}, we can obtain the Lemma result (\ref{eq:tmp139}).
\end{proof}
\begin{lemma}[Lyapunov Evolution]
	Consider the bi-level CIDERS algorithm under Assumptions A1–A8. We design the Lyapunov function as follows:
	\begin{equation}\label{eq:Lyap-CIDERS}
		\begin{aligned}
			V^{(t)}
			&=(1-4a^2)\,\mathbb{E}\|z^{(t)}-v^*\|^2
			+\omega_1 E_t
			+Q_{w2}\eta C_t
			+Q_{w3}\eta\widetilde{C}_t
			+Q_{w4}\eta J_t,
		\end{aligned}	 		
	\end{equation}
	where $J_t=\mathbb{E}\bigl[\Phi(v^{(t)})-\Phi^*\bigr]+\mathbb{E}\|w_S^{(t)}-w_S^*(v^{(t)})\|^2$ is the convergence quantity of interest,  	$  z^{(t)}  $ is the virtual messenger sequence defined by
	\begin{equation}
		z^{(t)}=v^{(t)}+\frac{1-\tau}{\tau}\bigl(v^{(t)}-v^{(t-1)}\bigr),
	\end{equation}
	and $Q_{w1},Q_{w2},Q_{w3},Q_{w4}$ are the positive weights that satisfies the following conditions
	\begin{equation}
		\begin{aligned}
			Q_{w 3}<\frac{(1-2 a) K}{\beta \kappa_1^g}, \quad Q_{w 4}\left(1-\rho_R\right)>6 \beta^2 Q_{w 2}\left[(1-p \alpha)\left(1+\frac{1}{\varepsilon}\right)+2 p \alpha\right]. 
		\end{aligned}
	\end{equation}	 	
	Then the one-round progress of the Lyapunov function satisfies
	\begin{equation}\label{eq:V-descent-CIDERS}
		V^{(t+1)}\le V^{(t)}-\delta_{\Phi}\eta\cdot J_t+D_{\mathrm{tot}}^{\prime},
	\end{equation}
	where $D_{\mathrm{tot}}^{\prime}$ is a postive value and $\delta_{\Phi}$ satisfies the following
	\begin{equation}\label{eq:tmp155}
		\begin{aligned}
			\delta_{\Phi}={\rm{min}}\left\{2(1-2 a) K-2 Q_{w 3} \beta \kappa_1^g,\,Q_{w 4}\left(1-\rho_R\right)-6 \beta^2 Q_{w 2}\left[(1-p \alpha)\left(1+\frac{1}{\varepsilon}\right)+2 p \alpha\right]\right\}.
		\end{aligned}
	\end{equation}
\end{lemma}
\begin{proof}
	Our main target is to evaluate each term's evolution. We first expand the Lyapunov difference as
	\begin{equation}\label{eq:V-expand}
		\begin{aligned}
			V^{(t+1)}-V^{(t)}
			&=\underbrace{(1-4a^2)\bigl(\mathbb{E}\|z^{(t+1)}-v^*\|^2-\mathbb{E}\|z^{(t)}-v^*\|^2\bigr)}_{\mathcal{V}_1}
			+\underbrace{\omega_1(E_{t+1}-E_t)}_{\mathcal{V}_2}\\
			&\quad+\underbrace{\omega_2(C_{t+1}-C_t)}_{\mathcal{V}_3}
			+\underbrace{\omega_3(\widetilde{C}_{t+1}-\widetilde{C}_t)}_{\mathcal{V}_4}
			+\underbrace{\omega_4(J_{t+1}-J_t)}_{\mathcal{V}_5}.
		\end{aligned}
	\end{equation}
	It can be seen the evaluation of the  Lyapunnov function $V$ evolution depends on  each term's evolution, hence we derive them term by term. Notice we have derived the corresponding evolution $E_t$, $C_t$, $\widetilde{C}_t$ and $\Phi$, hence we first foucs on $\mathcal{V}_1$.  Specifically, the server update $  v^{(t+1)}=v^{(t)}-\tau K\eta\,\bar\Delta^{(t)}  $ implies $z^{(t+1)}-z^{(t)}=-K \eta \bar{\Delta}^{(t)}$, where we define $\bar{\Delta}^{(t)}$ and recall $\Delta_{i, t}$
	\begin{equation}\label{eq:tmp153}
		\bar{\Delta}^{(t)}=\frac{1}{\left|S_t\right|} \sum_{i \in S_t} \Delta_{i, t},\quad \Delta_{i, t}=\frac{1}{K} \sum_{k=0}^{K-1} \widehat{g}_{i, k}\left(v_{i, k}^{(t)}, w_S^{(t)}\right)+\gamma\left(c^{(t)}-c_i^{(t)}\right)+r_{\mathrm{prox}, \mathrm{i}}^{(t)} 
	\end{equation}
	which we have used for the expansion of $\mathbb{E}\left\|z^{(t+1)}-v^*\right\|^2$ and it leads to 
	\begin{equation}\label{eq:z-expansion}
		\begin{aligned}
			\mathbb{E}\|z^{(t+1)}-v^*\|^2
			&=\mathbb{E}\|z^{(t)}-v^*-K\eta\,\bar\Delta^{(t)}\|^2\\
			&=\mathbb{E}\|z^{(t)}-v^*\|^2
			-2K\eta\,\mathbb{E}\langle z^{(t)}-v^*,\bar\Delta^{(t)}\rangle
			+K^2\eta^2\mathbb{E}\|\bar\Delta^{(t)}\|^2.
		\end{aligned}	 			
	\end{equation}
	For the inner product term $\mathbb{E}\langle z^{(t)}-v^*, \bar{\Delta}^{(t)}\rangle$, we decompose it as follows
	\begin{equation}\label{eq:inner-split}
		\begin{aligned}
			&\mathbb{E}\langle z^{(t)}-v^*,\bar\Delta^{(t)}\rangle
			=\mathbb{E}\langle v^{(t)}-v^*,\bar\Delta^{(t)}\rangle
			+\frac{1-\tau}{\tau}\mathbb{E}\langle v^{(t)}-v^{(t-1)},\bar\Delta^{(t)}\rangle\\
			&= \mathbb{E}\left\langle v^{(t)}-v^*, \nabla \Phi\left(v^{(t)}\right)\right\rangle+\frac{1}{\eta K} \mathbb{E}\left\langle v^{(t)}-v^*, b^{(t)}\right\rangle+\frac{1-\tau}{\tau}\mathbb{E}\langle v^{(t)}-v^{(t-1)},\bar\Delta^{(t)}\rangle,
		\end{aligned}
	\end{equation} 
	where we have rewritten the following for the decomposition in the first equality
	\begin{equation}
		z^{(t)}-v^*=\left(v^{(t)}-v^*\right)+\frac{1-\tau}{\tau}\left(v^{(t)}-v^{(t-1)}\right)
	\end{equation} 
	and for the second equality we have used
	\begin{equation}
		\mathbb{E}\left[\bar{\Delta}^{(t)} \mid v^{(t)}\right]=\nabla \Phi\left(v^{(t)}\right)+\frac{1}{\eta K} b^{(t)}, 
	\end{equation}
	moreover, we have incorporated the bias that measures the averaged local  direction deviation with the true meta-gradient
	\begin{equation}\label{eq:bias-def}
		\begin{aligned}
			b^{(t)}
			&:=\eta K\,\mathbb{E}\bigl[\bar\Delta^{(t)}\bigm|v^{(t)}\bigr]
			-\eta K\nabla\Phi(v^{(t)}),
		\end{aligned}	 			
	\end{equation}
	and it can be expanded via
	\begin{equation}\label{eq:bias-decomp}
		\begin{aligned}
			b^{(t)}
			&=\underbrace{\frac\eta N\sum_{i=1}^N\sum_{k=0}^{K-1}
				\mathbb{E}\bigl[\widehat g_{i,k}(v_{i,k}^{(t)},w_S^{(t)})-\widetilde g_{i,k}^{\rm ideal}(v^{(t)},w_S^*(v^{(t)}))\bigm|v^{(t)}\bigr]
			}_{\mathcal{T}_I:\; \text{local-to-global meta-gradient deviation}}\\
			&\quad\quad\quad\quad\quad\quad\quad+\underbrace{\eta K\gamma\bigl(c^{(t)}-\bar c^{(t)}\bigr)
			}_{\mathcal{T}_{II}:\;\text{consensus-variate deviation}}
			+\underbrace{\eta K\cdot\frac1N\sum_{i=1}^N\mathbb{E}\bigl[r_{\rm prox,i}^{(t)}\bigm|v^{(t)}\bigr]
			}_{\mathcal{T}_{III}:\;\text{proximal residual}},						
		\end{aligned}
	\end{equation}
	where $  \bar c^{(t)}=\nicefrac1N\sum_{i=1}^N c_i^{(t)}  $. Recall the joint smoothness of each client loss (Assumption A1) the map
	$(v,w_S)\mapsto\nabla_v\mathcal{L}_i\bigl(v-\eta_{\rm inner}\nabla_v\mathcal{L}_i(v,w_S),w_S\bigr)$. Consequently
	\begin{equation}
		\begin{aligned}
			&\mathbb{E}\bigl\|\widehat g_{i,k}(v_{i,k}^{(t)},w_S^{(t)})-\widetilde g_{i,k}^{\rm ideal}(v^{(t)},w_S^*(v^{(t)}))\bigr\|^2\\
			&\qquad\le2\beta^2\mathbb{E}\|v_{i,k}^{(t)}-v^{(t)}\|^2
			+2\beta^2\mathbb{E}\|w_S^{(t)}-w_S^*(v^{(t)})\|^2.
		\end{aligned}
	\end{equation}
	Averaging over the $  K  $ local steps and the $  N  $ clients, and using the definition of client drift $  E_t  $, it produces
	\begin{equation}\label{eq:termI}
		\begin{aligned}
			&\mathbb{E}\|\mathcal{T}_I\|^2=\eta^2\mathbb{E}\Biggl\|\frac1N\sum_{i=1}^N\frac1K\sum_{k=0}^{K-1}
			\bigl(\widehat g_{i,k}(v_{i,k}^{(t)},w_S^{(t)})-\widetilde g_{i,k}^{\rm ideal}(v^{(t)},w_S^*(v^{(t)}))\bigr)\Biggr\|^2\\
			&\qquad\qquad\qquad\qquad\qquad\qquad\qquad\le2\eta^2\beta^2 E_t+2\eta^2\beta^2 \mathbb{E}\|w_S^{(t)}-w_S^*(v^{(t)})\|,
		\end{aligned}
	\end{equation}
	For the averaged consensus-variate deviation term $\mathcal{T}_{II}$, it can be simply derived $\mathbb{E}\|c^{(t)}-\bar{c}^{(t)}\|^2 \leq 2(C_t+\widetilde{C}_t)$, which subsequently leads to
	\begin{equation}\label{eq:termII}
		\mathbb{E}\|\mathcal{T}_{\mathrm{II}}\|^2=\eta^2 K^2 \gamma^2 \mathbb{E}\left\|c^{(t)}-\bar{c}^{(t)}\right\|^2 \leq 2 \eta^2 K^2 \gamma^2\left(C_t+\widetilde{C}_t\right). 
	\end{equation}
	For the proximal residual $\mathcal{T}_{III}$, we can apply the assumption $\mathbb{E}\|r_{\text {prox, } \mathrm{i}}^{(t)}\|^2 \leq \varepsilon_2^2$, which leads to
	\begin{equation}\label{eq:termIII}
		\mathbb{E}\|\mathcal{T}_{III}\|^2
		\le\eta^2 K^2\varepsilon_2^2.
	\end{equation}
	Combining the inequalities \eqref{eq:termI}–\eqref{eq:termIII}  gives the full bound
	\begin{equation}\label{eq:bias-final}
		\begin{aligned}
			\mathbb{E}\|b^{(t)}\|^2
			&\le3\bigl(2\eta^2\beta^2 E_t+2\eta^2 \beta^2 \mathbb{E}\|w_S^{(t)}-w_S^*(v^{(t)})\|^2\bigr)
			+3\bigl(2\eta^2 K^2\gamma^2(C_t+\widetilde{C}_t)\bigr)
			+3\eta^2 K^2\varepsilon_2^2\\
			&=6\eta^2\beta^2 E_t
			+6\eta^2 K^2\gamma^2(C_t+\widetilde{C}_t)
			+6\eta^2 \beta^2 \mathbb{E}\|w_S^{(t)}-w_S^*(v^{(t)})\|^2+3\eta^2 K^2\varepsilon_2^2.
		\end{aligned}
	\end{equation}
	Hence, the third term in \eqref{eq:inner-split} can be evaluated via
	\begin{equation}\label{eq:tmp165}
		\begin{aligned}
			&\mathbb{E}\left\langle v^{(t)}-v^*, b^{(t)}\right\rangle \leq \frac{a}{2} \mathbb{E}\left\|v^{(t)}-v^*\right\|^2+\frac{1}{2 a} \mathbb{E}\left\|b^{(t)}\right\|^2\leq \frac{a}{2} \mathbb{E}\left\|v^{(t)}-v^*\right\|^2+ 6\eta^2\beta^2 E_t
			\\
			&\qquad+\frac{3 \eta^2 \beta^2}{a} E_t+\frac{3 \eta^2 K^2 \gamma^2}{a}\left(C_t+\widetilde{C}_t\right)+\frac{3 \eta^2 \beta^2}{a} \mathbb{E}\left\|w_S^{(t)}-w_S^*\left(v^{(t)}\right)\right\|^2+\frac{3 \eta^2 K^2 \varepsilon_2^2}{2 a} .
		\end{aligned}
	\end{equation}

	For  the term $\mathbb{E}\langle v^{(t)}-v^{(t-1)},\bar\Delta^{(t)}\rangle$ in \eqref{eq:inner-split}, assumption A6  together with the definition of the consensus variates and the proximal residual gives the uniform bound
	\begin{equation}
		\|\widehat g_{i,k}\|\le G,
		\quad
		\|c^{(t-1)}-c_i^{(t-1)}\|\le 2G,
		\quad
		\|r_{\rm prox,i}^{(t-1)}\|\le\varepsilon_2\le G,
	\end{equation}
	where without loss of generality the residual bound can be absorbed into the same constant $  G  $). Consequently, according to the definitions of $\Delta_{i, t}$ in \eqref{eq:tmp153}, we simply obtain
	\begin{equation}
		\left\|\Delta_{i, t-1}\right\| \leq G+\gamma \cdot 2 G+G=G(1+2 \gamma), \quad \mathbb{E}\|\bar{\Delta}^{(t-1)}\|^2 \leq G^2(1+2 \gamma)^2. 
	\end{equation}
	This leads to third second term in \eqref{eq:inner-split} to become
	\begin{equation}\label{eq:tmp168}
		\left|\frac{1-\tau}{\tau} \mathbb{E}\left\langle v^{(t)}-v^{(t-1)}, \bar{\Delta}^{(t)}\right\rangle\right| \leq(1-\tau) K \eta G^2(1+2 \gamma)^2 .
	\end{equation}
	Substituting \eqref{eq:tmp165}\eqref{eq:tmp168} and the  convexity of $  \Phi  $ that
	$\mathbb{E}\langle v^{(t)}-v^*,\nabla\Phi(v^{(t)})\rangle
	\ge\mathbb{E}\bigl[\Phi(v^{(t)})-\Phi^*\bigr]$ into \eqref{eq:inner-split}, we can obtain 
	\begin{equation}\label{eq:final-V1-inner}
		\begin{aligned}
			-2K\eta\,\mathbb{E}\langle z^{(t)}-v^*,\bar\Delta^{(t)}\rangle
			&\le
			-2K\eta\,\mathbb{E}\bigl[\Phi(v^{(t)})-\Phi^*\bigr]
			+a\,\mathbb{E}\|v^{(t)}-v^*\|^2
			+\frac{6\eta^2\beta^2}{a}E_t
			+\frac{6\eta^2 K^2\gamma^2}{a}(C_t+\widetilde{C}_t)\\
			&\quad+\frac{6\eta^2\beta^2}{a}\mathbb{E}\|w_S^{(t)}-w_S^*(v^{(t)})\|^2
			+\frac{3\eta^2 K^2\varepsilon_2^2}{a}
			+2(1-\tau)K^2\eta^2\,G^2(1+2\gamma)^2.
		\end{aligned}	 			
	\end{equation}
	Notice the term $\mathbb{E}\|v^{(t)}-v^*\|^2$ satisifies the elementary expansion that relates the true messenger distance to the virtual-messenger distance:
	\begin{equation}\label{eq:tmp170}
		\begin{aligned}
			&\mathbb{E}\left\|v^{(t)}-v^*\right\|^2 \leq 2 \mathbb{E}\left\|z^{(t)}-v^*\right\|^2+2\left(\frac{1-\tau}{\tau}\right)^2 \mathbb{E}\left\|v^{(t)}-v^{(t-1)}\right\|^2\\
			&\leq2 \mathbb{E}\left\|z^{(t)}-v^*\right\|^2+2(1-\tau^2) K^2 \eta^2 \mathbb{E}\left\|\bar{\Delta}^{(t-1)}\right\|^2=2 \mathbb{E}\left\|z^{(t)}-v^*\right\|^2+2(1-\tau)^2 K^2 \eta^2 G^2(1+2 \gamma)^2
		\end{aligned}
	\end{equation}
	
	For the last term in \eqref{eq:z-expansion}, the Server Messenger Update lemma together with the identity $  K^2\eta^2\mathbb{E}\|\bar\Delta^{(t)}\|^2=\frac1{\tau^2}\mathbb{E}\|v^{(t+1)}-v^{(t)}\|^2  $ yields
	\begin{equation}\label{eq:quad-bound}
		\begin{aligned}
			K^2\eta^2\mathbb{E}\|\bar\Delta^{(t)}\|^2
			&\le\frac{\Lambda}{\tau^2}(C_t+\widetilde{C}_t)
			+\nicefrac{\Upsilon}{\tau^2}\mathbb{E}\|w_S^{(t)}-w_S^*(v^{(t)})\|^2
			+\frac{\Omega}{\tau^2}E_t
			+\frac{\Psi}{\tau^2}\mathbb{E}\bigl[\Phi(v^{(t)})-\Phi^*\bigr]
			+\frac{\Gamma}{\tau^2},
		\end{aligned}	 			
	\end{equation}	 		
	and multiplying by $(1-4a^2)$ yields $\mathcal{V}_1$
	\begin{equation}\label{eq:V1-start}
		\begin{aligned}
			\mathcal{V}_1
			&=(1-4a^2)\Bigl(
			-2K\eta\,\mathbb{E}\langle z^{(t)}-v^*,\bar\Delta^{(t)}\rangle
			+K^2\eta^2\mathbb{E}\|\bar\Delta^{(t)}\|^2
			\Bigr).
		\end{aligned}	 			
	\end{equation}
	By substituting \eqref{eq:final-V1-inner}\eqref{eq:tmp170}\eqref{eq:quad-bound} into \eqref{eq:z-expansion} we have
	\begin{equation}\label{eq:correct-intermediate}
		\begin{aligned}
			&\mathbb{E}\|z^{(t+1)}-v^*\|^2
			-(1+2a)\mathbb{E}\|z^{(t)}-v^*\|^2
			\le\Biggl(-2K\eta+\frac{\Psi}{\tau^2}\Biggr)\mathbb{E}\bigl[\Phi(v^{(t)})-\Phi^*\bigr]
			\\
			&\quad+\Biggl(\frac{6\eta^2\beta^2}{a}+\frac{\Omega}{\tau^2}\Biggr)E_t+\Biggl(\frac{6\eta^2 K^2\gamma^2}{a}+\frac{\Lambda}{\tau^2}\Biggr)(C_t+\widetilde{C}_t)
			+\Biggl(\frac{6\eta^2\beta^2}{a}+\frac{\Upsilon}{\tau^2}\Biggr)\mathbb{E}\|w_S^{(t)}-w_S^*(v^{(t)})\|^2\\
			&\quad+\frac{3\eta^2 K^2\varepsilon_2^2}{a}
			+2(1-\tau)K^2\eta^2\,G^2(1+2\gamma)^2
			+2a(1-\tau)^2 K^2\eta^2\,G^2(1+2\gamma)^2
			+\frac{\Gamma}{\tau^2}.
		\end{aligned}
	\end{equation}
	Since $a<1/2$, we simply have $1-2a>1-4a^2$, this leads to
	\begin{equation}\label{eq:V1-final}
		\begin{aligned}
			&\mathcal{V}_1
			\le(1-2a)\Biggl(-2K\eta+\frac{\Psi}{\tau^2}\Biggr)\mathbb{E}\bigl[\Phi(v^{(t)})-\Phi^*\bigr]
			+(1-2a)\Biggl(\frac{6\eta^2\beta^2}{a}+\frac{\Omega}{\tau^2}\Biggr)E_t\\
			&\quad+(1-2a)\Biggl(\frac{6\eta^2 K^2\gamma^2}{a}+\frac{\Lambda}{\tau^2}\Biggr)(C_t+\widetilde{C}_t)
			+(1-2a)\Biggl(\frac{6\eta^2\beta^2}{a}+\frac{\Upsilon}{\tau^2}\Biggr)\mathbb{E}\|w_S^{(t)}-w_S^*(v^{(t)})\|^2\\
			&\quad+(1-2a)\Biggl(
			\frac{3\eta^2 K^2\varepsilon_2^2}{a}
			+2(1-\tau)K^2\eta^2\,G^2(1+2\gamma)^2
			+2a(1-\tau)^2 K^2\eta^2\,G^2(1+2\gamma)^2
			+\frac{\Gamma}{\tau^2}
			\Biggr),
		\end{aligned}
	\end{equation}
	which is further simplified to
	\begin{equation}
		\begin{aligned}
			\mathcal{V}_1
			&\le-A_z\,\mathbb{E}\bigl[\Phi(v^{(t)})-\Phi^*\bigr]
			+B_z E_t
			+C_z(C_t+\widetilde{C}_t)
			+D_z \mathbb{E}\|w_S^{(t)}-w_S^*(v^{(t)})\|^2
			+D_z',
		\end{aligned}
	\end{equation}
	and we have defined
	\begin{equation}
		\begin{aligned}
			&A_z=(1-2 a)\left(2 K \eta-\frac{\Psi}{\tau^2}\right),\quad B_z=(1-2 a)\left(\frac{6 \eta^2 \beta^2}{a}+\frac{\Omega}{\tau^2}\right),\\
			&C_z=(1-2 a)\left(\frac{6 \eta^2 K^2 \gamma^2}{a}+\frac{\Lambda}{\tau^2}\right),\quad D_z=(1-2 a)\left(\frac{6 \eta^2 \beta^2}{a}+\frac{\Upsilon}{\tau^2}\right)\\
			&D_z^{\prime}=(1-2 a)\left(\frac{3 \eta^2 K^2 \varepsilon_2^2}{a}+2(1-\tau) K^2 \eta^2 G^2(1+2 \gamma)^2+2 a(1-\tau)^2 K^2 \eta^2 G^2(1+2 \gamma)^2+\frac{\Gamma}{\tau^2}\right).
		\end{aligned}
	\end{equation}
	
	For $\mathcal{V}_2$, we can use the client drift Lemma \ref{lemma:clientdrift} for $E_{t+1}-E_t$: 
	\begin{equation}
		\begin{aligned}
			&E_{t+1}-E_t \leq \frac{(K-1) / 2}{1-K L^2 D^{\prime}}\left\{D^{\prime} \mathcal{A} \mathbb{E}\left[\Phi\left(v^{(t+1)}\right)-\Phi^*\right]+2 \mathcal{E}\left(C_{t+1}+\widetilde{C}_{t+1}\right)+\widetilde{\mathcal{D}}\right\}-E_t\\
			&\leq\frac{(K-1) / 2}{1-K L^2 D^{\prime}}D^{\prime} \mathcal{A} \big\{(1-\frac{\tau \eta K \mu}{2}) \mathbb{E}[\Phi(v^{(t)})-\Phi^*]+(\frac{\tau \eta K \beta^2}{2}+\frac{\tau \eta \varepsilon}{2}+\frac{\beta \Omega}{2}) E_t\\
			&\qquad\qquad\qquad\qquad\qquad\qquad\qquad\qquad\qquad+(\tau \eta \gamma+\frac{\beta \Lambda}{2})(C_t+\widetilde{C}_t)+\frac{\beta \Upsilon}{2} R_t+\frac{\beta \Gamma}{2} \big\}\\
			&+\frac{(K-1) / 2}{1-K L^2 D^{\prime}}D^{\prime} \mathcal{A}\cdot 	2 \mathcal{E}\big\{\left(\rho_C+\rho_g\right) C_t+\left(\kappa_{\widetilde{C}}+\tilde{\rho}_g\right) \widetilde{C}_t+\left(\kappa_E+\kappa_2^g \Omega\right) E_t\\
			&+\left(\kappa_w+\kappa_2^g \Upsilon\right) R_t 
			+\left(\kappa_{\Phi}+\widetilde{\kappa}_3^g\right) \mathbb{E}\left[\Phi\left(v^{(t)}\right)-\Phi^*\right]+\left(\kappa_0+\widetilde{\kappa}_g^4\right) \big\}+\frac{(K-1) / 2}{1-K L^2 D^{\prime}}\cdot\widetilde{D}-E_t,
		\end{aligned}
	\end{equation}
	where we have used the Lemma \ref{lemma:phi-evo} for evaluating the term $(\Phi(v^{(t+1)})-\Phi^*)$ and the Lemmas \ref{lemma:localcontrolevo} and \ref{lemma:globalcontroalevo} for evaluating the term $(C_{t+1}+\widetilde{C}_{t+1})$. Moreover, we have defined $R_t=\mathbb{E}\|w_S^{(t)}-w_S^*(v^{(t)})\|^2$ and used the Lemma \ref{lemma:average} for its evolution $\rho_R R_t+D_R$, with the definitions of $\rho_R=2(1-\mu \eta_{\mathrm{KD}})^E$ and $D_R=\nicefrac{\eta_{\mathrm{KD}} \sigma_{\mathrm{KD}}^2}{\mu}+8 R_w^2(1-\mu \eta_{\mathrm{KD}})^E$. For brevity, we let  $\kappa_{kl}=\nicefrac{(K-1) / 2}{(1-K L^2 D^{\prime})}$, then it leads to
	\begin{equation}
		E_{t+1}-E_t \leq A_E \mathbb{E}\left[\Phi\left(v^{(t)}\right)-\Phi^*\right]+\left(\rho_E-1\right) E_t+C_E^C C_t+C_E^{\widetilde{C}} \widetilde{C}_t+D_E R_t+D_E^{\prime},
	\end{equation}
	where the coefficients for $\mathbb{E}[\Phi(v^{(t)})-\Phi^*]$,  $R_t$, $C_t$, $\widetilde{C}_t$ and $E_t$ are provided
	\begin{equation}
		\begin{aligned}
			&A_E=\kappa_{k l} D^{\prime} A\left[\left(1-\frac{\tau \eta K \mu}{2}\right)+2\mathcal{E}\left(\kappa_{\Phi}+\widetilde{\kappa}_3^g\right)\right],\, \rho_E=\kappa_{k l} D^{\prime} A\left(\frac{\tau \eta K \beta^2}{2}+\frac{\tau \eta \varepsilon}{2}+\frac{\beta \Omega}{2}+2 \mathcal{E}\left(\kappa_E+\kappa_2^g \Omega\right)\right)\\
			&D_E=\kappa_{k l}\left[D^{\prime} A\left(\frac{\beta \Upsilon}{2}+2 \mathcal{E}\left(\kappa_w+\kappa_2^g \Upsilon\right)\right)\right],\,C^C_E=\kappa_{k l} D^{\prime} A\left(\tau \eta \gamma+\frac{\beta \Lambda}{2}+2 \mathcal{E}\left(\rho_C+\rho_g\right)\right),\\
			&C^{\widetilde{C}}_E=\kappa_{k l} D^{\prime} A\left(\tau \eta \gamma+\frac{\beta \Lambda}{2}+2 \mathcal{E}\left(\kappa_{\widetilde{C}}+\widetilde{\rho}_g\right)\right), \, D^{\prime}_E=\kappa_{k l}\left[D^{\prime} A\left(\frac{\beta \Gamma}{2}+2 \mathcal{E}\left(\kappa_0+\widetilde{\kappa}_g^4\right)\right)+\widetilde{D}\right].
		\end{aligned}
	\end{equation}
	Next, we continue to derive $\mathcal{V}_3$, which can be directly obtained from Lemma \ref{lemma:localcontrolevo}
	that
	\begin{equation}
		\begin{aligned}
			\mathcal{V}_3
			\le
			\omega_2(\rho_C-1)C_t
			+\omega_2\kappa_{\widetilde{C}}\widetilde{C}_t
			+\omega_2\kappa_E E_t
			+\omega_2\kappa_w R_t
			+\omega_2\kappa_\Phi\mathbb{E}\bigl[\Phi(v^{(t)})-\Phi^*\bigr]
			+\omega_2\kappa_0.
		\end{aligned}
	\end{equation}
	Similarly, $\mathcal{V}_4$ can be directly derived via Lemma \ref{lemma:globalcontroalevo} that
	\begin{equation}
		\begin{aligned}
			\mathcal{V}_4\le
			\omega_3(\widetilde{\rho}_g-1)\widetilde{C}_t
			+\omega_3\rho_g C_t
			+\omega_3\kappa_2^g\Upsilon R_t
			+\omega_3\kappa_2^g\Omega E_t
			+\omega_3\widetilde{\kappa}_3^g\mathbb{E}\bigl[\Phi(v^{(t)})-\Phi^*\bigr]
			+\omega_3\widetilde{\kappa}_g^4.
		\end{aligned}
	\end{equation}
	For $\mathcal{V}_5$, Lemmas \ref{lemma:phi-evo} and \ref{lemma:average} can be adopted and it directly leads to
	\begin{equation}
		\begin{aligned}
			\begin{aligned}
				\mathcal{V}_5= & \omega_4\left(J_{t+1}-J_t\right) \leq \omega_4\left(-\frac{\tau \eta K \mu}{2}\right) \mathbb{E}\left[\Phi\left(v^{(t)}\right)-\Phi^*\right]+\omega_4\left(\frac{\tau \eta K \beta^2}{2}+\frac{\tau \eta \varepsilon}{2}+\frac{\beta \Omega}{2}\right) E_t \\
				& +\omega_4\left(\tau \eta \gamma+\frac{\beta \Lambda}{2}\right)\left(C_t+\widetilde{C}_t\right)+\omega_4\left(\frac{\beta \Upsilon}{2}+\rho_R-1\right) R_t+\omega_4\left(\frac{\beta \Gamma}{2}+D_R\right) .
			\end{aligned}
		\end{aligned}
	\end{equation}
	Combine $(\mathcal{V}_1,\mathcal{V}_2,\mathcal{V}_3,\mathcal{V}_4,\mathcal{V}_5)$, we can have the Lyapunov evolution
	\begin{equation}\label{eq:tmp186}
		V^{(t+1)}-V^{(t)} \leq-A_{\mathrm{tot}} \mathbb{E}\left[\Phi\left(v^{(t)}\right)-\Phi^*\right]-B_{\mathrm{tot}} E_t-C_{\mathrm{tot}}^C C_t-C_{\mathrm{tot}}^{\widetilde{C}} \widetilde{C}_t-D_{\mathrm{tot}} R_t+D_{\mathrm{tot}}^{\prime},
	\end{equation}
	where we have defined $A_{\mathrm{tot}},B_{\mathrm{tot}},C_{\mathrm{tot}}^C,C_{\mathrm{tot}}^{\widetilde{C}},D_{\mathrm{tot}},D_{\mathrm{tot}}^{\prime}$ 
	\begin{equation}\label{eq:tmp187}
		\begin{aligned}
			& A_{\text {tot }}=A_z-\omega_1 A_E-\omega_2 \kappa_{\Phi}-\omega_3 \widetilde{\kappa}_3^g+\omega_4 \cdot \frac{\tau \eta K \mu}{2}\approx\left[2(1-2 a) K-2 Q_{w 3} \beta \kappa_1^g\right] \eta \\
			& B_{\text {tot }}=\omega_1-B_z-\omega_1\rho_E-\omega_2 \kappa_E-\omega_3 \kappa_2^g \Omega-\omega_4\left(\frac{\tau \eta K \beta^2}{2}+\frac{\tau \eta \varepsilon}{2}+\frac{\beta \Omega}{2}\right) \\
			&\qquad\qquad\qquad\qquad\qquad\qquad\qquad \approx \omega_1-2 p \alpha \beta^2 Q_{w2}\eta-Q_{w4} \cdot \frac{\tau}{2}\left(K \beta^2+\varepsilon\right) \eta^2\\
			& C_{\text {tot }}^C=\omega_2-C_z-\omega_1 C_E^C-\omega_2\rho_C-\omega_3 \rho_g-\omega_4\left(\tau \eta \gamma+\frac{\beta \Lambda}{2}\right) \\
			&\qquad\qquad\qquad\qquad\qquad\qquad\qquad \approx \eta Q_{w2}(1-(1-p \alpha)(1+\varepsilon))-Q_{w4} \tau \gamma \eta^2, \\
			& C_{\text {tot }}^{\widetilde{C}}=\omega_3-C_z-\omega_1 C_E^{\widetilde{C}}-\omega_2 \kappa_{\widetilde{C}}-\omega_3\tilde{\rho}_g-\omega_4\left(\tau \eta \gamma+\frac{\beta \Lambda}{2}\right)\approx \omega_3-Q_{w4} \tau \gamma \eta^2, \\
			& D_{\text {tot }}=\omega_4-D_z-\omega_1 D_E-\omega_2 \kappa_w-\omega_3 \kappa_2^g \Upsilon-\omega_4\left(\frac{\beta \Upsilon}{2}+\rho_R\right)\\
			&\qquad\qquad\qquad \approx\left\{Q_{w 4}\left(1-\rho_R\right)-6 \beta^2 Q_{w 2}\left[(1-p \alpha)\left(1+\frac{1}{\varepsilon}\right)+2 p \alpha\right]\right\} \eta\\
			&D_{\mathrm{tot}}^{\prime}=D_z^{\prime}+\omega_1 D_E^{\prime}+\omega_2 \kappa_0+\omega_3 \widetilde{\kappa}_g^4+\omega_4\left(\frac{\beta \Gamma}{2}+D_R\right)\approx C_{\eta}\eta+C_{\eta^2}\eta^2+C_{\eta^3}\eta^3
		\end{aligned}
	\end{equation}
	where especially for $D_{\mathrm{tot}}^{\prime}$, we have defined $C_{\eta},C_{\eta^2},C_{\eta^3}$ as follows: 
	\begin{equation}
		\begin{aligned}
			&C_{\eta}=\left\{\omega_1(K-1)(2 N \zeta^2 \epsilon_2)+\left(Q_{w 2} \kappa_0^{(0)}+Q_{w 3} \kappa_3^g+Q_{w 4} D_R\right)\right\}\\
			&C_{\eta^2}=\omega_1(K-1)\left[2 \epsilon_2 \mathcal{A}_0 \gamma \epsilon_3\left(\kappa_0^{(0)}+\kappa_3^g\right)+12 N \zeta^2\left(1+\eta_{\text {inner }}^2 L^2\right)+\sigma^2+2 L \widetilde{\mathcal{L}}_y \epsilon_2\right]\\
			&\qquad\qquad+(1-2 a)\left(\frac{3 K^2 \varepsilon_2^2}{a}+2(1-\tau) K^2 G^2(1+2 \gamma)^2+2 a(1-\tau)^2 K^2 G^2(1+2 \gamma)^2+\frac{\Gamma_2}{\tau^2}\right)\\
			&C_{\eta^3}=\omega_1(K-1)\left[\epsilon_2 \mathcal{A}_0 \cdot \mathcal{P}_2+\left(\epsilon_2 \mathcal{A}_1+6\left(1+\eta_{\text {inner }}^2 L^2\right) \mathcal{A}_0\right) \cdot 2 \gamma \epsilon_3\left(\kappa_0^{(0)}+\kappa_3^g\right)\right. \\
			& \left.\qquad+2 N \zeta^2 \epsilon_1 Q^2+12 L \widetilde{\mathcal{L}}_y\left(1+\eta_{\text {inner }}^2 L^2\right)+3 L^2 \widetilde{\mathcal{L}}_y \epsilon_2\right]+\left(Q_{w 2} \kappa_0^{(1)}+Q_{w 3} \kappa_2^g+Q_{w 4} \frac{\beta}{2}\right) \Gamma_2,
		\end{aligned}
	\end{equation}
	and we have set the following definitions for brevity:
	\begin{equation}
		\begin{aligned}
			&D^{\prime}  =\eta \epsilon_2+\eta^2 D_2^{\prime}, \, D_2^{\prime}=\epsilon_1 Q^2+6\left(1+\eta_{\text {inner }}^2 L^2\right),\, \mathcal{P}_2=6 \gamma^2\left(\kappa_0^{(0)}+\kappa_3^g\right)+\frac{\beta}{2} \Gamma_2+2 \gamma \epsilon_3\left(\kappa_0^{(2)}+\kappa_2^g\right) \Gamma_2, \\
			&\mathcal{A}  =\mathcal{A}_0+\eta \mathcal{A}_1+\eta^2\mathcal{A}_2, \quad \mathcal{A}_0=2 L N+4 \beta, \quad \mathcal{A}_1=4 \beta L \tilde{\mathcal{C}}, \quad\Gamma  =\eta^2 \Gamma_2\quad{\rm with}\quad \\
			&{{\Gamma_2=\frac{12 N \tau^2 K^2}{S}\left(\frac{6 \sigma^2}{K}+12 G^2+2 \varepsilon_2^2\right)+\left(\frac{12 N \tau^2 K^2 B^2}{S}+2 \tau^2 K^2\right)\left(32 \eta_{\text {inner }}^2 B^2 G^2+\left(40+16 L^2 \eta_{\text {inner }}^2\right) \zeta^2\right)}}\\
			&\kappa_0=\kappa_0^{(0)}+\eta^2 \kappa_0^{(1)} \Gamma_2\quad{\rm with }\, \kappa_0^{(1)}=(1-p \alpha)\left(1+\frac{1}{\varepsilon}\right) \beta^2+2 p \alpha \beta^2\quad{\rm and}\\
			&\kappa_0^{(0)}=(1-p \alpha)\left(1+\frac{1}{\varepsilon}\right) 16 \beta^2 R_w^2+p \alpha\left(\frac{2 \sigma^2}{K}+32 \beta^2 R_w^2\right). 
		\end{aligned}
	\end{equation}
	It can be verified simply from (\ref{eq:tmp187}) that when $\eta$ is sufficiently small, all $(B_{\text {tot }},C_{\text {tot }}^C, C_{\text {tot }}^{\widetilde{C}})>0$. Moreover,  \eqref{eq:tmp155} holds.		 Therefore, \eqref{eq:tmp186} can be simplified to $V^{(t+1)}\le V^{(t)}-\delta_{\Phi}\eta\cdot J_t+D_{\mathrm{tot}}^{\prime}$, which is exactly \eqref{eq:V-descent-CIDERS} in the lemma. 	 		
\end{proof}

\begin{theorem}[Ergodic Convergence]\label{theorem}
	Under Assumptions A1–A8, let the weights of the Lyapunov function $V^{(t)}=(1-4 a^2) \mathbb{E}\|z^{(t)}-v^*\|^2+\omega_1 E_t+Q_{w 2} \eta C_t+Q_{w 3} \eta \widetilde{C}_t+Q_{w 4} \eta J_t$  to be positive  $(\omega_1,Q_{w2},Q_{w3},Q_{w4})>0$,  with $  Q_{w2},Q_{w3},Q_{w4}  $ obeying
	$Q_{w3}<\nicefrac{(1-2a)K}{\beta\kappa_1^g}$, and $Q_{w4}(1-\rho_R)>6\beta^2Q_{w2}[(1-p\alpha)(1+\nicefrac{1}{\varepsilon})+2p\alpha]$. Moreover, the upper-level step-size be chosen as 
	\begin{equation}\label{eq:tmp190}
		\eta=\sqrt{\frac{V_0^{(0)}}{C_{\eta^2} T}},\,\text{where }V_0^{(0)}=\left(1-4 a^2\right) \mathbb{E}\left\|z^{(0)}-v^*\right\|^2+\omega_1 E_0,
	\end{equation}
	Then the ergodic average of the joint gap $J_t=\mathbb{E}[\Phi(v^{(t)})-\Phi^*]+\mathbb{E}\|w_S^{(t)}-w_S^*(v^{(t)})\|^2$ satisfies 
	\begin{equation}\label{eq:tmp191}
		\begin{aligned}
			&\frac{1}{T} \sum_{t=0}^{T-1} J_t \leq \frac{C_\eta}{\delta_{\Phi}}+\frac{2}{\delta_{\Phi}} \sqrt{\frac{V_0^{(0)} C_{\eta^2}}{T}}+\frac{V_1^{(0)}}{\delta_{\Phi} T}+\frac{C_{\eta^3} V_0^{(0)}}{\delta_{\Phi} C_{\eta^2} T}+O\left(T^{-3 / 2}\right)
		\end{aligned}
	\end{equation}
\end{theorem}
\begin{proof}
	Summing the inequality \eqref{eq:V-descent-CIDERS} from $  t=0  $ to $  T-1  $ yields
	\begin{equation}\label{eq:tmp192}
		\frac{1}{T} \sum_{t=0}^{T-1} J_t \leq \frac{V^{(0)}}{\delta_{\Phi} \eta T}+\frac{D_{\mathrm{tot}}^{\prime}}{\delta_{\Phi} \eta}=\frac{V_0^{(0)}}{\delta_{\Phi} \eta T}+\frac{V_1^{(0)}}{\delta_{\Phi} T}+\frac{C_\eta}{\delta_{\Phi}}+\frac{C_{\eta^2}}{\delta_{\Phi}} \eta+\frac{C_{\eta^3}}{\delta_{\Phi}} \eta^2+O\left(\eta^3\right),
	\end{equation}
	where we have defined 
	\begin{equation}
		V_0^{(0)}=\left(1-4 a^2\right) \mathbb{E}\left\|z^{(0)}-v^*\right\|^2+\omega_1 E_0,\quad V_1^{(0)}=Q_{w 2} C_0+Q_{w 3} \widetilde{C}_0+Q_{w 4} J_0. 
	\end{equation}
	We consider the function
	$\mathcal{H}(\eta)=\nicefrac{V_0^{(0)}}{(\delta_\Phi\eta T)}	+\nicefrac{C_{\eta^2}}{\delta_\Phi}$, by minimizing it we can obtain the step size $\eta^*=\sqrt{\nicefrac{V_0^{(0)}}{C_{\eta^2} T}}$ exactly as \eqref{eq:tmp190} and substituting it into \eqref{eq:tmp192} we can obtain the convergence result in \eqref{eq:tmp191}. 	 	
\end{proof}

\subsubsection{Discussion of the Convergence Result}
\label{sec:conv-discuss}

The convergence result in \eqref{eq:tmp191} provides several important insights into the behavior of CIDERS.

\paragraph{Convergence rate.} 
The leading term in \eqref{eq:tmp191} is $\mathcal{O}({T}^{\nicefrac{-1}{2}})$, which matches the standard sublinear rate for stochastic nonconvex optimization under heterogeneous settings \cite{bottou2018optimization}. The first term $C_\eta/\delta_\Phi$ represents the asymptotic neighborhood of convergence, which is assembled from the local gradient heterogeneity,  lower-level gap, and the proximal second-moment.	It also confirms that larger client participation improves the convergence rate by reducing the variance of the aggregated update.  The step-size choice $\eta = \mathcal{O}(1/\sqrt{T})$ is standard for stochastic gradient methods and balances the initial transient and the asymptotic variance \cite{khaled2020tighter}. In Lemma \ref{lemma:10}, the global learning rate $\tau$  satisfies 
\begin{equation}
	\tau K\eta \le \min\left\{\frac{1}{8\beta}, \frac{\mu}{4(\gamma^2\beta + \varepsilon^2\beta + \Psi/(2\tau\eta K))}\right\},
\end{equation}
which is the stability condition for smooth optimization with momentum-like updates. 

\paragraph{Effect of local step $K$.} 
While the local step $K$ can reduce communication, it may inflate client drift. In Lemma \ref{lemma:clientdrift}, the prefactor
$(1-KL^2D')^{-1}$ remains positive  for $K<1/(L^2D')$, with $D'=\mathcal{O}(\eta)$. Beyond this threshold, the Lyapunov descent inequality fails, which justifies the practical guideline that $K\eta$ must be sufficiently small to prevent divergence. This aligns with the client-drift phenomenon identified in prior federated learning analyses \cite{karimireddy2020scaffold, li2020federated}.

\paragraph{Effect of lower-level distillation.} 
The lower-level approximation error $R_t := \mathbb{E}\|w_S^{(t)} - w_S^*(v^{(t)})\|^2$ contracts at rate $\rho_R = 2(1-\mu\eta_{\mathrm{KD}})^E$ from Lemma \ref{lemma:average}. To ensure that this error does not dominate the upper-level progress, the number of KD steps $E$ should satisfy
\begin{equation}
	E \gtrsim {(\mu\eta_{\mathrm{KD}})^{-1}} \log{(\eta)^{-1}},
\end{equation}
This condition formalizes the requirement that the student backbone is sufficiently well aligned with the current messenger before the next round of client updates begins.

\paragraph{Comparison with baseline federated learning.} 
The convergence rate $\mathcal{O}(1+1/\sqrt{T})$ matches the best-known rates for federated learning with client heterogeneity \cite{karimireddy2020scaffold, acar2021federated, reddi2021adaptive}. Compared to standard FL, which suffers from biased convergence due to client drift \cite{li2020federated}, CIDERS retains a comparable asymptotic neighborhood  arising from the bilevel structure, i.e., the lower-level approximation error $R_t$ is the  additional sources of bias, thus the lower-level updates should be carefully controlled. Importantly, the consensus-variate correction effectively mitigates the heterogeneity from the asymptotic bias.  This aligns with the qualitative finding in Section \ref{sec:theoretical_tradeoff} that the consensus variates reduce the heterogeneity gap when ${c} \approx \widetilde G$ and ${c}_i \approx \widetilde g_i$. 


\end{document}